\documentclass[reqno]{amsart}

\usepackage{hyperref}
\usepackage{graphicx}
\usepackage{curves}
\usepackage{amssymb,amsmath,xcolor,color,framed,epstopdf,amsthm}
\usepackage[english]{babel}
\usepackage{enumerate}
\usepackage{mathrsfs}
\usepackage{color}
\usepackage{empheq}
\usepackage{braket}
\usepackage{hyperref}
\hypersetup{colorlinks = true, linkcolor = {blue} }
\usepackage{cancel}
\usepackage{subcaption}
 \usepackage{lipsum}
 \usepackage{mathtools}
 \usepackage{MnSymbol}
\usepackage{xcolor}

\usepackage[left=3.5cm, right=3.5cm, top=3.5cm, bottom=3.5cm]{geometry}

\newcommand{\edit}[1]{{#1}}

\usepackage{tikz} 
\usetikzlibrary{shapes,snakes}
\usepackage{pgf}
\usepackage{pgflibraryarrows}
\usepackage{pgflibrarysnakes}
\usetikzlibrary{decorations.text}
\usepgfmodule{shapes}
\usetikzlibrary{decorations.pathmorphing}
\usetikzlibrary{decorations.markings}
\usetikzlibrary{patterns}
\usetikzlibrary{automata}
\usetikzlibrary{positioning}
\usepackage{pgfplots}
\usetikzlibrary{backgrounds}
\tikzset{partial ellipse/.style args={#1:#2:#3}{insert path={+ (#1:#3) arc (#1:#2:#3)} }}
\tikzset{->-/.style={decoration={ markings, mark=at position #1 with {\arrow{>}}},postaction={decorate}}}

\definecolor{Dgreen}{RGB}{0,153,0}

\newcommand*{\mailto}[1]{\href{mailto:#1}{\nolinkurl{#1}}}

\newtheorem{theorem}{Theorem}[section]
\newtheorem{lemma}[theorem]{Lemma}

\newtheorem{remark}[theorem]{Remark}

\newtheorem{prop}[theorem]{Proposition}
\newtheorem{thm}[theorem]{Theorem}
\newtheorem{oss}[theorem]{Remark}

\theoremstyle{definition}

\newcommand{\R}{\mathbb{R}}
\newcommand{\Z}{\mathbb{Z}}
\newcommand{\C}{\mathbb{C}}

\newcommand{\be}{\begin{equation}}
\newcommand{\ee}{\end{equation}}
\newcommand{\bea}{\begin{eqnarray}}
\newcommand{\eea}{\end{eqnarray}}

\def\pmtwo#1#2#3#4{ \begin{bmatrix}#1&#2\\#3&#4\end{bmatrix}}

\DeclareMathOperator{\res}{Res}

\newcommand{\lb}{\lambda}

\def\XXint#1#2#3{{\setbox0=\hbox{$#1{#2#3}{\int}$}
     \vcenter{\hbox{$#2#3$}}\kern-.5\wd0}}

\def\d{{\rm d}}

\def\i{{\rm i}}
\def\1{\operatorname{Id}}

\def\wt{\widetilde}
\def\ov{\overline}
\def\res{\mathop{{\rm res}}}

\def\Ai{\mathrm{Ai}}
\def\Re{\operatorname{Re}}
\def\Im{\operatorname{Im}}
\def\exp{\operatorname{exp}}
\def\Bes{\mathrm{Bes}}
\def\Jac{\operatorname{dn}}

\def\le{\left}
\def\ri{\right}

\numberwithin{equation}{section}

\let\oldtocsection=\tocsection
 
\let\oldtocsubsection=\tocsubsection
 
\let\oldtocsubsubsection=\tocsubsubsection
 
\renewcommand{\tocsection}[2]{\hspace{0em}\oldtocsection{#1}{#2}}
\renewcommand{\tocsubsection}[2]{\hspace{2em}\oldtocsubsection{#1}{#2}}
\renewcommand{\tocsubsubsection}[2]{\hspace{4em}\oldtocsubsubsection{#1}{#2}}

\title{Rigorous asymptotics of a KdV soliton gas}
\author{M. Girotti}
\address{Mila Institute, Universit\'e de Montr\'eal, 6666 St-Urbain, H2S 3H1, Montr\'eal, QC}
\email{manuela.girotti@umontreal.ca}
\author{T. Grava}
\address{SISSA, via Bonomea 265, 34136 Trieste, Italy and School of Mathematics, University of Bristol, UK}
\email{grava@sissa.it}
\author{R. Jenkins}
\address{Department of Mathematics, University of Central Florida, 4393 Andromeda Loop N, Orlando, FL 32816}
\email{robert.jenkins@ucf.edu}
\author{K. D. T.-R. McLaughlin}
\address{Department of Mathematics, Colorado State University, 1874 campus delivery, Fort Collins, CO 80523}
\email{kenmcl@rams.colostate.edu}

\date{}

\begin{document}

\maketitle

\begin{abstract}
We analytically study the long time and large space asymptotics of a new  broad class of  solutions of the KdV equation introduced by Dyachenko, Zakharov, and Zakharov.  These solutions are characterized by a Riemann--Hilbert problem which we show arises as  the  limit  $N\to + \infty$ of a  gas of  $N$-solitons. We show that this gas of solitons  in the limit $N\to\infty$ is  slowly  approaching 
a cnoidal wave solution for $x \to - \infty$ (up to terms of order $\mathcal{O} (1/x)$), while  approaching zero exponentially fast for $x\to+\infty$.   
We establish an asymptotic description of the gas of solitons  for large times that is valid over the entire spatial domain, in terms of Jacobi elliptic functions.
\end{abstract}

\tableofcontents

\section{Introduction}
This paper concerns the concept of a gas of solitons for the Korteweg-de Vries (KdV) equation,
\begin{eqnarray}
\label{eq:KdV}
u_{t} - 6 u u_{x} + u_{xxx} = 0 \ .
\end{eqnarray}
It is well known that this nonlinear partial differential equation is integrable, arising as the compatibility condition of a Lax pair of linear differential operators.  The compatibility condition can be presented as the existence of a simultaneous solution to the pair of equations
\begin{eqnarray}
\label{eq:LP1}
&&-\psi_{xx} + u \psi = E \psi \, , \\
&&\psi_{t} -4\psi_{xxx} + 6 u \psi_{x} + 3 u_{x} \psi = 0 \ ,
\end{eqnarray}
where $E$ is the spectral parameter and  $\psi=\psi(x,t)$.
The Lax pair formulation yields a complete solution procedure for the initial value problem for (\ref{eq:KdV}) via the inverse scattering transform in the case of rapidly decaying  or step-like initial data, and has led to a large and ever-growing collection of results concerning the analysis of the initial value problem in many different asym\-ptotic regimes, including the behaviour in the small dispersion limit, as well as a complete description of the long-time behaviour for fairly general decaying or step-like initial conditions.  In the case of periodic boundary conditions as well, there have been many works that are aimed at understanding the behaviour of solutions as well as the geometry of the space of solutions.  These works have all been driven by the physical 
origins of the KdV equation as a basic model for one-dimensional wave motion of the interface between air and water, and in particular the discovery of the soliton.
The soliton is a rapidly decreasing  travelling wave  solution of the KdV equation, namely a solution of the form $u(x,t)=f(x-vt)$ and takes the form
\begin{equation}
\label{soliton}
u(x,t)=-2\eta^2\operatorname{sech}^2\le( 2\eta(x-4\eta^2t-x_0) \ri)
\end{equation}
where  $E=-\eta^2$  is  the energy  parameter of   Schr\"odinger equation in the Lax pair \eqref{eq:LP1}. 
The periodic travelling wave that can  be obtained by direct integration   of the KdV equation takes the form
\begin{equation}
\label{tw}
u(x,t)=\beta_1+\beta_2-\beta_3-2(\beta_1-\beta_3)\mbox{dn}^2(\sqrt{\beta_1-\beta_3}(x+2(\beta_1+\beta_2+\beta_3)t+x_0)|m)
\end{equation}
where $\mbox{dn}(z|m)$ is the  Jacobi elliptic function of modulus $m^2=\frac{\beta_2-\beta_3}{\beta_1-\beta_3}$ and $\beta_1>\beta_2>\beta_3$. In both formulas $x_0$ is an arbitrary phase.
Let us introduce the $\vartheta$ function
\[
\vartheta_3 (z;\tau)= \sum_{n\in \mathbb{Z}} e^{2\pi i \, nz +  \pi  n^2 i \tau} \ , \qquad z \in \mathbb{C} \ ,\quad \Im \tau>0.
\]
Using the standard relation between  Jacobi  elliptic functions and $\vartheta$-function  (see eg. \cite{Lawden}  pg. 45 exercise 16 and 3.5.5)
we re-write \eqref{tw}  as
\begin{equation}
\begin{split}
\label{elliptic}
u(x,t)&= \bar{u}-2\frac{\partial^2}{\partial
x^2}\log\vartheta_3\left(\dfrac{\sqrt{\beta_1-\beta_3}}{2 K(m)}[x+2 t(\beta_1+\beta_2+\beta_3) +x_0]-\frac{1}{2};\hat{\tau}\right)\,\\
\bar{u}&=\beta_1+\beta_2-\beta_3-2(\beta_1-\beta_3)\dfrac{E(m)}{K(m)}\, 
\end{split}
\end{equation}
with  $K(m)=\int_0^{\pi/2}\frac{\d\vartheta}{\sqrt{1-m^2\sin^2\vartheta}}$ and $E(m)=\int_0^{\pi/2}\d\vartheta\sqrt{1-m^2\sin^2\vartheta}$, the complete elliptic integrals of the first 
and second kind respectively,  $\hat{\tau}=iK'(m)/K(m)$ and  $K'(m)=K(\sqrt{1-m^{2}})$.  We observe that $\bar{u}$ is the average value   of $u(x,t)$ over an oscillation.
 The above formula coincides with the genus-one case of the more general  Its-Matveev  and Dubrovin-Novikov  formula  \cite{ItsMatveev},\cite{DubrovinNovikov}  for finite-gap solutions of KdV.

With the above potential (\ref{tw}),  the Schr\"odinger equation  (\ref{eq:LP1}) coincides with the Lam\'e equation  and the stability zones  (or Bloch spectrum)  of the potential are $[\beta_3,\beta_2]\cup[\beta_1,+\infty)$.

  Of the highest importance for applications to the theory of water waves was the discovery of families of explicit   more complex solutions, such as N-soliton solutions when the Schr\"odinger equation in \eqref{eq:LP1} has $N$ simple eigenvalues, or a $N$-gap solution  when there are $N+1$ disjoint  stability zones of the corresponding Schr\"odinger equation   or  solutions that connect to Painlev\'{e} transcendents.  

Since the early days of integrable nonlinear PDEs, researchers have considered the notion of a soliton gas (see \cite{ZintTurb}, and references contained therein).  The quest is for an understanding of the properties of an interacting ensemble of many solitons, ultimately in the presence of randomness.  However, even in the absence of randomness, the dynamics of a large collection of solitons is only understood with mathematical precision in a few specific settings (the small-dispersion limit of the KdV equation, as considered in the works of Lax and Levermore \cite{LaxLevI,LaxLevII,LaxLevIII}, could be interpreted as a highly concentrated soliton gas, with a smooth and rapidly decaying function being represented as an infinite accumulation of solitons).

Within integrable turbulence, the interest is in the computation of statistical quantities describing the evolution of random configurations of solitons.  In \cite{DutPel} and \cite{PelShu} the authors used computational methods to approximate such statistical quantities via the Monte-Carlo method, and presented a formal derivation of evolution equations for the first four statistical moments of the solution.  In another direction \cite{Z71,ElKa05} the interest is in computing a kinetic equation describing the evolution of the spectral distribution functions.  This has been extended to similar formal considerations based on properties of fundamental solutions in the periodic setting, as opposed to solitonic gasses \cite{Tovbis2,El16,ElKamPavZyk11}.  

\subsection{The soliton gas}

Towards the goal of discovering new, broad families of solutions to integrable nonlinear PDEs, the ``dressing method" as developed by Zakharov and Manakov \cite{ZM85} has yielded some interesting new results in \cite{ZakZakDya}.  In that paper, the authors show how the dressing method can be used to produce a new family of solutions they refer to as ``primitive potentials" which, although are not random, can be naturally interpreted as a soliton gas.  Cutting to the chase, the authors derive a Riemann--Hilbert problem which seeks a vector $\Xi = \begin{bmatrix} \Xi_{1}& \Xi_{2} \end{bmatrix}^{T}$ satisfying  a normalization condition at $\infty$, and the jump relations

\begin{eqnarray}
\Xi_{+}(i \lambda) = J(\lambda) \Xi_{-}( i \lambda) \, , \qquad \Xi_{+}(-i \lambda) = J^{T}(\lambda) \Xi_{-}( -i \lambda) \, ,  \qquad \lambda \in (\eta_1,\eta_2)\, 
\end{eqnarray}
where the jump matrix $J(\lambda)$ is given by
\begin{eqnarray}
\label{eq:DZZRHP}
J(\lambda) = \frac{1}{ 1 +r_1(\lb)r_2(\lb)}\begin{bmatrix}
\displaystyle 1 -r_1(\lb)r_2(\lb) &   \displaystyle 2 ir_1(\lb)  e^{ - 2 \lambda x} \\
\displaystyle 2 ir_2(\lb)  e^{2 \lambda x}  & \displaystyle  1 -r_1(\lb) r_2(\lb)  \end{bmatrix} \ .
\end{eqnarray}
The parameters $\eta_{1}$ and $\eta_{2}$ are taken to be real with $0 < \eta_{1} < \eta_{2}$, and the intervals $(i\eta_1,i\eta_2)$ and $(-i\eta_2,-i\eta_1)$ are oriented downwards.

The reflection coefficients  $r_1(\lb)=r_1(\lb; t)$ and $r_2(\lb)=r_2(\lb; t)$  evolve in time according to 
\begin{eqnarray}
 r_1(\lambda; t)=r_1(\lambda; 0) e^{ ( 8 \lambda^{3} - 12 \lambda) t} \, , \qquad r_2(\lambda; t) =r_2(\lambda; 0)e^{-(  8 \lambda^{3} - 12 \lambda) t} \ .
\end{eqnarray}
The authors consider a number of different settings, and use a combination of analytical and computational methods to provide a description of the solutions of the KdV equation determined by this Riemann--Hilbert problem.  In the case that $r_{2} \equiv 0$, the potential is exponentially decaying as $x \to +\infty$.   But the behavior as $x$ grows in the other direction (as well as the the asymptotic behavior for $|x|$ large in the case that both reflection coefficients are nontrivial) was mentioned as a challenging problem for both analysis and computation.

The configuration of solitons considered in \cite{ZakZakDya} is somewhat different than the solitonic gas configurations considered in \cite{DutPel} and \cite{PelShu}, where they considered a large number of solitons that were spaced quite far apart from each other at $t=0$.  In other words, they considered a dilute gas of solitons that had enough space between them to evolve as isolated solitons until they interact, usually in a pair-wise fashion.  In contrast, the soliton gas considered in \cite{ZakZakDya} (and considered here as well) is a configuration that cannot be viewed as a collection of isolated solitons.  Indeed, as we show, they are overlapping to the extent that,  at $t=0$  the potential approaches zero exponentially fast as $x \to +\infty$, while  for $x\to-\infty$ the potential approaches the cnoidal wave solution of KdV very slowly---the error decays with a rate of $O(\frac{1}{x})$.  Because of these different behaviors, these potentials represent a new large class of potentials which have not been previously  considered  in the literature.  This model is substantially different from the model of infinite solitons considered in \cite{Boyd} and \cite{Zaitsev} where  an infinite number of equally spaced and identical solitons can be identified with the cnoidal wave solution of KdV.

\subsection{Statement of the results} 
In Section \ref{sec:2} we consider a sequence of Riemann--Hilbert problems, indexed by $N$, for a pure $N$-soliton solution,  with spectrum  confined to the intervals $(-i\eta_2,-i\eta_1)\cup (i\eta_1,i\eta_2)$ for some $\eta_2>\eta_1>0$  and show that for this sequence, as $N \to + \infty$, the solution of the Riemann--Hilbert problem converges to the solution of the Riemann--Hilbert problem studied in \cite{ZakZakDya}, for the case $ r_2(\lambda) \equiv 0$.

{\bf Remark}. Since the Riemann--Hilbert problem emerges in a limit, the existence and uniqueness of a solution is not a-priori known.  For completeness, we provide a proof of existence which is valid for all $x$ and $t$ in the Appendix.

In Section \ref{sec:3}  (Theorem \ref{thm:3.4}) we establish that the potential  $u(x,0)$ determined by this Riemann--Hilbert problem coincides with the  periodic travelling wave  as 
 $x  \to - \infty$:
\begin{equation}
\label{udn_intro}
 u(x,0)=\eta_2^2-\eta_1^2-2\eta_2^2\Jac^2\le( \eta_2(x+\phi) + K(m)\le| \, m\ri. \ri)+  \mathcal{O}\le(x^{-1}\ri) \,.
\end{equation}
The function $\Jac \le( z\le| \, m\ri. \ri)$ is  the Jacobi elliptic function of modulus $m=\eta_1/\eta_2$.  It is  periodic with period $2K(m)$, and satisfies $\Jac \le( 0 \le| \, m\ri. \ri)=1$ and $\Jac \le( K(m) \le| \, m\ri. \ri)=\sqrt{1-m^2}$. 
The expression \eqref{udn_intro} for the elliptic solution of KdV coincides with  the   travelling wave solution \eqref{tw}  in the introduction   by identifying $\beta_1=0$, $\beta_2=-\eta_1^2$ and $\beta_3=-\eta_2^2$.

The     function (\ref{udn_intro}) is periodic in $x$ with period $ 2K(m)/\eta_2$.The minimum amplitude  of the oscillations is $-\eta_2^2-\eta_1^2$ and the maximum amplitude is 
$ \eta_1^2-\eta_2^2$ so that the amplitude of the oscillations is $2\eta_1^2$. The average value of $u(x)$ over an oscillation   can be obtained from \eqref{elliptic}.

The phase $\phi$ in formula (\ref{udn_intro}) depends on the coefficient $r_1(\lb)$ that characterizes the continuum limit of the norming constants of the soliton gas and it is equal to
\begin{equation}
\label{phi_intro}
\phi=\int_{\eta_1}^{\eta_2}\dfrac{\log 2r_1(i\zeta)}{\sqrt{(\zeta^2-\eta_1^2)(\zeta^2-\eta_2^2)}}\dfrac{\d\zeta}{\pi i}\in\R \, .
\end{equation}
\begin{remark}
The potential $u(x,0)$ is a step-like finite gap potential. The slow decay rate as $x \to -\infty$  implies that such potential does not fall in the class considered in \cite{Monvel}.
When $\eta_1=0$ the potential $u(x,0)=-\eta_2^2+  \mathcal{O}\le(x^{-1}\ri)$ as $x  \to - \infty$. Such a potential 
is a step-like potential with zero reflection coefficient on the real axis.
It  is not included in the class of potentials studied in \cite{EGKT} and \cite{CohenKappeler} because  of  the low  decaying condition at $x\to-\infty$.
 Potentials with a low decay rate have appeared when studying rogue waves of infinite order of the focusing nonlinear Schr\"odinger equation \cite{Bilman1}, see also \cite{Bilman2}.
\end{remark}

Finally in Sections \ref{sec:4}-\ref{sec:6} we provide a global long-time asymptotic description of the solution $u(x,t)$ to the KdV equation with this  initial data $u(x,0)$.  The asymptotic behaviour depends on the quantity $\xi = x / 4 t$.  There are  three main regions:  (1) a constant  region;   (2)   a region where the solution is approximated by a  periodic  traveling wave with constant coefficients  specified by the spectral data; and (3) a region where the solution is approximated by a  periodic travelling wave with  modulated coefficients  (see Figure~\ref{fig1}). More precisely:
\begin{itemize}
\item[(1)]  for fixed $\xi > \eta_{2}^{2}$, there is a positive constant $C = C(\xi)$ so that 
\[
u(x,t)= \mathcal{O}\le(e^{- C t}\ri) \, .
\]
\item[(2)] For $\xi<\xi_{\rm crit}$ we have 
\begin{equation}
\label{u_dn_intro2}
u(x,t)=\eta_2^2-\eta_1^2-2\eta_2^2\Jac^2\le( \eta_2(x-2(\eta_1^2+\eta_2^2)t+\phi) + K(m)\le| \, m\ri. \ri)+  \mathcal{O}\le(t^{-1}\ri) \, ,
\end{equation}
with $m=\eta_1/\eta_2$ and $\phi$ as in (\ref{phi_intro}).  The critical value $\xi_{\rm crit}$ is obtained from the equation
\begin{equation}
\label{Whitham}
\xi_{\rm crit}=\dfrac{\eta_2^2}{2}W(m)\, ,\quad   W(m)=1+m^2+2\dfrac{m^2(1-m^2)}{1-m^2-\frac{E(m)}{K(m)}}\, ,\quad m=\frac{\eta_1}{\eta_2}\, .
\end{equation}

\item[(3)]  For  $\xi_{\rm crit}<\xi<\eta_2^2$  we have that 
\begin{equation}
\label{u_dn_intro1}
u(x,t)=\eta_2^2-\alpha^2-2\eta_2^2\Jac^2\le( \eta_2(x-2(\alpha^2+\eta_2^2)t+\wt{\phi}) + K(m_\alpha) \le| \, m_\alpha\ri. \ri)+  \mathcal{O}\le(t^{-1}\ri) \, ,
\end{equation}
where $\Jac\le( z\le| \, m_{\alpha}\ri. \ri)$ is the Jacobi elliptic function of modulus $m_\alpha=\alpha/\eta_2$,
\[
\wt{\phi}=\int_{\alpha}^{\eta_2}\dfrac{\log 2r_1(i\zeta)}{\sqrt{(\zeta^2-\alpha^2)(\zeta^2-\eta_2^2)}}\dfrac{\d\zeta}{\pi i}\in\R\,
\]
and the coefficient $\alpha=\alpha(\xi)$ is determined from  the Whitham  modulation equation \cite{Whitham}
\begin{equation}
\label{xim_intro}
\xi  =\dfrac{x}{4t}=\dfrac{\eta_2^2}{2}W(m_\alpha) \, ,
\end{equation}
where $W(m)$ has been defined in (\ref{Whitham}).

 \end{itemize}
 The  equation \eqref{xim_intro} was used by Gurevich and Pitaevskii \cite{GP73} to describe the modulation of the travelling wave  that is formed in the solution of the KdV equation
with a step initial data $u(x,0)=-\eta_2^2$ for $x<0$ and $u(x,0)=0$ for $x>0$.   Such a modulated travelling wave is also called a dispersive shock wave.
The rigorous analysis of the dispersive shock wave emerging from step-like initial data problem was  obtained via inverse scattering in \cite{Hruslov} and more recently via Riemann--Hilbert methods in \cite{EGKT}.  

 \begin{figure}[h]
\includegraphics[width=.5\textwidth]{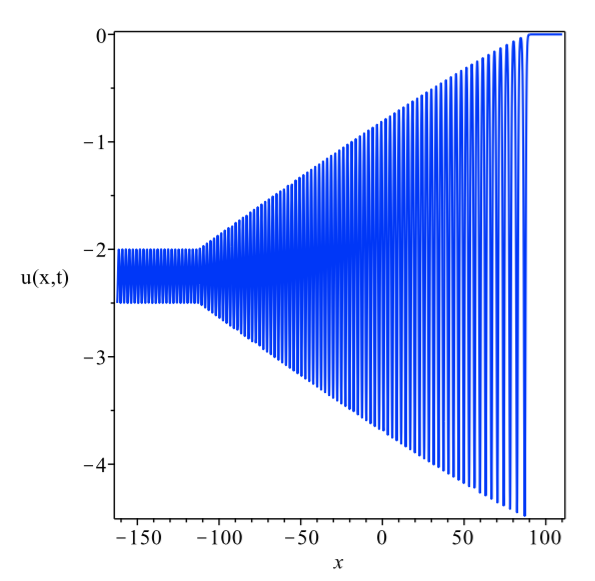}
\caption{Soliton gas behaviour at $t=10$ with endpoints $\eta_1=0.5$ and $\eta_2=1.5$ and reflection coefficient $r_1(\lambda) \equiv 1$. \label{fig1}}
\end{figure}
\section{Soliton gas as limit of $N$ solitons as $N\to +\infty$}
\label{sec:2}

The Riemann--Hilbert  problem for a pure $N$-soliton solution (see for example \cite{Teschl}) is described as follows:  find a $2$-dimensional row vector $M$ such that
\begin{enumerate}
\item $M(\lb)$ is meromorphic in $\mathbb{C}$, with simple poles at $\{\lambda_{j} \}_{j=1}^{N}$ in $i\mathbb{R}_{+}$, and at the corresponding conjugate points $\{\overline{\lambda_{j}} \}_{j=1}^{N}$ in $i\mathbb{R}_{-}$;
\item $M$ satisfies the residue conditions
\begin{align}
\label{residue_soliton}
\res_{\lambda=\lambda_{j}} M (\lb)= \lim_{\lambda \to \lambda_{j}} M(\lambda) \begin{bmatrix} 0 & 0 \\ \displaystyle \frac{c_j e^{2 i \lambda_{j} x }}{N} & 0 \end{bmatrix}\, ,  \quad \res_{\lambda=\overline{\lambda_{j}}} M(\lb) = \lim_{\lambda \to \overline{\lambda_{j}}} M(\lambda) \begin{bmatrix} 0 & \displaystyle \frac{- c_j e^{-2 i \overline{\lambda_{j}} x }}{N} \\ 0 & 0\end{bmatrix} \, ,
\end{align}
where $c_j\in i\R_+$;
\item $\displaystyle M(\lb)= \begin{bmatrix}1&1 \end{bmatrix} + \mathcal{O}\le(\frac{1}{\lambda}\ri)$ as $\lambda \rightarrow \infty$,
\item $M$ satisfies the symmetry
\[
M(-\lb)=M(\lb)\begin{bmatrix}0&1\\1&0\end{bmatrix}\,.
\]
\end{enumerate}
The solution of the above Riemann--Hilbert problem is determined from the relation
\begin{eqnarray}
\label{M_soliton}
M(\lambda) = \left(  1 + \sum_{j=1}^{N} \frac{e^{i\lambda_j x}\alpha_{j}}{\lambda- \lambda_{j}} , \ 1 - \sum_{j=1}^{N} \frac{e^{i\lambda_j x}\alpha_{j}}{\lambda+ \lambda_{j}}  \right) \ , 
\end{eqnarray}
where the constants $\alpha_j$ are uniquely determined by the residue conditions \eqref{residue_soliton}.
The $N$-soliton potential $u(x)$ is determined from $M$ via
\begin{gather}
u(x) = 2 \frac{\d}{\d x} \left(
\lim_{\lambda \to \infty} \frac{ \lambda}{i} \left( M_{1}(\lambda) - 1 \right)
\right) \ ,
\end{gather}
where $M_1(\lambda)$ is the first entry of the vector $M(\lambda)$.
In particular, for a one-soliton potential, namely $N=1$, one recovers  the expression (\ref{soliton}) where the shift  $ x_0$ is given by
\[
x_0=\dfrac{1}{4\eta_1}\log\frac{c_1}{2i\eta_1}\in\R.
\]

We are interested in the limit as $N  \to + \infty$, under the additional assumptions: 
\begin{enumerate}
\item The poles $\left\{ \lambda_{j}^{(N)} \right\}_{j=1}^{N}$ are sampled from a density function $\varrho(\lambda)$ so that $\int_{\eta_{1}}^{-i  \lambda_{j}} \varrho( \eta ) d \eta = j / N $, for $j = 1, \ldots, N$.
\item The coefficients $\{c_j\}_{j=1}^N$ are purely imaginary (in fact $c_{j} \in i \mathbb{R}_{+}$) and are assumed to be discretizations of a given function:
\begin{gather}
c_j = \frac{i (\eta_2-\eta_1) r_1(\lambda_j)}{\pi} \qquad j=1,\ldots, N \ . \label{cjasR1}
\end{gather}
where $ r_1(\lambda)$ is an analytic function for $\lambda$ near the intervals $(i \eta_{1}, i \eta_{2})$ and $(-i \eta_{2}, -i \eta_{1})$, with the symmetry $ r_1(\ov \lambda) =   r_1(\lambda)$, and is further assumed to be a real valued positive and non-vanishing function of $\lambda$ for $\lambda \in [i \eta_{1}, i \eta_{2}]$.
\end{enumerate}

In the regime $x \to + \infty$, it is easy to notice that all residue conditions   \eqref{residue_soliton} contain only exponentially small terms and therefore, by a small norm argument, the potential is exponentially small.  

On the other hand, for $x \to - \infty$ all of those terms are exponentially large.  To show that the solution is also exponentially small in this latter case, one may reverse the triangularity of the residue conditions, by defining
\begin{eqnarray}
A(\lambda)=M(\lambda) \prod_{j=1}^{N} \left( \frac{\lambda - \lambda_j }{\lambda-\overline{\lambda_j}} \right)^{\sigma_{3}} \ .
\end{eqnarray}

Now the residue conditions are
\begin{eqnarray}&&
\res_{\lambda=\lambda_{j}} A(\lb) = \lim_{\lambda \to \lambda_{j}} A(\lambda) \begin{bmatrix} 0 &\displaystyle  \frac{N}{c_j} e^{-2 i \lambda_{j} x }(\lambda_{j}-\overline{\lambda_{j}})^{2}\prod_{k\neq j} \left( \frac{\lambda_{j}-\overline{\lambda_{k}}}{\lambda_{j}-\lambda_{k}}\right)^{2} \\ 0 & 0 \end{bmatrix} \\ &&
\res_{\lambda=\overline{\lambda_{j}}} A(\lb) = \lim_{\lambda\to \overline{\lambda_{j}}} A(\lambda) \begin{bmatrix} 0 & 0 \\ \displaystyle \frac{-N}{c_j} e^{2 i \overline{\lambda_{j}} x }
(\overline{\lambda_{j}}-\lambda_{j})^{2} \prod_{k\neq j} \left( \frac{\overline{\lambda_{j}}-\lambda_{k}}{\overline{\lambda_{j}}-\overline{\lambda_{k}}}\right)^{2} & 0 \end{bmatrix} 
\end{eqnarray}
while the potential $u(x)$ is still extracted from $A$ via the same calculation:
\begin{eqnarray*}
u(x) = 2  \frac{\d}{\d x} \left(
\lim_{\lambda \to \infty} 
\frac{\lambda}{i} \left( A_{1}(\lambda) - 1 
\right)
\right) \, .
\end{eqnarray*}

%
%
The quantity $e^{ -2 i \lambda_{j} x} $ now decays exponentially as $x  \to - \infty$, and this implies (again by a standard small-norm argument) exponential decay of the potential $u(x)$ for $x  \to - \infty$.  On the other hand, the product term is exponentially large in $N$.  One may show that there is $C>0$ so that 
\begin{gather} \left| \frac{N}{c_{j}} ( \lambda_{j} - \overline{\lambda_{j}})
\prod_{k\neq j } \left(\frac{\lambda_{j}-\overline{\lambda_{k}}}{\lambda_{j}-\lambda_{k}} \right)^{2} \right| < D e^{CN}  \qquad \text{for all } j=1,\ldots, N \ .
\end{gather}

Therefore this exponential decay does not set in until $x$ is rather large.  Indeed, in order for the residue conditions to all be exponentially small, it must be that $x \ll - C N$. 
In other words, the $N$-soliton solution that we are considering has very broad support, and in the large-$N$ limit, it is not exponentially decaying for $x \to -\infty$.  To be more precise, the above computations can be used to show the following lemma.
\begin{lemma}
For any $0 < \tilde{k} < \eta_{1}/2$, there exists a constant $\tilde{C}$ so that if $\displaystyle x < - \tilde{C} N $, then 
\begin{eqnarray}
|u(x)| < e^{-\tilde{k} |x|} \ .
\end{eqnarray}
In other words, for $x < - \tilde{C} N$, the potential $u(x)$ is exponentially decreasing.  
\end{lemma}
The proof of this lemma is straightforward: under the hypotheses of the lemma, all residue conditions are exponentially small.  One may replace the residue conditions with jumps across small circles encircling the poles, and the jumps are all of the form $I + \mathcal{O}\left(e^{-\tilde{k} |x|} \right)$, so small norm existence theory applies.

We will show here how to derive the Riemann--Hilbert  problem for a soliton gas with one reflection coefficient $ r_1$ (as described in \cite{ZakZakDya}) from a meromorphic Riemann--Hilbert  problem for $N$ solitons in the limit as $N \to+\infty$.  First, we remove the poles by defining
\begin{eqnarray}
\label{eq:2.10Z}
Z(\lambda) = M(\lambda) \begin{bmatrix} 1 & 0 \\ \displaystyle - \frac{1}{N}\sum_{j=1}^{N}\frac{c_j e^{2 i \lambda x }}{\lambda - \lambda_{j}} & 1\end{bmatrix}
\end{eqnarray}
within a closed curve $\gamma_+$ encircling the poles counterclockwise in the upper half plane $\mathbb{C}_{+}$, and
\begin{eqnarray}
\label{eq:2.11Z}
Z(\lambda) = M(\lambda) \begin{bmatrix}{1} & \displaystyle { \frac{1}{N}\sum_{k=1}^{N}\frac{c_j e^{-2 i \lambda x }}{\lambda - \overline{\lambda_{j}}}}\\
{0}&{1} \end{bmatrix}
\end{eqnarray}
within a closed curve $\gamma_-$ surrounding the poles clockwise in the lower half plane $\mathbb{C}_{-}$.  Outside these two sets, we take $Z(\lambda)=M(\lambda)$.

Then the jumps are
\begin{gather}
\label{eq:2.12Z}
Z_+(\lambda) = Z_-(\lambda) \begin{cases} \begin{bmatrix} 1 & 0 \\ \displaystyle - \frac{1}{N}\sum_{j=1}^{N}\frac{c_j e^{2 i \lambda x }}{\lambda - \lambda_{j}} & 1\end{bmatrix} & \lambda \in \gamma_+ \\
 \begin{bmatrix}{1} & \displaystyle { -\frac{1}{N}\sum_{k=1}^{N}\frac{c_j e^{-2 i \lambda x }}{\lambda - \overline{\lambda_{j}}}}\\
{0}&{1} \end{bmatrix} & \lambda \in \gamma_- \ 
\end{cases}
\end{gather}
where, for $\lambda \in \gamma_{+}$ or $\gamma_{-}$, the boundary values $Z_{+}(\lambda)$ are taken from the left side of the contour as one traverses it according to its orientation, and the boundary values $Z_{-}(\lambda)$ are taken from the right.  The quantity $Z(\lambda)$ is normalized so that $Z(\lambda) = \begin{bmatrix}1&1 \end{bmatrix} + \mathcal{O}\le(\lambda^{-1}\ri)$ as $\lambda \to \infty$.

We assume now that in the limit as the number of poles goes to infinity, the poles are distributed according to some distribution $\varrho(\lambda)$ with density compactly supported in $(i\eta_1,i\eta_2)$ (and extended by symmetry on the corresponding interval in the lower half plane). 

For the sake of simplicity, we can assume that the $N$ poles are equally spaced along $(i\eta_1,i\eta_2)$ with distance between two poles equal to $| \Delta \lambda| = \frac{\eta_2-\eta_1}{N}$ and with (atomic) density:
\begin{gather}
\varrho_{N}(\lambda) = \frac{1}{Z_N} \sum_{j=1}^N c_j \delta_{\lambda_j}(\lambda) \qquad \lambda \in (i\eta_1,i\eta_2) \ ,
\end{gather}
for some normalization constant $Z_N$.

\begin{oss}
In the case where the poles are distributed according to a more general measure $\varrho(\lambda)$, the steps to follow are very similar. The entries of the jump matrices will carry the density function along, which can be eventually incorporated in the definition of the reflection coefficient $ r_1(\lambda)$.
\end{oss}

As the number of poles increases within the support of the measure, the following result holds.
\begin{prop}For any open set $K_{+}$ containing the interval $[i \eta_{1}, i \eta_{2}]$, and any open set  $K_{-}$ containing the interval $[-i \eta_{2}, - i \eta_{1}]$,  the following limit holds uniformly for all  $\lambda \in \mathbb{C }\backslash K_{+}$:
\begin{gather}
\lim_{N\to +\infty} \frac{1}{N}\sum_{j=1}^{N}\frac{c_j }{\lambda - \lambda_{j}} =\int_{i\eta_1}^{i\eta_2} \frac{2i  r_1(\zeta)}{\lambda - \zeta}\frac{\d \zeta}{2\pi i} \,,
\end{gather}
and the following limit holds uniformly for all $\lambda \in \mathbb{C} \backslash K_{-}$:
\begin{gather}
\lim_{N\to +\infty} \frac{1}{N}\sum_{j=1}^{N}\frac{c_j }{\lambda - \overline{\lambda_{j}}} = \int_{-i\eta_2}^{-i\eta_1} \frac{2i  r_1(\zeta)}{\lambda - \zeta}\frac{\d \zeta}{2\pi i} \ ,
\end{gather}
 where $ r_1(\lambda)$ is an analytic function for $\lambda$ near the intervals $(i \eta_{1}, i \eta_{2})$ and $(-i \eta_{2}, -i \eta_{1})$,  and it is   assumed to be a positive real-valued and non-vanishing function of $\lambda$ for $\lambda \in [i \eta_{1}, i \eta_{2}]$.
\end{prop}
\begin{proof}
Using (\ref{cjasR1}), the expressions in the jumps can be rewritten as 
\begin{gather}
\frac{1}{N}\sum_{j=1}^{N}\frac{c_j }{\lambda - \lambda_{j}} = \frac{1}{N}\sum_{j=1}^{N}\frac{1}{\lambda - \lambda_{j}}  \frac{(\eta_2-\eta_1)i  r_1(\lambda_j)}{\pi} 
= \frac{1}{2\pi i}\sum_{j=1}^{N}\frac{ 2i r_1(\lambda_j)  }{\lambda - \lambda_{j}} \Delta \lambda  \ .
\end{gather}
The convergence follows from the convergence of the Riemann sum to the Riemann--Stieltjes integral for $x \in K$ any compact subset of $\R$.
Positivity of $r_1(\lb)$ follows from the fact that $c_j\in i\R_+$.
\end{proof}

Thanks to the proposition above and a small norm argument, we arrive at a limiting Riemann--Hilbert  problem (which we still call $Z$ with abuse of notation) 
\begin{align}
&Z_{+} (\lambda)= Z_{-}(\lambda) 
\begin{cases}
\label{eq:2.17Z}
\begin{bmatrix}{1}& {0}\\ \displaystyle {e^{2i\lambda x}\int_{i\eta_1}^{ i\eta_2} \frac{ 2i  r_1(\zeta)}{\zeta - \lambda }\frac{\d \zeta}{2\pi i}} & {1} \end{bmatrix}  \qquad \lambda \in  \gamma_+ \\
\begin{bmatrix} {1}&\displaystyle {e^{-2i\lambda x} \int_{-i\eta_2}^{-i\eta_1}  \frac{ 2i  r_1(\zeta)}{    \zeta- \lambda }\frac{\d \zeta}{2\pi i} } \\{0}& {1} \end{bmatrix} \qquad \lambda \in \gamma_-
\end{cases} \\
\label{eq:2.18Z}
&Z(\lambda) = \begin{bmatrix}1&1 \end{bmatrix} + \mathcal{O}\le(\frac{1}{\lambda}\ri) \qquad \lambda \to \infty \ .
\end{align}
At this point it is important to point out to the reader that the contour $(i \eta_{1}, i \eta_{2})$   and  $(-i \eta_{2}, -i \eta_{1})$ are both oriented upwards.

Next, we define
\begin{eqnarray}
X(\lambda) = Z(\lambda) \begin{bmatrix} {1}& {0}\\ -\displaystyle { e^{ 2 i \lambda x} \int_{i\eta_1}^{ i\eta_2} \frac{2i  r_1(\zeta)}{\zeta-\lambda} \frac{\d \zeta}{2\pi i} }& {1}\end{bmatrix}
\end{eqnarray}
within the loop  $\gamma_+$, and
\begin{eqnarray}
X(\lambda)= Z(\lambda) \begin{bmatrix} {1}&\displaystyle { e^{- 2 i \lambda x} \int_{-i\eta_2}^{-i\eta_1} \frac{2i  r_1(\zeta)}{ \zeta-\lambda } \frac{\d \zeta}{2\pi i}  }\\ {0}& {1} \end{bmatrix}
\end{eqnarray}
within the loop $\gamma_-$.  Outside these two curves, we define $X(\lambda) = Z(\lambda)$.

The jumps across the curves are no longer present, but there are jumps across $(i \eta_{1}, i \eta_{2})$ and $(-i \eta_{2}, -i \eta_{1})$ because the integrals have jumps across those intervals. Using the Sokhotski-Plemelj formula,
we arrive at a Riemann--Hilbert  problem for $X$: 
\begin{align}
&X_{+}(\lambda) = X_{-}(\lambda) \begin{cases} \begin{bmatrix} {1}& {0} \\ \displaystyle {-2i   r_1(\lambda) e^{ 2 i \lambda x} } & {1} \end{bmatrix} & \lambda \in (i\eta_1,i\eta_2) \\
\begin{bmatrix} {1}& \displaystyle {2i  r_1(\lambda) e^{ -2 i \lambda x} } \\{0} & {1} \end{bmatrix} &  \lambda \in (-i\eta_2,-i\eta_1) 
\end{cases} \label{2.16} \\
& X(\lambda) = \begin{bmatrix}1&1 \end{bmatrix} + \mathcal{O}\le(\frac{1}{\lambda}\ri) \qquad \lambda \rightarrow \infty \ , \label{2.20}\\
&X(-\lb)=X(\lb)\begin{bmatrix}0&1\\1&0\end{bmatrix}\,.\label{2.21}
\end{align}

This Riemann-Hilbert  problem is equivalent to the one described in \cite{ZakZakDya} with $r_2(\lambda) = 0$, up to a transposition ($X$ is a row vector here, while the solution of the Riemann--Hilbert  problem in \cite{ZakZakDya} is a column vector) and using the symmetry that $ r_1(\overline{\lambda})  =  r_1(\lambda) $ for $\lambda \in (-i \eta_{1}, - i \eta_{2})$.  We note that there is a sign discrepancy between this Riemann-Hilbert problem and the one appearing in \cite{ZakZakDya}, which is resolved by a careful interpretation of the sign conventions therein.

Since this Riemann-Hilbert problem has been derived through a limiting process, it is not at all clear that it actually possesses a solution.  Although this will eventually follow for large $x$ from the asymptotic analysis presented herein, we present a self-contained proof of existence and uniqueness in Appendix~\ref{appendix} of this paper.  In fact in the appendix we establish the existence of a matrix-valued solution ${\bf Y}$.  Equipped with that, it is straightforward to prove the following lemma.
\begin{lemma}
Let $B$ be an arbitrary positive number.  For all $x$ such that $|x|<B$, the quantity $Z$ defined in (\ref{eq:2.10Z})-(\ref{eq:2.11Z}) satisfying the jump relation (\ref{eq:2.12Z}) converges as $N \to \infty$ to the solution of the Riemann-Hilbert problem (\ref{eq:2.17Z})-(\ref{eq:2.18Z}),  and the $N$-soliton potential $u(x)$ converges to the potential determined by the solution to the soliton gas Riemann-Hilbert problem (\ref{2.16})-(\ref{2.21}).
\end{lemma}
Note:  A similar  limiting procedure introduced in this section has already appeared in the literature when studying the focusing nonlinear Schr\"odinger equation \cite{Bilman2,Bilman3}.
In those papers   the limiting procedure $N\to\infty$ refers  to the order of a soliton or a breather of the nonlinear Schr\"odinger equation.
From a different point of view, our limiting procedure can be understood as  replacing the reflection coefficient with its semiclassical limit, see for example
\cite{Tovbis},\cite{LaxLevI}.

In what follows we have already taken the $N \to \infty$ limit, and we are considering the behavior for $x$ (and later $x$ and $t$) large for this soliton gas Riemann-Hilbert problem. 

\section{Behaviour of the potential $u(x,0)$ as $x\to -\infty$}
\label{sec:3}

We consider a soliton gas Riemann--Hilbert  problem as in (\ref{2.16})--(\ref{2.20}) with $0 < \eta_{1} < \eta_{2}$, and reflection coefficient $ r_1(\lambda)$   defined on $( i \eta_{1}, i \eta_{2})$  such that it has an analytic extension to a neighbourhood of this interval.  Furthermore, we assume that $r_1(-\lb)=r_1(\lb)$ on the imaginary axis.
We set $\Sigma_1 = (\eta_1,\eta_2)$ and $\Sigma_2 = (-\eta_2,-\eta_1)$. The vector valued function $X$ that will determine the KdV potential $u(x)$ is the solution to the following Riemann--Hilbert  problem:
\begin{align}
& X(\lb) \text{ is analytic for } \lambda \in  \C \backslash \le\{ i\Sigma_1 \cup i\Sigma_2 \ri\} \nonumber \\
&X_+(i\lambda) = X_-(i\lambda) \begin{cases}
\displaystyle \begin{bmatrix} 1 & 0 \\ -2i  r_1(i\lb)e^{-2\lambda x} & 1 \end{bmatrix} &\quad \lambda \in \Sigma_1\\
\displaystyle \begin{bmatrix} 1 &2i r_1(i\lb) e^{2\lambda x}  \\ 0 & 1 \end{bmatrix} & \quad \lambda  \in \Sigma_2
 \end{cases} 
 \label{RHP00} \\
 & X (\lambda) = \begin{bmatrix}1 &1\end{bmatrix} + \mathcal{O}\le(\frac{1}{\lambda}\ri) \qquad  \lambda \rightarrow \infty \ .\nonumber
\end{align}

As explained in \cite{ZakZakDya}, we can recover the potential $u(x)$ of the Schr\"odinger operator via the formula 
\begin{gather}
u(x) = 2\frac{\d}{\d x} \le[ \lim_{\lambda \rightarrow \infty} \frac{\lambda}{i} (X_1(\lambda;x) - 1) \ri] \ ,
\end{gather}
where $X_1(\lambda;x)$ is the first component of the solution vector $X$.


We first perform a rotation of the problem in order to place the jumps on the real line. By setting
\begin{gather}
Y (\lambda)= X(i\lambda) \, ,\quad r(\lb)=2r_1(i\lb)\, ,
\end{gather}
the Riemann--Hilbert  problem for $Y$ reads as follows:
\begin{align}\label{jumpY}
&Y_+ (\lambda) = Y_-(\lambda) \begin{cases}
\displaystyle \begin{bmatrix} 1 &0  \\ -i r(\lb) e^{-2\lambda x}  & 1 \end{bmatrix} &\quad \lambda \in  \Sigma_1\\
\displaystyle \begin{bmatrix} 1 & i r(\lb)e^{2\lambda x} \\ 0  & 1 \end{bmatrix} & \quad \lambda \in  \Sigma_2
 \end{cases}\\
&Y(\lambda) = \begin{bmatrix}1&1 \end{bmatrix} + \mathcal{O}\le(\frac{1}{\lambda}\ri) \qquad \lambda \rightarrow \infty\\
\label{symY}
&Y(-\lambda) = Y(\lb)\begin{bmatrix}0&1\\1&0 \end{bmatrix} \,.
\end{align}
The contours $\Sigma_1$ and $\Sigma_2$ are shown in \figurename~\ref{RHPY}. 
We can recover $u(x)$ from 
\begin{gather}
u(x) = 2\frac{\d}{\d x} \le[ \lim_{\lambda \rightarrow \infty} \lambda (Y_1(\lambda;x) - 1) \ri] \, .
\end{gather}

\begin{figure}
\centering
\scalebox{1}{
\begin{tikzpicture}[>=stealth]
\path (0,0) coordinate (O);

\draw[dashed, ->] (-5,0) -- (5,0) coordinate (x axis);
\draw[dashed, ->] (0,-1.5) -- (0,1.5) coordinate (y axis);

\draw[->- = .5,thick] (1,0) -- (3,0);
\node [above] at (2,.5) {\large $\Sigma_1$};

\draw[fill] (1,0) circle [radius=0.05];
\node[above] at (1,0) {$\eta_1$};
\draw[fill] (3,0) circle [radius=0.05];
\node[above] at (3,0) {$\eta_2$};

\draw[->- = .5,thick] (-3,0) -- (-1,0);
\node [above] at (-2,.5) {\large $\Sigma_2$};

\draw[fill] (-1,0) circle [radius=0.05];
\node[above] at (-1,0) {$-\eta_1$};
\draw[fill] (-3,0) circle [radius=0.05];
\node[above] at (-3,0) {$-\eta_2$};

\end{tikzpicture}
}
\caption{Riemann--Hilbert  problem for $Y$.}
\label{RHPY}
\end{figure}
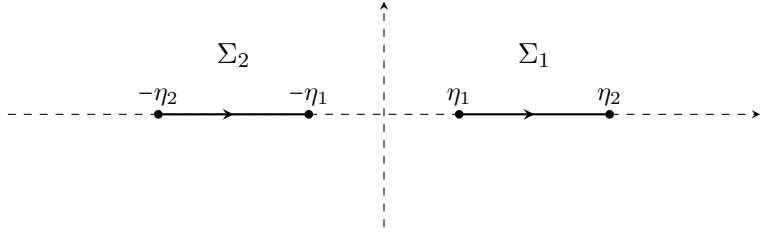

\subsection{Large $x$ asymptotic}
Introduce the following new  vector function
\[
T(\lb)=Y(\lb)e^{x g(\lambda)\sigma_3}f(\lambda)^{\sigma_3}
\]
where $g(\lb)$ and $f(\lambda)$ are  scalar  functions to be determined below.
We require that
\begin{itemize}
\item $g(\lb)$ is analytic in $\C\backslash[-\eta_2, \eta_2] $ and  
\begin{align}
& g_+(\lambda) + g_-(\lambda) = 2\lambda & \lambda \in \Sigma_1\cup\Sigma_2 \label{gconstraint1} \\
& g_+ (\lambda) - g_-(\lambda) = \Omega & \lambda \in [-\eta_1,\eta_1]  \label{gconstraint2}\\
& g(\lambda) =\mathcal{O}\le(\frac{1}{\lambda}\ri) & \lambda \rightarrow  \infty \, , \label{gconstraint3}
\end{align}
where   $\Omega$ is a constant independent of $x$, still to be determined, and 
\item $f(\lb)$ is analytic in $\C\backslash[-\eta_2, \eta_2] $ and  
\begin{equation}
\label{f1}
f(\lb)=1+\mathcal{O}\le(\frac{1}{\lb}\ri),\quad \mbox{as $ \lambda \rightarrow  \infty \,$}\, .
\end{equation}
\end{itemize}

In order to solve the  scalar  Riemann--Hilbert  problem  \eqref{gconstraint1} -- \eqref{gconstraint3} for $g$ we observe that
\begin{align}
&g'_+(\lambda) +g'_-(\lambda) =2 &\lambda \in \Sigma_1\cup \Sigma_2 \\
&g'_+(\lambda) -g'_-(\lambda) =0 &\lambda \in [-\eta_1,\eta_1] \\
&g'(\lambda) =\mathcal{O}\le(\frac{1}{\lambda^2}\ri) &\lambda \rightarrow \infty \ .
\end{align}
From the above,   we can write $g'(\lb)$ as
\begin{gather}
g'(\lambda) = 1 - \frac{\lambda^2+\kappa}{R(\lambda)}\, ,
\end{gather}
where
\begin{gather}
\label{eq:Rdef}
R(\lambda) = \sqrt{(\lambda^2-\eta_1^2)(\lb^2-\eta_2^2)}\, ,
\end{gather}
is real and positive on $(\eta_2,+\infty)$ with branch cuts on the contours $\Sigma_1$ and $\Sigma_2$ 
and $\kappa$  is a constant  to be determined.
By integration we obtain 
\begin{equation}
\label{ggg}
g(\lb)=\lb-\int_{\eta_2}^{\lb}\frac{\zeta^2+\kappa}{R(\zeta)}\d\zeta \, .
\end{equation}

The condition \eqref{gconstraint1}  implies that $$\int_{-\eta_1}^{\eta_1}\frac{\zeta^2+\kappa}{R(\zeta)}\d\zeta=0$$ and the condition  \eqref{gconstraint2} implies that
\[
\Omega=2\int_{\eta_1}^{\eta_2}\frac{\zeta^2+\kappa}{R_+(\zeta)}\d\zeta \, .
\]
This gives
\begin{gather}
\Omega = \frac{2\pi i}{\int_{-\eta_1}^{\eta_1}\frac{\d \zeta}{R(\zeta)} } =-\dfrac{i\pi \eta_2}{K(m)}\in i \R_- \ ,\quad m=\frac{\eta_1}{\eta_2}\, ,\label{Omega}
\end{gather}
where $K(m)=\int_0^{\frac{\pi}{2}}\frac{\d\vartheta }{\sqrt{1-m^2\sin\vartheta }}$ is the complete elliptic integral of the first kind with modulus $m=\eta_1/\eta_2$ and 
\begin{gather}
\kappa = - \int_{-\eta_1}^{\eta_1} \frac{ \zeta^2 \Omega }{R(\zeta)} \frac{\d \zeta}{2\pi i}=\eta_2^2\left(\dfrac{E(m)}{K(m)}-1\right) \quad  \in \R_- \ , \label{kappag'}
\end{gather}
where $E(m)=\int_0^{\frac{\pi}{2}}\sqrt{1-m^2\sin\vartheta }\, \d\vartheta $ is the complete elliptic integral of the second kind.

%
%
The Riemann--Hilbert  problem for $T(\lb)$ is 
\begin{align}
\label{RH_T}
&T_+(\lb)=T_-(\lb)V_T(\lambda)\\
&V_T(\lambda)=\begin{cases}
\displaystyle \begin{bmatrix} e^{x\le(g_+(\lb) - g_-(\lb)\ri)}\dfrac{f_+(\lb)}{f_-(\lb)} & 0 \\- i r(\lambda)f_+(\lb)f_-(\lb)  & e^{-x\le(g_+(\lb) - g_-(\lb)\ri)}\dfrac{f_-(\lb)}{f_+(\lb)} \end{bmatrix} &\quad \lambda \in  \Sigma_1\\
\displaystyle \begin{bmatrix} e^{x\le(g_+(\lb) - g_-(\lb)\ri)}\dfrac{f_+(\lb)}{f_-(\lb)}&\dfrac{ i r(\lambda)}{f_+(\lb)f_-(\lb)}   \\ 0 & e^{-x\le(g_+(\lb) - g_-(\lb)\ri)}\dfrac{f_-(\lb)}{f_+(\lb)} \end{bmatrix} & \quad \lambda \in \Sigma_2 \\
\displaystyle \begin{bmatrix} e^{x\Omega }\dfrac{f_+(\lb)}{f_-(\lb)} & 0 \\0& e^{-x \Omega}\dfrac{f_-(\lb)}{f_+(\lb)} \end{bmatrix} &\quad \lambda \in  [-\eta_1,\eta_1]\\
  \end{cases} \\
  &T(\lambda) = \begin{bmatrix}1&1 \end{bmatrix} + \mathcal{O}\le(\frac{1}{\lambda}\ri) \qquad \lambda \rightarrow \infty \ .
\end{align}

In order to solve the Riemann--Hilbert  problem for $T(\lb)$ we   wish  to obtain a constant jump matrix $J_T$. For this purpose we make the  following ansatz on the function $f$
\begin{align}
& f_+(\lambda) f_-(\lambda) = \frac{1}{r(\lambda)} & \lambda \in \Sigma_1 \label{fconstraint1} \\
&f_+(\lambda) f_-(\lambda) = r(\lambda) & \lambda \in \Sigma_2\\
&  \frac{f_+(\lambda)}{ f_-(\lambda) }=e^{\Delta} & \lambda \in [-\eta_1,\eta_1] \label{fconstraint2}\\
& f(\lambda) = 1+\mathcal{O}\le(\frac{1}{\lambda}\ri) & \lambda \rightarrow  \infty\, . \label{fconstraint3}
\end{align}
It is easy to check that the function $f(\lb)$ is given by
\begin{equation}
\label{f}
f(\lb)=\exp\le\{\frac{R(\lb)}{2\pi i}\left[\int_{ \Sigma_1}  \frac{\log \frac{1}{r(\zeta)} }{R_+(\zeta)(\zeta-\lambda)} \d \zeta+
\int_{\Sigma_2}  \frac{\log r(\zeta) }{R_+(\zeta)(\zeta-\lambda)} \d \zeta+\int^{\eta_1}_{ -\eta_1}  \frac{\Delta}{R(\zeta)(\zeta-\lambda)} \d \zeta \ri]
\right\} \, .
\end{equation}
\edit{The inclusion of the constant jump \eqref{fconstraint2} allows us to satisfy \eqref{fconstraint3} by taking} 
\begin{align}
\Delta&=\left[\int_{\Sigma_1}  \frac{\log r(\zeta)}{R_+(\zeta)} \d \zeta-
\int_{\Sigma_2}  \frac{\log r(\zeta)}{R_+(\zeta)} \d \zeta\right]\left[\int^{\eta_1}_{ -\eta_1}  \frac{ \d \zeta}{R(\zeta)}\right]^{-1} = -\dfrac{\eta_2}{K(m)}\int^{ \eta_2}_{ \eta_1}  \frac{\log r(\zeta)}{R_+(\zeta)} \d \zeta\, ,
\label{Delta}
\end{align}
where in the last equality  in (\ref{Delta}) we use the fact that $r(-\lb)=r(\lb)$.  We remind the reader that we are assuming the function $r$ to be  real, positive, and non-vanishing on $\Sigma_{1}$ and $\Sigma_{2}$.  
The positivity of $r(\lb)$ guarantees that $\Delta$ is pure imaginary.


\subsection{Opening lenses} \label{TtoS}
We start by defining  the analytic  continuation  $\hat{r}(\lb)$  of the function $r(\lb)$  off the interval $(-\eta_2,-\eta_1)\cup (\eta_1,\eta_2)$  with the requirement that
\begin{equation}\label{rhat}
\hat{r}_{\pm}(\lb)=\pm r(\lb)\, , \quad \lb\in(-\eta_2,-\eta_1)\cup(\eta_1,\eta_2)\, .
\end{equation}
We can factor the jump matrix $J_T$  on $\Sigma_{1}$ as follows
\begin{align*}
&\begin{bmatrix} e^{x\le(g_+(\lb) - g_-(\lb)\ri)}\dfrac{f_+(\lb)}{f_-(\lb)} & 0 \\ -i  & e^{-x\le(g_+(\lb) - g_-(\lb)\ri)}\dfrac{f_-(\lb)}{f_+(\lb)} \end{bmatrix} =\\
&\quad\quad
\begin{bmatrix}1&-\dfrac{ie^{x\le(g_+(\lb) - g_-(\lb)\ri)}}{\hat{r}_-(\lambda)f_-^2 (\lb)}\\ 0 &1 \end{bmatrix} 
\begin{bmatrix}0&-i \\ -i &0\end{bmatrix} 
\begin{bmatrix}1&\dfrac{ie^{-x\le(g_+(\lb) - g_-(\lb)\ri)}}{  \hat{r}_+(\lambda)f_+^2(\lb)}&  \\0&1 \end{bmatrix} 
\end{align*}
and on $\Sigma_2$ as
\begin{align*}
& \begin{bmatrix} e^{x\le(g_+(\lb) - g_-(\lb)\ri)}\dfrac{f_+(\lb)}{f_-(\lb)}&i \\ 0 & e^{-x\le(g_+(\lb) - g_-(\lb)\ri)}\dfrac{f_-(\lb)}{f_+(\lb)} \end{bmatrix}=\\
&\quad\quad
=\begin{bmatrix}1&0\\i\dfrac{f_-^2 (\lb)}{\hat{r}_-(\lambda)}e^{-x\le(g_+(\lb) - g_-(\lb)\ri)} &1 \end{bmatrix} 
\begin{bmatrix}0&i\\i&0\end{bmatrix} 
\begin{bmatrix}1&0\\ -i\dfrac{f_+^2(\lb)}{ \hat{r}_+(\lambda)}e^{x\le(g_+(\lb) - g_-(\lb)\ri)}&  &1 \end{bmatrix}\, .
\end{align*}

%
%
%
%

We can now proceed with ``opening lenses". 
We define a new vector function  $S$  as follows
\begin{gather}
S(\lambda) = \begin{cases}
\displaystyle T(\lambda)\begin{bmatrix}1& \dfrac{-i}{  \hat{r}(\lambda)f^2(\lb)}e^{-2x(g(\lb)-\lb)}&  \\0&1 \end{bmatrix} 
 &\quad \text{in the upper lens, above } \Sigma_{1} \\
\displaystyle T(\lambda)\begin{bmatrix}1&\dfrac{-i}{\hat{r}(\lambda)f^2 (\lb)}e^{-2x\le(g(\lb) - \lb\ri)}\\ 0 &1 \end{bmatrix} &\quad \text{in the lower lens, below } \Sigma_{1} \\
\displaystyle T(\lambda)\begin{bmatrix}1&0\\ i\dfrac{f^2(\lb)}{ \hat{r}(\lambda)}e^{2x\le(g(\lb)-\lb \ri)}&  &1 \end{bmatrix}
&\quad \text{in the upper lens, above } \Sigma_{2} \\
\displaystyle T(\lambda)  \begin{bmatrix}1&0\\ i\dfrac{f^2 (\lb)}{\hat{r}(\lambda)}e^{2x\le( g(\lb)-\lb\ri)} &1 \end{bmatrix}  &\quad \text{in the lower lens, below } \Sigma_{2}\\
\displaystyle T(\lambda) &\quad \text{outside the lenses}\,.
\end{cases} \label{RHPS}
\end{gather}
The vector $S(\lb)$ satisfies 
\begin{equation}
\label{VS}
\begin{split}
S_+(\lambda)&=S_-(\lb)V_S(\lb),\\
S(\lb)&=\begin{bmatrix}1&1 \end{bmatrix} + \mathcal{O}\le(\frac{1}{\lambda}\ri) \qquad \lambda \rightarrow \infty \ .
\end{split}
\end{equation}
where the matrix $V_S$ for the jumps of $S(\lb)$  is depicted in \figurename \ \ref{openinglenses}.
In order to proceed we need the following lemma
\begin{lemma}\label{lemma3.2}
The following inequalities are satisfied
\begin{align}
\label{in1}
&\operatorname{Re} \le(g(\lambda)-\lb\ri) <0\, ,\quad \lb\in {\mathcal C}_1\backslash\{\eta_1,\eta_2\}\\
\label{in2}
&\operatorname{Re} \le(g(\lambda)-\lb \ri) >0\, ,\quad  \lb\in {\mathcal C}_2\backslash\{-\eta_1,-\eta_2\}\, ,
\end{align}
where $\mathcal{C}_1$ and $\mathcal{C}_2$ are the contours defining the lenses as shown in \figurename \ \ref{openinglenses}.
\end{lemma}

\begin{proof}
Given $\lambda = x+iy$, we write $g_+(\lambda)-\lb= u(x,y)+ iv(x,y)$. 
From the formula (\ref{ggg}) for $g$, it follows that $g_+(\lambda)-\lambda $ is purely imaginary on $\Sigma_1 \cup \Sigma_2$; furthermore, for $\lambda \in \Sigma_1$
\begin{gather}
v_x = \operatorname{Im} \le(g'_+(\lambda)-1\ri) = \frac{\lambda^2+\kappa}{\le|R_+(\lambda)\ri|} = \frac{ \Omega}{\le|R_+(\lambda)\ri|} \int_{-\eta_1}^{\eta_1} \frac{ \lambda^2- \zeta^2 }{R(\zeta)} \frac{\d \zeta}{2\pi i}  > 0 \ .
\end{gather}
Using the Cauchy--Riemann equation it follows that $ u_y = - v_x <0$ for $\lambda \in \Sigma_1$ and thus  $\Re \left( g(\lambda)-\lambda \right) <0 $  for $\lb$ above $\Sigma_1$ and $\lambda\in  {\mathcal C}_1$. 
Repeating the same reasoning for the function $g_-(\lambda)-\lambda $   we obtain that  $\Re \left( g(\lambda)-\lambda \right) <0 $  for $\lb$ below $\Sigma_1$ and $\lambda\in  {\mathcal C}_1$. 
In a similar way the inequality  (\ref{in2}) can be obtained.
\end{proof}

Lemma \ref{lemma3.2} guarantees that the off-diagonal entries of the jump matrices along the upper and lower lenses are exponentially small in the regime as $x \to -\infty$, therefore those jump matrices are asymptotically close to the identity  outside a small neighbourhoods of $\pm\eta_1$ and $\pm\eta_2$.
We are left with  the model problem
\begin{gather}
\label{Sinfinity1}
S^{\infty}_+(\lambda) = S^{\infty}_-(\lambda) 
\begin{cases} 
\begin{bmatrix} e^{x\Omega+\Delta} &0 \\ 0 & e^{-x\Omega-\Delta}\end{bmatrix}  & \lambda \in [-\eta_1,\eta_1]  \\
 \begin{bmatrix}0 & -i\\ -i & 0 \end{bmatrix} &\lambda \in \Sigma_{1}  \\
  \begin{bmatrix}0 & i\\ i & 0 \end{bmatrix} &\lambda \in \Sigma_{2}  \\
\end{cases}\\
 S^{\infty}(\lb)=\begin{bmatrix} 1& 1\end{bmatrix}+\mathcal{O}\le(\frac{1}{\lb}\ri),\quad \lb\to\infty\,.
 \label{Sinfinity2}
\end{gather}
The Riemann--Hilbert problem for $S^{\infty}$  has previously appeared in the study of long time asymptotics for KdV with step-like initial data \cite{EGKT}. Below we follow the lines of \cite{EGKT} to obtain the solution.

\begin{figure}
\centering
\scalebox{.75}{
\begin{tikzpicture}[>=stealth]op
\path (0,0) coordinate (O);

\draw[dashed] (-8,0) -- (-7,0);
\draw[dashed] (7,0) -- (8,0);

\node [above] at (6,2) { $ \begin{bmatrix}1& {\color{gray} \dfrac{i}{  \hat{r}(\lambda)f^2(\lb)}e^{-2x(g(\lb)-\lb)} }\\ 0 & 1 \end{bmatrix}$};
\node [below] at (6,-2) { $ \begin{bmatrix}1& {\color{gray} \dfrac{-i}{  \hat{r}(\lambda)f^2(\lb)}e^{-2x(g(\lb)-\lb)} }\\ 0 & 1 \end{bmatrix}$};
\node [below] at (4,-1) { $ {\mathcal C}_1$};
\node [above] at (4.5,0) { $ \begin{bmatrix}0 & -i \\ -i & 0 \end{bmatrix}$};
\node [above] at (5.7,0){$\Sigma_1$};
\node [below] at (0,0) { $ \begin{bmatrix} e^{x\Omega+\Delta} &0 \\ 0 & e^{-x\Omega-\Delta}\end{bmatrix}$};

\node [above] at (-6,2) { $ \begin{bmatrix}1& 0 \\ {\color{gray}-i\dfrac{f^2 (\lb)}{\hat{r}(\lambda)}e^{2x\le( g(\lb)-\lb\ri)} } & 1 \end{bmatrix}$};
\node [below] at (-6,-2) { $ \begin{bmatrix}1& 0 \\ {\color{gray} i\dfrac{f^2 (\lb)}{\hat{r}(\lambda)}e^{2x\le( g(\lb)-\lb\ri)} } & 1 \end{bmatrix}$};
\node [below] at (-4,-1) { ${\mathcal C}_2$};
\node [above] at (-4.5,0) { $ \begin{bmatrix}0 & i \\ i & 0 \end{bmatrix}$};
\node [above] at (-3.2,0) { $\Sigma_2$};

\draw[fill] (2,0) circle [radius=0.05];
\node[above left] at (2,0) {$\eta_1$};
\draw[fill] (7,0) circle [radius=0.05];
\node[above right] at (7,0) {$\eta_2$};

\draw[fill] (-2,0) circle [radius=0.05];
\node[above right] at (-2,0) {$-\eta_1$};
\draw[fill] (-7,0) circle [radius=0.05];
\node[above left] at (-7,0) {$-\eta_2$};

\draw[->-=.7, thick] (-2,0)--(2,0);

\draw[->-=.7,thick] (2,0)--(7,0);
\draw[->-=.7,thick] (-7,0)--(-2,0);

\draw[->- = .7,thick] (2,0) .. controls + (60:2.5cm) and + (120:2.5cm) .. (7,0);
\draw[->- = .7,thick] (7,0) .. controls + (-120:2.5cm) and + (-60:2.5cm) .. (2,0);

\draw[->- = .7,thick] (-7,0) .. controls  + (60:2.5cm) and +  (120:2.5cm) .. (-2,0);
\draw[->- = .7,thick] (-2,0) .. controls   + (-120:2.5cm)  and + (-60:2.5cm)  .. (-7,0);

\end{tikzpicture}
}
\caption{  Riemann--Hilbert problem for $S(\lb)$ defined in (\ref{RHPS}). Opening lenses: the entries in gray in the jump matrices   on the contours ${\mathcal C}_1$ and $ {\mathcal C}_2$ 
are exponentially small in the regime as $x \to -\infty$.}
\label{openinglenses}
\end{figure}
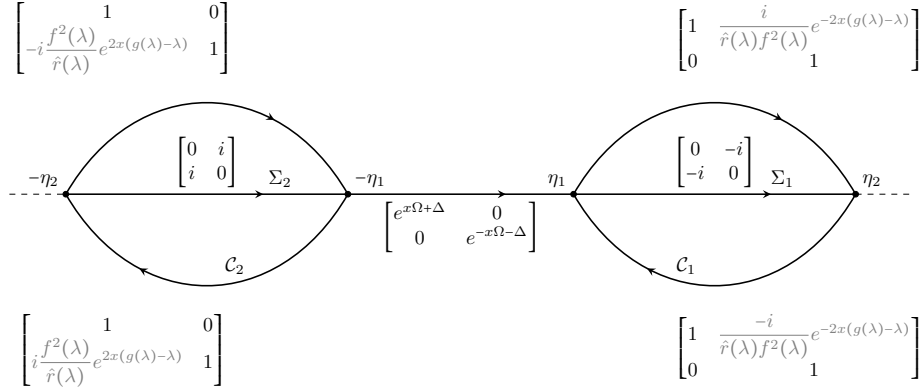

\subsection{The outer parametrix $S^{\infty}$  }
\label{outerparametrix}
To solve the Riemann--Hilbert problem \eqref{Sinfinity1}  and \eqref{Sinfinity2}
we introduce a two-sheeted Riemann surface $\mathfrak{X}$ of genus $1$ associated to  the multivalued function $R(\lambda)$,  namely
\[
\mathfrak{X}=\le\{ (\lambda,\eta)\in\mathbb{C}^2\;|\; \eta^2=R^{2}(\lambda)=(\lambda^2-\eta_1^2)(\lambda^2-\eta_2^2)\ri\} \, .
\]
The first sheet of the surface is identified with the sheet where $R(\lb)$ is real and positive for $\lb\in (\eta_2,+\infty)$.
We introduce  a canonical homology basis  with the $B$ cycle   encircling $\Sigma_1$    clockwise  on the first sheet
 and the $A$ cycle going from $\Sigma_2$   to $\Sigma_1$ on the first sheet and coming back to $\Sigma_2$ on the second sheet.
 The points at infinity on the surface  are denoted by  $\infty^{\pm}$ where  $\infty^+$   is on the first sheet and $\infty^-$ on the second sheet of $\mathfrak{X}$. See \figurename \ \ref{frakX}.
 We introduce the holomorphic differential
 \begin{gather} \label{omegadef} \omega = \frac{\Omega}{R(\lb)} \frac{\d \lb}{4\pi i} \end{gather}
 so that
 \[
\oint_{A}\omega=1 \, .
\]
We also have 
\[
\tau=\oint_{B}\omega=\frac{i}{2}\dfrac{K(\sqrt{1-m^2})}{K(m)}\, ,\quad m=\dfrac{\eta_1}{\eta_2} \, .
\]
Next, we introduce the   Jacobi elliptic function
\begin{gather}
\vartheta_3 (z;\tau)= \sum_{n\in \mathbb{Z}} e^{2\pi i \, nz +  \pi  n^2 i \tau} \ , \qquad z \in \mathbb{C} \ ,
\end{gather}
which is an even function of $z$ and satisfies the periodicity conditions
\begin{equation}
\label{periods}
\vartheta_3 (z+h+k\tau;\tau)=e^{-\pi  i k^2 \tau-2\pi i kz}\vartheta_3 (z;\tau)\, ,\quad h,k\in\Z \, .
\end{equation}
We also  recall that the Jacobi elliptic function   with half-period ratio $\tau$ vanishes  on the half period $\frac{\tau}{2}+\frac{1}{2}$.
Finally, we define the integral
\begin{equation}\label{Abeldef}
w(\lb)=\int_{\eta_2}^{\lambda} \omega
\end{equation}
and we observe that 
\begin{equation}
\label{w_eval}
w(+\infty)=-\frac{1}{4},\quad w_+(\eta_1)=-\frac{\tau}{2},\quad w_+(-\eta_1)=-\frac{\tau}{2}-\frac{1}{2}\, \ ,
\end{equation}
and
\begin{eqnarray} \label{eq:wsymm}
w(-\lambda) = -w(\lambda) - 1/2 \ \mbox{ for } \lambda \in \mathbb{C} \setminus \mathbb{R} \ .
\end{eqnarray}

We introduce the following  functions
\[
\psi_{1}(\lb)=\frac{\vartheta_3 \le(2w(\lambda) +\frac{x\Omega+\Delta}{2\pi i}-\frac{1}{2};2\tau\ri)}{\vartheta_3 \le(2w(\lambda)  -\frac{1}{2};2\tau\ri)}\dfrac{\vartheta_3(0;2\tau)}{\vartheta_3( \frac{x\Omega+\Delta}{2\pi i};2\tau)}\, ,
\]
\[
\psi_{2}(\lb)=\frac{\vartheta_3 \le(-2w(\lambda) +\frac{x\Omega+\Delta}{2\pi i}-\frac{1}{2} ;2\tau\ri)}{\vartheta_3 \le(-2w(\lambda)  -\frac{1}{2};2\tau\ri)}\dfrac{\vartheta_3(0;2\tau)}{\vartheta_3( \frac{x\Omega+\Delta}{2\pi i};2\tau)}\,, 
\]
and we observe that 
\[
\begin{cases}
\vartheta_3 \le(\pm2w_+(\eta_1)  -\frac{1}{2};2\tau\ri)=\vartheta_3 \le(\mp\tau  -\frac{1}{2};2\tau\ri)=0,\\
 \vartheta_3 \le(\pm2w_+(-\eta_1)  -\frac{1}{2};2\tau\ri)=\vartheta_3 \le(\mp\tau\mp 1  -\frac{1}{2};2\tau\ri)=0.
\end{cases}
\]
It follows that the functions $\psi_1$ and $\psi_2$ are analytic except at $\lambda=\pm\eta_1$ where they admit, at most, square root singularities.
Furthermore, the following jump relations are satisfied:
 \begin{align}
 \label{jumpw1}
& w_+(\lambda)-w_-(\lb)=0\quad \lb\in[\eta_2,+\infty)\\
  &w_+(\lambda)+w_-(\lb)=0,\quad \lb\in\Sigma_1\\
   &w_+(\lambda)-w_-(\lb)=-\tau,\quad \lb\in (-\eta_1,\eta_1)\\
    &w_+(\lambda)+w_-(\lb)=-1,\quad \lb\in \Sigma_2 \, .
    \label{jumpw2}
  \end{align}
 Therefore  for $\lb\in\Sigma_1\cup\Sigma_2$  we have 
  \begin{equation}
  \label{psi_rel1}
  \psi_{1+}(\lb)=\psi_{2-}(\lb)\, ,\quad   \psi_{2+}(\lb)=\psi_{1-}(\lb)\, ,
  \end{equation}
  while   for $\lb\in(-\eta_1,\eta_1)$
  \begin{equation}
  \label{psi_rel2}
  \begin{split}
 & \psi_{1+}(\lb)=\psi_{1-}(\lb)e^{x\Omega+\Delta }\, ,\quad   \psi_{2+}(\lb)=\psi_{2-}(\lb)e^{-x\Omega-\Delta  }\,.
 \end{split}
\end{equation}
 Next we introduce   the quantity 
    \[
 \gamma(\lb)=\left(\dfrac{\lb^2-\eta_1^2}{\lb^2-\eta_2^2}\right)^{\frac{1}{4}}\,,
    \]
analytic in $\C \setminus (\Sigma_1 \cup \Sigma_2)$ and normalized such that $\gamma(\lb) \to 1$ as $\lb \to \infty$.    
    Then,
    \begin{equation}
    \label{jump_g}
   \gamma_+(\lb)=-i\gamma_-(\lb)\, , \ \  \text{for } \lb \in\Sigma_1 \quad \text{and} \quad \gamma_+(\lb)=i\gamma_-(\lb)\, ,\ \ \text{for }\in\Sigma_2 \,.
    \end{equation}
  We are now ready to construct the  solution of the Riemann--Hilbert problem \eqref{Sinfinity1}--\eqref{Sinfinity2}.
    \begin{theorem}
  The \edit{vector} $S^{\infty}(\lb)$ given by
  \begin{equation}
  \label{Sinfty_sol}
  \begin{split}
S^{\infty}(\lb)=\gamma(\lb)\dfrac{\vartheta_3(0;2\tau)}{\vartheta_3\le( \frac{x\Omega+\Delta}{2\pi i};2\tau\ri)}
\begin{bmatrix} 
\displaystyle \frac{\vartheta_3 \le(2w(\lambda) +\frac{x\Omega+\Delta}{2\pi i}-\frac{1}{2};2\tau\ri)}{\vartheta_3 \le(2w(\lambda)  -\frac{1}{2};2\tau\ri)} & \displaystyle  \frac{\vartheta_3 \le(-2w(\lambda) +\frac{x\Omega+\Delta}{2\pi i}-\frac{1}{2};2\tau\ri)}{\vartheta_3 \le(-2w(\lambda)  -\frac{1}{2};2\tau\ri)}
\end{bmatrix}
\end{split}
\end{equation}
solves the Riemann--Hilbert problem  \eqref{Sinfinity1}.  
\end{theorem}
\begin{proof}
We  observe that   $S^{\infty}(\lb)$ has at most fourth root singularities at the branch points and it is bounded everywhere else on the complex plane.
 Because of  \eqref{periods} and  \eqref{w_eval} we have $S^{\infty}(\infty)=\begin{bmatrix} 1& 1\end{bmatrix}$, namely the condition \eqref{Sinfinity2} is satisfied.
 Combining \eqref{psi_rel1}, \eqref{psi_rel2} and \eqref{jump_g}, we conclude that the jump conditions \eqref{Sinfinity1} are satisfied. 
\end{proof}

\begin{figure}
\centering
\scalebox{.9}{
\begin{tikzpicture}[>=stealth]
\path (0,0) coordinate (O);

\draw (-4,-1) -- (5,-1);
\draw (-2,1) -- (7,1);
\draw (-4,-1) -- (-2,1); 
\draw (5,-1) -- (7,1);
\node at (5.4,-0.1) {$\times$};
\node[above] at (5.5,-0.1) {$\infty^-$};

\draw (-4, 3) -- (5,3);
\draw (-2,5) -- (7,5);
\draw (-4,3) -- (-2,5);
\draw (5,3) -- (7,5);
\node at (5.4,3.9) {$\times$};
\node[above] at (5.5,3.9) {$\infty^+$};

\draw (-2,0) -- (.5,0);
\draw (1.5,0) -- (4,0);

\draw (-2,4) -- (.5,4);
\draw (1.5,4) -- (4,4);

\draw[dashed, black!30] (-2,0) -- (-2,4);
\draw[dashed, black!30] (.5,0) -- (.5,4);
\draw[dashed, black!30] (1.5,0) -- (1.5,4);
\draw[dashed, black!30] (4,0) -- (4,4);

\draw[fill] (-2,0) circle [radius=0.025];
\node[below ] at (-2,0) {\tiny $-\eta_2$};
\draw[fill] (0.5,0) circle [radius=0.025];
\node[below ] at (0.5,0) {\tiny $-\eta_1$};
\draw[fill] (1.5,0) circle [radius=0.025];
\node[below ] at (1.5,0) {\tiny $\eta_1$};
\draw[fill] (4,0) circle [radius=0.025];
\node[below ] at (4,0) {\tiny $\eta_2$};

\draw[fill] (-2,4) circle [radius=0.025];
\node[above ] at (-2,4) {\tiny $-\eta_2$};
\draw[fill] (0.5,4) circle [radius=0.025];
\node[above ] at (0.5,4) {\tiny $-\eta_1$};
\draw[fill] (1.5,4) circle [radius=0.025];
\node[above ] at (1.5,4) {\tiny $\eta_1$};
\draw[fill] (4,4) circle [radius=0.025];
\node[above ] at (4,4) {\tiny $\eta_2$};



\draw[->- = .25, red] (-1,4) .. controls + (70:.5cm) and + (70:.5cm) .. (2.5,4);
\draw[->- = .25, red] (2.5,0) .. controls + (-110:.5cm) and + (-110:.5cm) .. (-1,0);
\draw[red!30, dashed] (-1,0) -- (-1,4);
\draw[red!30, dashed] (2.5,0) -- (2.5,4);
\node[above, red] at (1,4.3) {\small $A$};

\draw[->- = .7, blue] (1,4) .. controls + (70:.5cm) and + (70:.5cm) .. (4.5,4);
\draw[->- = .7, blue] (4.5,4) .. controls + (-110:.5cm) and + (-110:.5cm) .. (1,4);
\node[blue, above] at (3.25,4.3) {\small $B$};

\end{tikzpicture}
}
\caption{Construction of the genus-$1$ Riemann surface $\mathfrak{X}$ and its basis of cycles.}
\label{frakX}
\end{figure}
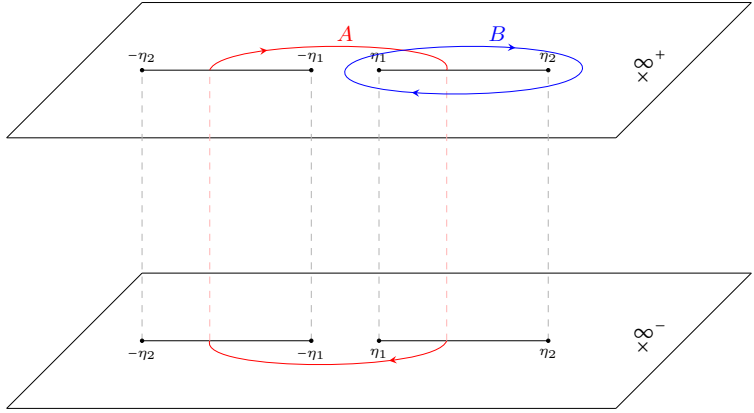

This vector solution provides the asymptotic behaviour of the solution $S$ to Riemann--Hilbert problem depicted in \figurename \ \ref{openinglenses}, for all $\lambda$ bounded away from the endpoints.  However, in order to prove this, we need to construct a {\it matrix solution} to this  Riemann--Hilbert problem, which we call $P^{\infty}(\lb)$. The matrix solution we construct has a pole at $\lb=0$, however this \edit{pole} does not affect the 
vector behaviour of our local and outer parametrices.

  This will be accomplished in the next two subsections, by creating a second, independent vector solution.

\subsection{The outer matrix parametrix $P^{\infty}$\label{Outer}}
\label{Sec_outer}

Consider the  $1$-form $\d p(\lambda)=(1-g'(\lb))\ \d\lb=\dfrac{\lambda^2+\kappa}{R(\lambda)}\d\lambda$ with $\kappa$ as in \eqref{kappag'}  and the Abelian integral
\begin{equation}
\label{p}
p(\lambda)=\int_{\eta_2}^{\lambda}\d p,\
\end{equation}
which satisfies the relations
\begin{equation}
\label{rel1}
p_+(\lambda)+p_-(\lambda)=0\quad \lb\in\Sigma_1\cup\Sigma_2 \ ,
\end{equation}

\begin{equation}
\label{rel2}
p_+(\lambda)-p_-(\lambda)=-\Omega\quad \lb\in(-\eta_1,\eta_1)\ ,
\end{equation}
and for $\lambda \in \mathbb{C} \setminus \mathbb{R}$,
\begin{equation}
\label{eq:psymm}
p(-\lambda) =- p(\lambda) \ .
\end{equation}

Then the  vector function 
\begin{equation}
\label{Psi}
\Psi:= \begin{bmatrix} \varphi_1 & \varphi_2 \end{bmatrix} = \begin{bmatrix} S^{\infty}_1 & S^{\infty}_2 \end{bmatrix} e^{xp(\lambda)\sigma_3}
\end{equation}
solves a  Riemann--Hilbert problem with constant jumps (independent from $x$). 

Indeed on $\Sigma_1\cup\Sigma_2$ we have 
\begin{align*}
\Psi_+(\lb)&=\begin{bmatrix} S^{\infty}_{1+} & S^{\infty}_{2+} \end{bmatrix} e^{xp_+(\lambda)\sigma_3}= \begin{bmatrix} S^{\infty}_{1-} & S^{\infty}_{2-} \end{bmatrix} \begin{bmatrix}0 & \mp i\\ \mp i & 0 \end{bmatrix} e^{xp_+(\lambda)\sigma_3}\\
&=\Psi_-(\lb)e^{-xp_-(\lambda)\sigma_3}\begin{bmatrix}0 & \mp i\\ \mp i & 0 \end{bmatrix} e^{xp_+(\lambda)\sigma_3}\\
&=\Psi_-(\lb)\begin{bmatrix}0 & \mp i\\ \mp i & 0 \end{bmatrix},
\end{align*}
where the $\mp$ signs correspond to $\Sigma_1$ and $\Sigma_2$ respectively, and the last identity has been obtained using \eqref{rel1}.
On $(-\eta_1,\eta_1)$  we have
\begin{align*}
\Psi_+(\lb)&=\begin{bmatrix} S^{\infty}_{1+} & S^{\infty}_{2+} \end{bmatrix} e^{xp_+(\lambda)\sigma_3}= \begin{bmatrix} S^{\infty}_{1-} & S^{\infty}_{2-} \end{bmatrix} e^{(x\Omega +\Delta)\sigma_3} e^{xp_+(\lambda)\sigma_3}\\
&=\Psi_-(\lb)e^{-xp_-(\lambda)\sigma_3} e^{(x\Omega +\Delta)\sigma_3} e^{xp_+(\lambda)\sigma_3}\\
&=\Psi_-(\lb)e^{\Delta\sigma_3},
\end{align*}
 where the last identity has been obtained from \eqref{rel2}.
 Therefore, the $x$ derivative  of $\Psi$ in \eqref{Psi}, namely
  $$\Psi_x(\lb)=\begin{bmatrix} \varphi_{1x} & \varphi_{2x} \end{bmatrix} = \begin{bmatrix} S^{\infty}_{1x} & S^{\infty}_{2x} \end{bmatrix} e^{xp(\lambda)\sigma_3}+ \begin{bmatrix}p(\lb) S^{\infty}_{1} & -p(\lb)S^{\infty}_{2} \end{bmatrix} e^{xp(\lambda)\sigma_3},$$
  has  the same jumps on $(-\eta_2,\eta_2)$ as $\Psi(\lb)$.
For this reason we consider the matrix function \cite{MinakovP} (see also \cite{ClaeysGrava})
\begin{align} 
\Phi(\lb):&=\begin{bmatrix} \varphi_1(\lb)&\varphi_2(\lb)\\
\varphi_{1x}(\lb)&\varphi_{2x}(\lb)
\end{bmatrix}\\
&=\begin{bmatrix}
S^{\infty}_1(\lb)&S^{\infty}_2(\lb)\\
p(\lambda)S^{\infty}_1(\lb)+S^{\infty}_{1x}(\lb)&-p(\lambda)S^{\infty}_2(\lb)+S^{\infty}_{2x}(\lb)
\end{bmatrix}e^{xp(\lambda)\sigma_3}\\
&=\dfrac{1}{2}\begin{bmatrix}
1&1\\
\lambda&-\lambda\end{bmatrix}
\begin{bmatrix}
(1+\frac{p(\lambda)}{\lambda})S^{\infty}_1+\dfrac{1}{\lambda}S^{\infty}_{1x}
&(1-\frac{p(\lambda)}{\lambda})S^{\infty}_2+\dfrac{1}{\lambda}S^{\infty}_{2x}\\
(1-\frac{p(\lambda)}{\lambda})S^{\infty}_1-\dfrac{1}{\lambda}S^{\infty}_{1x}&
(1+\frac{p(\lambda)}{\lambda})S^{\infty}_2-\dfrac{1}{\lambda}S^{\infty}_{2x}
\end{bmatrix}e^{xp(\lambda)\sigma_3}\,,
\end{align}
where the last expression is an algebraic manipulation that can be verified by performing the matrix multiplication.
It follows that for $\lb\neq 0$, the matrix function $\begin{bmatrix}
1&1\\
\lambda&-\lambda\end{bmatrix}^{-1}\Phi(\lb)$  has the same jumps as the vector $\Psi(\lb)$ on the interval $(-\eta_2,\eta_2)$ 
and the matrix function $\begin{bmatrix}
1&1\\
\lambda&-\lambda\end{bmatrix}^{-1}\Phi(\lb)e^{-xp(\lambda)\sigma_3}$ has the same jumps as the vector $S^{\infty}$ defined in \eqref{Sinfty_sol}.
For this reason we take as a matrix solution  for the exterior parametrix
\begin{equation}
\label{P_infinity0}
P^{\infty}(\lb)=\frac{1}{2}\begin{bmatrix}
(1+\frac{p(\lambda)}{\lambda})S^{\infty}_1+\dfrac{1}{\lambda}S^{\infty}_{1x}
&(1-\frac{p(\lambda)}{\lambda})S^{\infty}_2+\dfrac{1}{\lambda}S^{\infty}_{2x}\\
(1-\frac{p(\lambda)}{\lambda})S^{\infty}_1-\dfrac{1}{\lambda}S^{\infty}_{1x}&
(1+\frac{p(\lambda)}{\lambda})S^{\infty}_2-\dfrac{1}{\lambda}S^{\infty}_{2x}
\end{bmatrix}\,.
\end{equation}
It satisfies 
the following Riemann--Hilbert problem:
\begin{gather}
\mbox{$P^{\infty}(\lambda)$  is analytic  for  $\lb\in\mathbb{C}\backslash[-\eta_2,\eta_2]$ with  a singularity at $\lb=0$,}\\
P^{\infty}_+(\lambda) = P^{\infty}_-(\lambda) 
\begin{cases} 
\begin{bmatrix} e^{x\Omega+\Delta} &0 \\ 0 & e^{-x\Omega-\Delta}\end{bmatrix}  & \lambda \in [-\eta_1,\eta_1]  \\
 \begin{bmatrix}0 & -i\\ -i & 0 \end{bmatrix} &\lambda \in \Sigma_{1}  \\
  \begin{bmatrix}0 & i\\ i & 0 \end{bmatrix} &\lambda \in \Sigma_{2},  \\
\end{cases} 
\label{Pinfinityjump}\\
 P^{\infty}(\lb)=\begin{bmatrix} 1& 0\\ 0&1\end{bmatrix}+\mathcal{O}\le(\frac{1}{\lb}\ri),\quad \lb\to\infty\,.
\end{gather}
Despite the singularity of the matrix $P^{\infty}(\lb)$ at $\lb=0$, its determinant is equal to one.
Before proving this fact we first make a slight change of notation that will be relevant in the next sections.
We observe that the $x$-derivative  of $S_1^\infty$ and $S_2^\infty$ can be written in the form
 \[
S^{\infty}_{1x}(\lambda)=\gamma(\lb)\dfrac{\vartheta_3(0;2\tau)}{\vartheta_3 \le(2w(\lambda)  -\frac{1}{2};2\tau\ri)} 
\dfrac{\Omega}{2\pi i}\dfrac{\d}{\d z}\left[ \frac{\vartheta_3 \le(z+2w(\lambda) +\frac{x\Omega+\Delta}{2\pi i}-\frac{1}{2};2\tau\ri)}{\vartheta_3\le(z+ \frac{x\Omega+\Delta}{2\pi i};2\tau\ri)}\right]\bigg|_{z=0}
\]
and similarly for $S^{\infty}_2(\lb)$. Since in the next sections we will use similar formulas where the quantity $\Omega$  is replaced by  $\tilde{\Omega}$ that is 
dependent on $x$ and $t$, it is important to distinguish the operation of derivative with respect to $x$ from the operation on the right hand side of the above expression.
For this reason we introduce the notation 
\[
\nabla_{\Omega}S^{\infty}_1(\lambda):=\gamma(\lb)\dfrac{\vartheta_3(0;2\tau)}{\vartheta_3 \le(2w(\lambda)  -\frac{1}{2};2\tau\ri)} 
\dfrac{\Omega}{2\pi i}\dfrac{\d}{\d z}\left[ \frac{\vartheta_3 \le(z+2w(\lambda) +\frac{x\Omega+\Delta}{2\pi i}-\frac{1}{2};2\tau\ri)}{\vartheta_3\le(z+ \frac{x\Omega+\Delta}{2\pi i};2\tau\ri)}\right]\bigg|_{z=0}
\]
and similarly for $S^{\infty}_2(\lb)$. Clearly, when $\Omega$ is  $x$-independent then $\nabla_{\Omega}S^{\infty}(\lambda)\equiv S_{x}(\lb)$.
Therefore the exterior parametrix $P^{\infty}$ in \eqref{P_infinity0}  will be written  in the form
\begin{equation}
\label{P_infinity}
P^{\infty}(\lb)=\frac{1}{2}\begin{bmatrix}
(1+\frac{p(\lambda)}{\lambda})S^{\infty}_1(\lambda)+\dfrac{1}{\lambda}\nabla_{\Omega}S^{\infty}_1(\lambda)
&(1-\frac{p(\lambda)}{\lambda})S^{\infty}_2(\lambda)+\dfrac{1}{\lambda}\nabla_{\Omega}S^{\infty}_2(\lambda)\\
(1-\frac{p(\lambda)}{\lambda})S^{\infty}_1(\lambda)-\dfrac{1}{\lambda}\nabla_{\Omega}S^{\infty}_1(\lambda)&
(1+\frac{p(\lambda)}{\lambda})S^{\infty}_2(\lambda)-\dfrac{1}{\lambda}\nabla_{\Omega}S^{\infty}_2(\lambda)
\end{bmatrix}\,.
\end{equation}
\begin{lemma}
We have
\begin{equation}
\label{det_Pinfinity}
\det P^{\infty}(\lb) \equiv 1\,.
\end{equation}
\end{lemma}
\begin{proof}
We observe that 
\begin{equation}
\label{detP1}
\det P^{\infty}(\lb)=-\dfrac{1}{2\lb}\Big( -2p(\lb)S^{\infty}_2(\lb)S^{\infty}_1(\lb)+S^{\infty}_1(\lb) \nabla_{\Omega} S^{\infty}_{2}(\lb)-\nabla_{\Omega}S^{\infty}_{1}(\lb)S^{\infty}_2(\lb) \Big)
\end{equation}
does not have any jumps on the complex plane and therefore it is a meromorphic function on the complex plane.
Considering the behaviour near $\lambda=\eta_2$,
we have
\[
S^{\infty}_1(\lb)=\dfrac{1}{(\lb-\eta_2)^{\frac{1}{4}}}(\sum_{k=0}^{\infty}\gamma_k(\lb-\eta_2)^k)\left(\sum_{j=0}^{\infty}S_{j}^{(1)}(\lb-\eta_2)^j+\sqrt{\lb-\eta_2}\sum_{j=0}^{\infty}S_j^{(2)}(\lb-\eta_2)^j\right)
\]
where $\gamma_k$ are the coefficients of the Puiseux expansion of $\gamma(\lb)$ and  $S_{j}^{(1)}$ and $S_{j}^{(2)} $  are  the coefficients of the Puiseux expansion of  the $\vartheta_3 $ function   terms of $S^{\infty}_1(\lb)$   near $\lb=\eta_2$; in particular,   $\gamma_0\neq 0$ and $S_{0}^{(1,2)}\neq 0$.
In a similar way we obtain 
\[
S^{\infty}_2(\lb)=\dfrac{1}{(\lb-\eta_2)^{\frac{1}{4}}}(\sum_{k=0}^{\infty}\gamma_k(\lb-\eta_2)^k)\left(\sum_{j=0}^{\infty}S_{j}^{(1)}(\lb-\eta_2)^j-\sqrt{\lb-\eta_2}\sum_{j=0}^{\infty}S_j^{(2)}(\lb-\eta_2)^j\right)
\]
and
\[
p(\lb)=2\sqrt{\lb-\eta_2}\sum_{k=0}^{\infty}c_k(\lb-\eta_2)^k.
\]
Plugging the above three expansions into \eqref{detP1}  it is straightforward to check  that $\det P^{\infty}(\lb)$ has a Taylor expansion  at the point $\lb=\eta_2$. It can be checked similarly that  $\det P^{\infty}(\lb)$  has a Taylor expansion  at $\lb=\pm\eta_1$ and $\lb=-\eta_2$. 
Regarding the point  $\lb=0$,   we consider the   Abelian integral $p(\lb)$ defined in \eqref{p} and denote by $p_{\pm}(\lb)$ the boundary values of $p(\lb)$ on the real axis.
 We have
\begin{equation}
\label{intp}
p_\pm(0)=\int_{\eta_2}^0 \d p_{\pm}(\xi)=\int_{\eta_2}^{\eta_1}\d p_{\pm}(\xi)+\int_{\eta_1}^0 \d p(\xi)=\mp\frac{\Omega}{2},
\end{equation}
where, in the last relation we use the identity 
\[
0=\int_{-\eta_1}^{\eta_1}\d p(\xi)=\int_{-\eta_1}^{0}\d p(\xi)+\int_{0}^{\eta_1}\d p(\xi)=\int_{\eta_1}^{0}\d p(-\xi)+\int_{0}^{\eta_1}\d p(\xi)=
2\int_{0}^{\eta_1}\d p(\xi).
\]
In this last line we do not use $\d p_\pm(\lambda)$ because the function $R(\lb)$ is analytic off the contours $\Sigma_1$ and $\Sigma_2$ and
the same property holds for $\d p$.
Using the periodicity  properties of the Jacobi elliptic function
\begin{equation*}
\vartheta_3 (z+h+k\tau;\tau)=e^{-\pi  i k^2 \tau-2\pi i kz}\vartheta_3 (z;\tau)\, ,\quad h,k\in\Z \,,
\end{equation*}
we have
\begin{equation}
\label{exp_S}
S^{\infty}_{\pm}(0)=
\gamma(0)\dfrac{\vartheta_3(0;2\tau)}{\vartheta_3 \le( \pm\tau;2\tau\ri)} \frac{\vartheta_3 \le(\pm\tau +\frac{x\Omega+\Delta}{2\pi i};2\tau\ri)}{
\vartheta_3( \frac{x\Omega+\Delta}{2\pi i};2\tau)} 
\begin{bmatrix} 
 \displaystyle  e^{ \pm(x\Omega+\Delta)}&1\end{bmatrix}
\end{equation}
and
\begin{equation}
\label{exp_Sx}
\begin{split}
\nabla_{\Omega}S^{\infty}_{\pm}(0)&=
\dfrac{\Omega}{2\pi i}S^{\infty}_{\pm}(0)\left. \partial_z\left[\log\frac{\vartheta_3 \le(z\pm \tau +\frac{x\Omega+\Delta}{2\pi i};2\tau\ri)}{\vartheta_3 \le(z+\frac{x\Omega+\Delta}{2\pi i};2\tau\ri)
} \right]\right|_{\mathrlap{z=0}} \,
+\begin{bmatrix}\pm\Omega S_{1\pm}^{\infty}(0)\;0\end{bmatrix}.\end{split}
\end{equation}
We conclude that 
\begin{align*}
&\Big( -2p(\lb)S^{\infty}_2(\lb)S^{\infty}_1(\lb)+S^{\infty}_1(\lb)\nabla_{\Omega}S^{\infty}_{2}(\lb)-\nabla_{\Omega}S^{\infty}_{1}(\lb)S^{\infty}_2(\lb) \Big)_{\pm}=\pm\Omega S^{\infty}_{2\pm}(0)S^{\infty}_{1\pm}(0)+ \\
&S^{\infty}_{1\pm}(0) \dfrac{\Omega}{2\pi i}S^{\infty}_{2\pm}(0)\left. \partial_z\left[\log\frac{\vartheta_3 \le(z\pm \tau +\frac{x\Omega+\Delta}{2\pi i};2\tau\ri)}{\vartheta_3 \le(z+\frac{x\Omega+\Delta}{2\pi i};2\tau\ri)
} \right]\right|_{\mathrlap{z=0}}-\dfrac{\Omega}{2\pi i}S^{\infty}_{1\pm}(0)S^{\infty}_{2\pm}(0)\left. \partial_z\left[\log\frac{\vartheta_3 \le(z\pm \tau +\frac{x\Omega+\Delta}{2\pi i};2\tau\ri)}{\vartheta_3 \le(z+\frac{x\Omega+\Delta}{2\pi i};2\tau\ri)
} \right]\right|_{\mathrlap{z=0}}\\
&\mp\Omega S ^{\infty}_{1\pm}(0)S^{\infty}_{2\pm}(0)+O(\lb)=O(\lb)
\end{align*}
as $\lb\to 0 $. 
  Therefore $$\det P^{\infty}(\lb)=-\dfrac{1}{2\lb} \Big( -2p(\lb)S^{\infty}_2(\lb)S^{\infty}_1(\lb)+S^{\infty}_1(\lb)\nabla_{\Omega}S^{\infty}_{2}(\lb)-\nabla_{\Omega}S^{\infty}_{1}(\lb)S^{\infty}_2(\lb) \Big)$$ 
 is a holomorphic function of $\lb$ near $\lb=0$.
 Since
\[
\det P^{\infty}(\lb)=1+O(\lb^{-1}),\quad \mbox{as $\lb\to\infty$,}
\]
it follows  by Liouville's theorem  that 
\[
\det P^{\infty}(\lb) \equiv 1\,.
\]
\end{proof}
{\bf Remark}:  The reader may verify using the definition (\ref{P_infinity}), along with the symmetry relations (\ref{eq:wsymm}) and (\ref{eq:psymm}), that $P^{\infty}$ satisfies the symmetry 
\begin{equation}
\label{eq:PInfSymm}
P^{\infty}(-\lambda) = \pmtwo{0}{1}{1}{0} P^{\infty}(\lambda) \pmtwo{0}{1}{1}{0} \ .
\end{equation}

\subsection{The local parametrix $P^{\pm \eta_j}$ at the endpoints}\label{localparam}

Thanks to Lemma \ref{lemma3.2}, the off-diagonal entries of the jump matrices for $S$ exponentially vanish as $x\to -\infty$ along the upper and lower lenses, while near the endpoints the $g$ function has a square-root-vanishing behaviour 
\begin{gather}
g_+(\lambda) -g_-(\lambda) = \mathcal{O}\le( \sqrt{\lambda\mp \eta_2}\ri) \qquad \text{as } \lambda \rightarrow \pm \eta_2 \ , \ 
\end{gather}
and 
\begin{gather}
g_+(\lambda) -g_-(\lambda) - \Omega = \mathcal{O}\le( \sqrt{\lambda\mp \eta_1}\ri) \qquad \text{as } \lambda \rightarrow \pm \eta_1 \ . \ 
\end{gather}
Additionally, the original solution $Y$ of the Riemann--Hilbert problem \eqref{jumpY}--\eqref{symY} has a logarithmic singularity in those points. Therefore, the jump matrices for $S$ are bounded in a neighbourhood of those points (but they are not close to the identity).

On the other hand, the outer parametrix $P^{\infty}$ is a good approximation of the solution $S$ to the Riemann--Hilbert  problem away from the endpoints $\lambda = \pm \eta_2, \pm \eta_1$, where $P^{\infty}$ exhibits a fourth-root singularity. So, we need to introduce four local parametrices $P^{\pm \eta_j}$ ($j=1,2$) in a suitable neighbourhood of each endpoint.


\subsubsection{Local parametrix near $\lambda = \eta_2$.}
We show here the construction of a (matrix) local parametrix $P^{\eta_2}$ around $\lambda = \eta_2$.

Performing the same calculations as in \cite[Section 6]{KMcLVAV}, we will construct a local parametrix $P^{\eta_2}$ with the help of modified Bessel functions. We fix a small disc $B^{(\eta_2)}_{\rho} = \le\{ \lambda \in \mathbb{C} \le| \,  \le|\lambda - \eta_2\ri|< \rho \ri.  \ri\}$ centered at $\eta_2$ of radius $\rho$, and we define the (local) conformal map
\begin{gather}
\zeta = \frac{1}{4} \le[ x \le(g(\lambda) - \lambda \ri)\ri]^2, \qquad \lambda \in B^{(\eta_2)}_\rho \ . 
\label{conformalBessel}
\end{gather}

To define the local parametrix $P^{\eta_{2}}$ in $B^{(\eta_2)}_\rho$, we consider 
\begin{eqnarray*}
P(\lambda) = S(\lambda) \le(\frac{ e^{  i \pi  /4}}{\sqrt{\pm \hat{r}} f}
\right)^{\sigma_{3}} & \lambda \in B^{(\eta_2)}_{\rho} \cap \mathbb{C}_{\pm}
\ ,
\end{eqnarray*}
and then, using the inverse of the transformation $\zeta(\lambda)$, we define
\begin{gather*}
P^{(1)}(\zeta) = P(\lambda(\zeta)) e^{-2 \zeta^{\frac{1}{2}} \sigma_3} \begin{bmatrix}0&1\\1&0 \end{bmatrix}\ , \quad \zeta \in \C \, ,
\end{gather*}
with branch cut $(-\infty,0]$. By construction, $P^{(1)}$ satisfies a Riemann--Hilbert  problem with jumps
\begin{gather}
P^{(1)}_+(\zeta) = P^{(1)}_-(\zeta) \begin{cases} \begin{bmatrix} 1 & 0 \\ 1 & 1 \end{bmatrix} & \text{on $\{$upper and lower lenses$\}$}\cap B^{(\eta_2)}_\rho \\
\begin{bmatrix} 0 & 1 \\ -1\ & 0 \end{bmatrix} &  \text{on } (-\infty,0] \cap B^{(\eta_2)}_\rho\ .
\end{cases}
\end{gather}

 We introduce now the model parametrix $\Psi_\Bes(\zeta)$ as in \cite[formul\ae \ (6.16)--(6.20)]{KMcLVAV}). The Riemann--Hilbert  problem for $\Psi_\Bes$ is the following:
\begin{enumerate}[(a)]
\item $\Psi_\Bes$ is analytic for $\zeta \in \mathbb{C} \backslash \Gamma_{\Psi}$, where $\Gamma_\Psi$ is the union of the three contours $\Gamma_\pm = \le\{ \arg \zeta = \pm\frac{2\pi }{3}\ri\} $ and $\Gamma_0 = \le\{ \arg \zeta = \pi \ri\}$;
\item $\Psi$ satisfies the following jump relations
\begin{gather}
\Psi_{\Bes \, +}(\zeta) =  \Psi_{\Bes\, -}(\zeta) \begin{cases} \begin{bmatrix} 1 & 0 \\ 1 & 1 \end{bmatrix} & \text{on } \Gamma_+ \cup \Gamma_-\\
\begin{bmatrix} 0 & 1 \\ -1\ & 0 \end{bmatrix} & \text{on }\Gamma_0, 
 \end{cases}\,  
\end{gather}
\item as $\zeta \rightarrow 0$
\begin{gather}
\Psi_\Bes(\zeta) =  \begin{bmatrix} \mathcal{O}\le( \ln |\zeta| \ri) & \mathcal{O}\le( \ln |\zeta| \ri) \\ \mathcal{O} \le( \ln|\zeta| \ri)& \mathcal{O} \le(\ln |\zeta|\ri)  \end{bmatrix} \ .
\end{gather}
\end{enumerate}

The solution is the following
\begin{gather}
\Psi_\Bes(\zeta) =
\begin{cases}
 \begin{bmatrix} \displaystyle I_0 (2\zeta^{\frac{1}{2}}) &\displaystyle \frac{i}{\pi} K_0 (2\zeta^{\frac{1}{2}})\\
\displaystyle 2\pi i \zeta^{\frac{1}{2}}I'_0 (2\zeta^{\frac{1}{2}}) & \displaystyle - 2\zeta^{\frac{1}{2}}K_0(2\zeta^{\frac{1}{2}})
\end{bmatrix} & |\arg \zeta| < \frac{2\pi }{3} \\
& \\
 \begin{bmatrix} \displaystyle \frac{1}{2}H^{(1)}_0 (2(-\zeta)^{\frac{1}{2}}) &\displaystyle \frac{1}{2} H^{(2)}_0 (2(-\zeta)^{\frac{1}{2}})\\
\displaystyle \pi \zeta^{\frac{1}{2}} \le[H^{(1)}_0 (2(-\zeta)^{\frac{1}{2}})\ri]' & \displaystyle  \pi \zeta^{\frac{1}{2}} \le[H^{(2)}_0(2(-\zeta)^{\frac{1}{2}})\ri]'
\end{bmatrix}  &    \frac{2\pi }{3} < |\arg \zeta| < \pi \\
& \\
\begin{bmatrix} \displaystyle \frac{1}{2}H^{(2)}_0 (2(-\zeta)^{\frac{1}{2}}) &\displaystyle -\frac{1}{2} H^{(1)}_0 (2(-\zeta)^{\frac{1}{2}})\\
\displaystyle -\pi \zeta^{\frac{1}{2}} \le[H^{(2)}_0 (2(-\zeta)^{\frac{1}{2}})\ri]' & \displaystyle  \pi \zeta^{\frac{1}{2}} \le[H^{(1)}_0(2(-\zeta)^{\frac{1}{2}})\ri]'
\end{bmatrix}  &  -\pi < |\arg \zeta|< - \frac{2\pi }{3}
\end{cases} \label{Besparamsol}
\end{gather}
with asymptotic behaviour at infinity 
\begin{gather}
\Psi_\Bes(\zeta) = \le( 2\pi \zeta^{\frac{1}{2}}\ri)^{-\frac{1}{2}\sigma_3} \frac{1}{\sqrt{2}} \begin{bmatrix}
\displaystyle 1& i  \displaystyle \\
\displaystyle i& \displaystyle 1
\end{bmatrix} \le(I +  \mathcal{O} \le(\frac{1}{\zeta^{\frac{1}{2}}}\ri)\ri)e^{2\zeta^{\frac{1}{2}}\sigma_3}
\end{gather}
uniformly as $\zeta \rightarrow \infty$ everywhere in the complex plane aside from the jumps.

In the above formul\ae \  $I_0(\zeta)$, $K_0(\zeta)$  are the modified Bessel functions of first and second kind, respectively, and $H^{(j)}(\zeta)$ the Hankel functions.

In conclusion, the local parametrix around the endpoint $\lambda = \eta_2$ is
\begin{gather}
 P^{\eta_{2}}(\lambda) = 
 A(\lambda) \Psi_\Bes(\zeta(\lambda)) \begin{bmatrix} 0&1\\1&0\end{bmatrix} e^{2\zeta(\lambda)^{\frac{1}{2}}\sigma_3} \left(
\frac{ e^{ i \pi  /4}}{\sqrt{\pm \hat{r}(\lb)} f(\lb)}
\right)^{-\sigma_{3}}  \quad  \lambda \in B^{(\eta_2)}_{\rho} \cap \mathbb{C}_{\pm}
\ ,
\end{gather}
where $A$ is a prefactor that is  determined by imposing that 
\begin{gather}
P^{\eta_{2}}(\lambda)  \le(P^{\infty}(\lambda) \ri)^{-1} = I + \mathcal{O}\le(|x|^{-1}\ri)  \qquad \text{as }  x \to - \infty, \ \text{for } \lambda  \in \partial B^{(\eta_2)}_\rho \backslash \Sigma_\Psi \ .
\end{gather}
Therefore, we set
\begin{gather}
\label{eq:Adef3}
A(\lambda) 
= P^{\infty}(\lambda)   \left(
\frac{ e^{ i \pi  /4}}{\sqrt{\pm \hat{r}(\lb)} f(\lb)}
\right)^{\sigma_{3}} 
\frac{1}{\sqrt{2}} \begin{bmatrix}-i& 1 \\ 1  &  -i\end{bmatrix} \le(2\pi \zeta^{\frac{1}{2}}\ri)^{\frac{1}{2}\sigma_3} \quad  \lambda \in B^{(\eta_2)}_{\rho} \cap \mathbb{C}_{\pm}\ .
\end{gather}
By construction, $A$ is well-defined and analytic in a neighbourhood of $\eta_2$, minus the cut $(-\infty, \eta_2]$; additionally, it is easy to see that $A$ is invertible ($\det A(\lambda) \equiv 1 $).

\begin{lemma}
$A(\lambda)$ is analytic everywhere in the neighbourhood $B^{(\eta_2)}_\rho$ of $\eta_2$.
\end{lemma}

\begin{proof}
To prove the statement, one needs to check that $A$ has no jumps across the interval $\edit{\Sigma_1} \cap B^{(\eta_2)}_{\rho}$ and that it has at most a removable singularity at $\lambda = \eta_2$.  
\edit{
Starting from \eqref{eq:Adef3} we observe from \eqref{fconstraint1} and \eqref{rhat} that for $\lambda \in \Sigma_1$, $\sqrt{ \hat{r}(\lambda)} f_+(\lambda) = \left( \sqrt{-\hat{r}(\lambda)} f_-(\lambda) \right)^{-1}$. Using this and the jump \eqref{Pinfinityjump} of $P^\infty$ on $\Sigma_1$ we have
\begin{align*}
	A_+(\lambda) 
	&= 
	P_-^{\infty}(\lambda) \begin{bmatrix} 0&-i\\-i&0\end{bmatrix}  \left(\frac{ e^{ i \pi  /4}}{\sqrt{\hat{r}_-(\lb)} f_-(\lb)} \right)^{-\sigma_{3}} 
	i^{\sigma_3}\frac{1}{\sqrt{2}} \begin{bmatrix}-i& 1 \\ 1  &  -i\end{bmatrix} i^{\sigma_3} \le(2\pi \zeta_-^{\frac{1}{2}}\ri)^{\frac{1}{2}\sigma_3}  \\
	&= 
	P_-^{\infty}(\lambda) \left(\frac{ e^{ i \pi  /4}}{\sqrt{\hat{r}_-(\lb)} f_-(\lb)} \right)^{\sigma_{3}} \begin{bmatrix} 0&-i\\-i&0\end{bmatrix}
	i^{\sigma_3}\frac{1}{\sqrt{2}} \begin{bmatrix}-i& 1 \\ 1  &  -i\end{bmatrix} i^{\sigma_3} \le(2\pi \zeta_-^{\frac{1}{2}}\ri)^{\frac{1}{2}\sigma_3}  \\
	&= 
	P_-^{\infty}(\lambda) \left(\frac{ e^{ i \pi  /4}}{\sqrt{\hat{r}_-(\lb)} f_-(\lb)} \right)^{\sigma_{3}} 
	\frac{1}{\sqrt{2}} \begin{bmatrix}-i& 1 \\ 1  &  -i\end{bmatrix}  \le(2\pi \zeta_-^{\frac{1}{2}}\ri)^{\frac{1}{2}\sigma_3}  = A_-(\lambda)\,.
\end{align*}
}
Next, we notice that $\zeta(\lambda)$ has a simple zero at $\eta_2$ by construction, thus $\zeta(\lambda)^{\frac{1}{4}\sigma_3}$ has at most a fourth-root singularity at the point $\lambda = \eta_2$.
Also the outer parametrix $P^{\infty}(\lambda)$ has at most a fourth-root singularity near $\eta_2$ and consequently all the entries of $A(\lambda)$ have at most a square root singularity at $\lambda = \eta_2$. 

On the other hand $A(\lambda)$ is analytic in $B^{(\eta_2)}_\rho \backslash \{\eta_2\}$, therefore the point $\lambda = \eta_2$ is a removable singularity and $A(\lambda)$ is indeed analytic everywhere in $B^{(\eta_2)}_\rho$.
\end{proof}

\subsubsection{Local parametrix near other branch points.}
\label{sec:3.6}
The construction of the parametrix in a vicinity $B^{(\eta_1)}_\rho$ of $\eta_{1}$ is quite similar, and has also been carried out in \cite[Section 6]{KMcLVAV}, so we will not present the formula here. 

For the parametrices near $- \eta_{2}$ and $- \eta_{1}$, it will prove convenient to construct them {\it explicitly} via the underlying $\lambda \mapsto - \lambda$ symmetry, as follows:
\begin{eqnarray}
\label{eq:LocalDef2}
&&P^{-\eta_{2}} := \pmtwo{0}{1}{1}{0} P^{\eta_{2}}(-\lambda) \pmtwo{0}{1}{1}{0}, \\
\label{eq:LocalDef1}
&&
P^{-\eta_{1}} := \pmtwo{0}{1}{1}{0} P^{\eta_{1}}(-\lambda) \pmtwo{0}{1}{1}{0}.
\end{eqnarray}

First, the reader may verify that, if $P^{\eta_{j}}$ satisfies the appropriate jump relationships along the contours within the disc centered at $\eta_{j}$, then $P^{-\eta_{j}}$ satisfies the appropriate jump relationships along the contours within the disc (of the same radius) centered at $- \eta_{j}$.  Along the way, the following symmetry relations are needed (and are easy to establish) for $ \lb\in \C\backslash(-\eta_2,\eta_2)$,
\begin{eqnarray}
&&\hat{r}(-\lambda) = \hat{r}(\lambda), \\
&&f^{2}(-\lambda) = f^{-2}(-\lambda) ,\\
&& g(-\lambda) = - g(\lambda) \ .
\end{eqnarray}
Moreover, since $P^{\eta_{j}}$ has been constructed to satisfy 
\begin{eqnarray*}
P^{\eta_{j}}(\lambda)P^{\infty}(\lambda)^{-1} = I + \mathcal{O} \left( \frac{1}{x} \right),  \quad  \mbox{as $x\to-\infty$},
\end{eqnarray*}
for $\lambda$ on the boundary of $B_{\rho}^{(\eta_{j})}$ (the small disc of radius $\rho$ centered at $\eta_{j}$), it follows that $P^{-\eta_{j}}$ satisfies
\begin{eqnarray*}
P^{-\eta_{j}}(\lambda)P^{\infty}(\lambda)^{-1} = I + \mathcal{O} \left( \frac{1}{x} \right),  \quad  \mbox{as $x\to-\infty$},
\end{eqnarray*}
for $\lambda$ on the boundary of an analogous small disc  $B^{(-\eta_j)}_{\rho}$ centered at $-\eta_{j}$.  

\subsection{Small norm argument and determination of $u(x,0)$ for large  negative  $x$}
\label{smaxinfty}


Define the error vector 
\begin{gather}
	\mathcal{E}(\lb) = S(\lb) \left( P(\lb) \right)^{-1}
	\label{error_def}
\shortintertext{where the global parametrix $P(\lb)$ is defined by}
	P(\lb) = \begin{cases}
		P^{\infty}(\lb)  & \lb\in\mathbb{C}\backslash\cup_{j=1,2}B^{(\pm \eta_j)}_{\rho}  \\
		P^{\eta_{2}}(\lb) & \lb\in B^{(\eta_2)}_{\rho} \\
		P^{\eta_{1}} (\lb) & \lb\in B^{(\eta_1)}_{\rho} \\
		P^{-\eta_{1}}(\lb) & \lb\in B^{(-\eta_1)}_{\rho} \\
		P^{-\eta_{2}}(\lb) & \lb\in B^{(-\eta_2)}_{\rho} \ .
	\end{cases}
	\label{global_p}
\end{gather}
 \begin{figure}[h!]
\centering
\scalebox{.8}{
\begin{tikzpicture}[>=stealth]
\path (0,0) coordinate (O);


%
%

\draw[->- = .7,thick] (2,0) .. controls + (60:2.cm) and + (120:2.cm) .. (7.,0);
\draw[->- = .7,thick] (2,-0) .. controls + (-60:2.cm) and + (-120:2.cm) .. (7,-0);

\draw[->- = .7,thick] (-2,0) .. controls + (120:2.cm) and + (60:2.cm) .. (-7,0);
\draw[->- = .7,thick] (-2,-0) .. controls + (-120:2.cm) and + (-60:2.cm) .. (-7,-0);

\draw[->-=.7,thick, fill=white] (2,0) circle [radius=1.];
\node  at (2,0.4) { $B^{(\eta_1)}_{\rho}$};
\node  at (2,-0.2) {$\eta_1$};
\node  at (2,0) {$\bullet$};
\draw[->-=.7,thick, fill=white] (7,0) circle [radius=1.0];
\node  at (7,0.4) { $B^{(\eta_2)}_{\rho}$};
\node  at (7,-0.2) {$\eta_2$};
\node  at (7,0) {$\bullet$};
\draw[->-=.7,thick, fill=white] (-2,0) circle [radius=1.];
\node  at (-2,0.4) { $B^{(-\eta_1)}_{\rho}$};
\node  at (-2,-0.2) {$-\eta_1$};
\node  at (-2,0) {$\bullet$};

\draw[->-=.7,thick, fill=white] (-7,0) circle [radius=1];
\node  at (-7,0.4) { $B^{(-\eta_2)}_{\rho}$};
\node  at (-7,-0.2) {$-\eta_2$};
\node  at (-7,0) {$\bullet$};

\node  at (4.5,0) { $ \begin{bmatrix}1& \mathcal{O}(e^{-cx}) \\ 0 & 1 \end{bmatrix}$};
\draw[ ->, dotted, thick] (4.5,.5) -- (4.5,1.2);
\draw[ ->, dotted, thick] (4.5,-.5) -- (4.5,-1.2);

\node  at (-4.5,0) { $ \begin{bmatrix}1& 0 \\ \mathcal{O}(e^{-cx})  & 1 \end{bmatrix}$};
\draw[ ->, dotted, thick] (-4.5,.5) -- (-4.5,1.2);
\draw[ ->, dotted, thick] (-4.5,-.5) -- (-4.5,-1.2);
\node [below] at (4.6,-1.4) {${\mathcal C}_1$};
\node [below] at (4.6, 1.8) {${\mathcal C}_1$};
\node [below] at (-4.6,-1.4) {${\mathcal C}_2$};
\node [below] at (-4.6, 1.8) {${\mathcal C}_2$};
\node [above] at (0,2) { $ \begin{bmatrix}1 + \mathcal{O}(x^{-1}) &  \mathcal{O}(x^{-1})  \\  \mathcal{O}(x^{-1})  & 1+\mathcal{O}(x^{-1})  \end{bmatrix}$};
\draw[ ->, dotted, thick] (0,2) -- (-1,.5);
\draw[ ->, dotted, thick] (0,2) -- (1,.5);
\draw[ ->, dotted, thick] (2,2.6) -- (7,2.6) -- (7,1.1);
\draw[ ->, dotted, thick] (-2,2.6) -- (-7,2.6) -- (-7,1.1);


\end{tikzpicture}
}
\caption{The Riemann--Hilbert problem for the remainder $ \mathcal E $.}
\label{RHPremainder}
\end{figure}
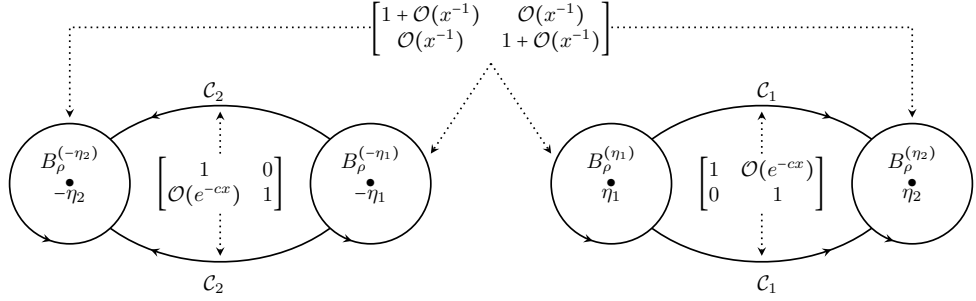

Then across any contour where either $S$ is non-analytic or any boundary in the definition of $P$ the
matrix $\mathcal{E}$ has a jump given by
\begin{gather} 
 	\mathcal E_{+} (\lb) =  \mathcal E_{-} (\lb) 	V_{\mathcal E}(\lb)
 \shortintertext{with}
   \begin{multlined}[.8\textwidth]	
 	V_{\mathcal E}(\lb) = \left( \mathcal E_{-} (\lb) \right)^{-1} \mathcal E_{+}  (\lb) 
		= P_-(\lb) \left(S_-(\lb) \right)^{-1} S_+(\lb) \left( P_+(\lb) \right)^{-1}  \\
		= P_-(\lb) V_S(\lb) \left( V_P(\lb) \right)^{-1} \left( P_-(\lb)  \right)^{-1} \, 
  \end{multlined}
\label{VE}
\end{gather}
where the jump matrix $V_S$ is as defined in Figure~\ref{openinglenses}, $V_P$ is the jump of $P$, and both jumps are understood to be the identity matrix anywhere $S$ or $P$ is analytic respectively.
Observe that both within the discs  $B^{(\pm\eta_j)}_{\rho}$,  and across any component of $(-\eta_{2}, \eta_{2})$ outside the discs  $B^{(\pm \eta_j)}_\rho$, $j=1,2$, the quantities $V_{S}$ and $V_P$ coincide, and hence $\mathcal E$ has no jump across those contours.
Across the lens boundaries (outside the discs) we have $V_{P} = I$, and hence 
\begin{align}
\label{VE1}
&V_{\mathcal E}(\lb)=\left(P^{\infty}(\lb) \right) V_{S}(\lb) \left( P^{\infty} (\lb)\right)^{-1} = \left( I + \mathcal{O} \left( e^{- c x } \right) \right), 
&&\lb\in {\mathcal C}_j,\, j=1,2,
\end{align}
while across the circles centered at $\pm \eta_{j}$ (which we have chosen to orient counter-clockwise), we have
\begin{align}
\label{VE2}
&V_{\mathcal E}(\lb)=\left(P^{\infty}(\lambda) \right)^{-1} P^{\pm \eta_{j}}(\lambda) =  \left( I + \mathcal{O} \left( x^{-1} \right) \right),
&& \lb\in\partial B^{(\pm \eta_j)}_\rho,\;j=1,2.
\end{align}

Finally, since $P = P^\infty(\lb)$ near $\lb =0$ and $P^\infty(\lb)$ is singular there, we need to check the behaviour of $\mathcal{E}$  at $\lb = 0$. 
\begin{lemma}
\label{lemma4.1}
The error vector  $\mathcal{E}$ defined by \eqref{error_def}-\eqref{global_p} 
is regular at $\lb=0$.
\end{lemma}
\begin{proof}
Near $\lb=0$, $\mathcal{E}(\lb) = S(\lb) \left( P^{\infty}(\lb) \right)^{-1}$.  We have  $S(\lb)=Y(\lb)e^{xg(\lb)\sigma_3}f(\lb)^{\sigma_3}$   with the functions $g(\lb)$  and $f(\lb)$ defined in \eqref{ggg} and \eqref{f} respectively and where $Y(\lb)$ is the solution of the Riemann--Hilbert problem defined by \eqref{jumpY}-\eqref{symY} whose existence is established in the Appendix.
We  need to prove that
\[
Y(\lb)e^{xg(\lb)\sigma_3}f(\lb)^{\sigma_3}(P^{\infty}(\lb))^{-1}
\]
is regular at $\lb=0$ where  $Y(\lb)$ satisfies the symmetry \eqref{symY} so that $Y_1(0)=Y_2(0)$.
We observe that    $g(\lb)=\lb-p(\lb)$  so that by \eqref{intp} 
$$g_\pm(0)=\pm \frac{\Omega}{2}.$$ We conclude that 
\[
e^{xg_{\pm}(\lb)\sigma_3}=e^{\pm\frac{x\Omega}{2}\sigma_3}(1+O(\lb))\quad \mbox{as $\lb\to 0$}.
\]
In a similar way it can be proved that  that $f_{\pm}(\lb)=e^{\pm\frac{\Delta}{2}}(1+O(\lb))$ as $\lb\to0$.
Using  the above expansion we have that as $\lb\to 0_+$ 
\begin{equation}
\label{eql1}
\begin{split}
&Y_+(\lb)e^{xg_+(\lb)\sigma_3}f_+(\lb)^{\sigma_3}(P_+^{\infty}(\lb))^{-1}=-\dfrac{Y_1(0)}{2\lb} \begin{bmatrix} e^{\frac{x\Omega+\Delta}{2}} & e^{-\frac{x\Omega+\Delta}{2}} \end{bmatrix} \times\\
&\quad \times\left(
\begin{bmatrix}
p_+(0)S^{\infty}_{2+}(0)-\nabla_{\Omega}S^{\infty}_{2+}(0)
&p_+(0)S^{\infty}_{2+}(0)-\nabla_{\Omega}S^{\infty}_{2+}(0)\\
p_+(0)S^{\infty}_{1+}(0)+\nabla_{\Omega}S^{\infty}_{1+}(0)&
p_+(0)S^{\infty}_{1+}(0)+\nabla_{\Omega}S^{\infty}_{1+}(0)
\end{bmatrix}
 +O(\lb)\right)\,.
 \end{split}
\end{equation}
Using  the relations \eqref{exp_S} and \eqref{exp_Sx} we  obtain 
\[
Y_+(\lb)e^{xg_+(\lb)\sigma_3}f_+(\lb)^{\sigma_3}(P_+^{\infty}(\lb))^{-1}=O(1),\quad \mbox{as $\lb\to0_+ $}.
\]
In a similar way it can be verified  the   regular behaviour at $0_-$ which concludes the  proof of the lemma.
\end{proof}

Let $\Sigma_\mathcal{E}$ be the system of contours shown in Figure~\ref{RHPremainder}. The above arguments show that the error vector $\mathcal{E}$ satisfies the following Riemann--Hilbert problem
\begin{align*}
	{\mathcal E}_+(\lb)={\mathcal E}_-(\lb) V_\mathcal{E}(\lb)
	\qquad
	\lb \in \Sigma_\mathcal{E}
\end{align*}	
and  as $\lb\to\infty$
\begin{eqnarray}
\mathcal E(\lb) = \begin{bmatrix}1& 1 \end{bmatrix} + \mathcal{O} \left( \lambda^{-1} \right) \ .
\end{eqnarray}
where the jump matrix $V_{\mathcal{E}}$ satisfies
 \begin{gather}
 V_\mathcal{E} (\lambda) = \begin{cases}
	I + \mathcal{O} \le( e^{-c\edit{|x|}} \ri)  & \lb \in \mathcal{C}_j, \ j=1,2, \\
	I + \mathcal{O}\le( |x|^{-1} \ri)  & \lb \in \partial B_\rho^{\pm \eta_j}, \ j=1,2. 
 \end{cases} 
 \end{gather}

Therefore, by a standard small norm argument (see, for example \cite[Section 5.1.3]{smallnormRH}) there exists a unique solution $\mathcal E$, which possesses an asymptotic expansion for large \edit{negative} $x$ and large $\lambda$ of the form:  
\begin{eqnarray}
\label{eq:errorexpand}
\mathcal E (\lambda) = \begin{bmatrix} 1 & 1 \end{bmatrix} \  + \ \frac{\mathcal{E}_{1}(x)}{x \lambda} \ + \mathcal{O}\left( \frac{1}{\lambda^{2}} \right) \ , 
\end{eqnarray}
where $\mathcal{E}_{1}(x)$ possesses bounded derivatives in $x$.  

We note in passing that the construction of a matrix-valued global approximation is very useful, in that we arrive directly at a small-norm Riemann--Hilbert problem.  

We also notice that the solution $\mathcal E$ which we have constructed obeys the symmetry
\begin{eqnarray}
\label{eq:Esymm}
\mathcal E(-\lambda) = \mathcal E(\lambda) \pmtwo{0}{1}{1}{0} \ .
\end{eqnarray}
Indeed, the jump matrices $V_{\mathcal E}$ for $\mathcal E$ all satisfy the symmetry
\begin{eqnarray}
\label{eq:Vsymm}
V_{\mathcal E}(-\lambda) = \pmtwo{0}{1}{1}{0}
V_{\mathcal E}(\lambda)
\pmtwo{0}{1}{1}{0}\ 
\end{eqnarray}
where $V_{\mathcal{E}}$ is given in \eqref{VE1} and \eqref{VE2}.
Properly, to see this, one must ensure that the contours for the Riemann-Hilbert problem for $\mathcal E$ are symmetric with respect to the mapping $\lambda \mapsto - \lambda$, and then verify that $V_{\mathcal E}$ satisfies (\ref{eq:Vsymm}).  We have already specified in Subsection \ref{sec:3.6} that the circular contours should possess this symmetry, and it is clear that the lens boundaries may be chosen to satisfy this symmetry. 

The verification of (\ref{eq:Vsymm}) for $\lambda$ in any of the four circles follows from the definitions (\ref{eq:LocalDef2}) and (\ref{eq:LocalDef1}).  The verification of (\ref{eq:Vsymm}) for $\lambda$ in any of the lens boundaries follows by inspection of the jump matrices for $S$ (only those defined on the lens boundaries) as described in  \figurename \ \ref{openinglenses}, and using (\ref{eq:PInfSymm}).  The fact that (\ref{eq:Vsymm}) implies (\ref{eq:Esymm}) is a straighforward exercise from the theory of Riemann-Hilbert problems.

Because $\mathcal E$ is analytic in a vicinity of $\lambda=0$, the symmetry relation (\ref{eq:Esymm}) implies that $\mathcal{E}(\lambda)$ has the expansion
\begin{eqnarray}
\label{eq:ELocal0}
\mathcal{E}(\lambda) = \hat{c}_{0} \begin{bmatrix} 1 & 1\end{bmatrix} + \lb \hat{c}_{1} \begin{bmatrix} 1 & -1\end{bmatrix} \ , \mbox{ for } \lambda \ \mbox{ near } \ 0 , 
\end{eqnarray}
for some constants $\hat{c}_{0}$ and $\hat{c}_{1}$.


Taking into account all the transformations we performed, we are now able to explicitly solve the original Riemann--Hilbert  problem $Y$ in the large negative $x$ regime:
\begin{equation}
\label{E_infinity}
\begin{split}
Y(\lambda) = T(\lambda) e^{-xg(\lambda) \sigma_3}f(\lb)^{-\sigma_3} = S(\lambda) e^{-xg(\lambda) \sigma_3}f(\lb)^{-\sigma_3}  \\
=\mathcal E(\lb) P(\lambda) e^{-xg(\lambda) \sigma_3}f(\lb)^{-\sigma_3} =\le( \begin{bmatrix} 1& 1\end{bmatrix} + \frac{\mathcal{E}_{1}(x)}{x \lambda } \ + \mathcal{O}\left( \frac{1}{\lambda^{2}} \right)  \ri) P(\lambda) e^{-xg(\lambda) \sigma_3}f(\lb)^{-\sigma_3}  \ ,
\end{split}
\end{equation}
where $P(\lb)$ is the global parametrix defined by \eqref{global_p}.

In particular, for $\lambda$ near $0$, $\mathcal{E}(\lambda) P(\lambda)$ appearing in (\ref{E_infinity}) is actually $\mathcal{E} (\lambda) P^{\infty}(\lambda)$.  The reader will recall that $P^{\infty}$ has a pole at $\lambda = 0$ (see (\ref{P_infinity})).  However, because of the behavior of $\mathcal{E}(\lambda)$ for $\lambda$ near $0$ shown in (\ref{eq:ELocal0}), the product $\mathcal{E}(\lambda) P^{\infty}(\lambda)$ has no pole at $\lambda=0$.

We recall that the potential $u(x)$ can be calculated from the solution $Y(\lb)$ as  
\begin{gather}
\label{uexp0}
u(x) = 2\frac{\d}{\d x} \le[ \lim_{\lambda\rightarrow\infty} \lambda(Y_1(\lambda;x) -1) \ri] ,
\end{gather}
where $Y_1(\lambda; x)$ is the first entry of the vector $Y$.
\begin{thm}
\label{thm:3.4}
In the regime $x \to - \infty$,  the potential $u(x)$ has the following asymptotic behaviour
\begin{gather}
\label{u_theo}
u(x) = \eta_2^2-\eta_1^2  -2\eta_2^2\dfrac{E(m)}{K(m)} -2\dfrac{\partial^2 }{\partial x^2}\log \vartheta_3 \le(\frac{\eta_2}{2K(m)}(x+\phi);2\tau \ri)+ \mathcal{O}\le(|x|^{-1}\ri)
\end{gather}
where $E(m)$ and $K(m)$ are the complete elliptic integrals of the first and second kind respectively with modulus $m=\eta_1/\eta_2$,  $\phi$ is given by
\begin{equation}
\label{phi}
\phi=\int_{\eta_1}^{\eta_2}\dfrac{\log r(\zeta)}{R_+(\zeta)}\dfrac{\d \zeta}{\pi i}\in\R\,
\end{equation}
and $2\tau=i\frac{K(m')}{K(m)}$, $m'=\sqrt{1-m^2}$.  The formula (\ref{u_theo}) can be written in the equivalent form
\begin{equation}
\label{udn}
u(x)=\eta_2^2-\eta_1^2-2\eta_2^2\Jac^2\le( \eta_2(x+\phi) +K(m)\le| \, m\ri. \ri)+ \mathcal{O}\le(|x|^{-1}\ri)
\end{equation}
where $\Jac\le(z \le| \, m\ri. \ri)$ is the Jacobi elliptic function of modulus $m$.

\end{thm}

\begin{proof}
We are interested in the first entry of the vector $Y(\lambda)$ (for $\lambda$ large), and we have, from (\ref{E_infinity}), 
\begin{eqnarray*}
Y(\lb)&=\le( \begin{bmatrix} 1& 1\end{bmatrix} +\displaystyle \frac{\mathcal{E}_{1}(x)}{x \lambda } \ + \mathcal{O}\left( \frac{1}{\lambda^{2}} \right) \ri) P^{\infty}(\lb)e^{-xg(\lambda) \sigma_3}f(\lb)^{-\sigma_3}\\
& =\left(S^{\infty}(\lb)+ \displaystyle \frac{\mathcal{E}_{1}(x)}{x \lambda } \ + \mathcal{O}\left( \frac{1}{\lambda^{2}} \right) \right)e^{-xg(\lambda) \sigma_3}f(\lb)^{-\sigma_3}\,.
\end{eqnarray*}
Hence
\begin{gather}
Y_1(\lambda) =\le[ S^{\infty}_{1}(\lb) + \frac{\left( \mathcal{E}_{1}\right)_{1}(x)}{x \lambda } \ + \mathcal{O}\left( \frac{1}{\lambda^{2}} \right) \ri] \frac{e^{-xg(\lambda) }}{f(\lb)} \, .
\end{gather}
From the expression  (\ref{ggg}) for the $g$ function, we have
\begin{gather}
e^{-xg(\lambda) } =  1 - \frac{x}{\lambda} \le[ \frac{\eta_1^2+\eta_2^2}{2}  +\eta_2^2\left(\dfrac{E(m)}{K(m)}-1\right) \ri] + \mathcal{O}\le(\frac{1}{\lambda^2}\ri) \, .
\end{gather}
From the formula of $f(\lb)$  in (\ref{f}) we have 
\[
f(\lb)=1+\frac{f_1}{\lb}+\mathcal{O}\le(\frac{1}{\lambda^2}\ri)\, ,
\]
where $f_1$ is  independent of $x$.
\edit{
Starting from the vector $S^{\infty}(\lb)$ in (\ref{Sinfty_sol}) we observe that $\gamma(\lambda) = 1 + \mathcal{O} \le(\lambda^{-2} \ri)$, using \eqref{omegadef} and \eqref{Abeldef} we have
\[
	2w(\lambda) = -\frac{1}{2} - \dfrac{1}{\lb}\dfrac{\Omega}{2\pi i}+\mathcal{O}\le(\frac{1}{\lambda^2}\ri)\, ,\;\quad \frac{\Omega}{2\pi i}=-\dfrac{\eta_2}{2K(m)}
\]
so expanding \eqref{Sinfty_sol} gives
\begin{align*}
S^{\infty}_{1}(\lambda) 
&= 1 - \frac{1}{\lambda} \frac{\Omega}{2\pi i} \left[ \frac{ \vartheta_3' \le( \frac{x\Omega+\Delta}{2\pi i}; 2 \tau \ri) }{ \vartheta_3 \le( \frac{x\Omega+\Delta}{2\pi i}; 2 \tau \ri) } - \frac{ \vartheta_3'(0; 2\tau)}{\vartheta_3(0;2\tau)}
 \right] + \le(\frac{1}{\lambda^2}\ri) \\
&= 1 - \dfrac{1}{\lb}\dfrac{\partial }{\partial x}\log \vartheta_3 \le(\frac{x\Omega+\Delta}{2\pi i };2\tau \ri)+ \mathcal{O}\le(\frac{1}{\lambda^2}\ri),
\end{align*}
where we have used the property that $ \vartheta_3'(0; 2\tau)=0$ because $ \vartheta_3(z; 2\tau)$ is an even function of $z$.}
Therefore
\begin{eqnarray*}
Y_1(\lb)=1+\dfrac{1}{\lb}\left(f_1-x \le[ \frac{\eta_1^2+\eta_2^2}{2}  +\eta_2^2\left(\dfrac{E(m)}{K(m)}-1\right) \ri] -\dfrac{\partial }{\partial x}\log \vartheta_3 \le(\frac{x\Omega+\Delta}{2\pi i };2\tau \ri)+ \frac{\left(\mathcal{E}_{1}(x) \right)_{1}}{x} \right)  +\mathcal{O}\le(\frac{1}{\lambda^2}\ri).
\end{eqnarray*}
From the above  expansions,   using  (\ref{uexp0}), and the explicit expression of  $\Delta$ in (\ref{Delta}),  
we obtain the expression of $u(x)$ in \eqref{u_theo}.
In order to obtain the expression (\ref{udn}) we need the following identity (see e.g. \cite{Lawden}   pg. 45 exercise 16 and 3.5.5)
\begin{equation}
\label{theta-jacobi}
\dfrac{1}{4K^2(m)}\dfrac{\d^2}{\d z^2} \log \vartheta_3(z;2\tau)=-\dfrac{E(m)}{K(m)}+\Jac^2\le( 2K(m) z+K(m)\le| \, m\ri. \ri)\, ,
\end{equation}
where $\Jac\le( z\le| \, m\ri. \ri)$ is the Jacobi elliptic function   of modulus $m$ and  period $2K(m)$  and we recall that $2\tau=iK(m')/K(m)$.
Then we can write 
\[
 \dfrac{\partial^2 }{\partial x^2}\log \vartheta_3 \le(\frac{x\Omega+\Delta}{2\pi i };2\tau\ri )=-\eta_2^2\dfrac{E(m)}{K(m)}+\eta_2^2\Jac^2\le( \eta_2(x+\phi) + K(m)\le| \, m\ri. \ri)\, ,
\]
so that  the expression for $u(x)$ in (\ref{u_theo}) can be written in the form (\ref{udn}).
\end{proof}

\section{Behaviour of the potential $u(x,t)$ as $t \to +\infty$}
\label{sec:4}
Letting the potential $u(x,t)$ evolve in time according to the KdV equation, the reflection coefficient evolves as $\edit{r}(\lambda; t) = \edit{r}(\lambda) e^{-8\lambda^3t}$. This will lead to the study of a Riemann--Hilbert  problem $Y$ for the soliton gas described as follows
\begin{align}
&Y_+ (\lambda) = Y_-(\lambda) \begin{cases}
\displaystyle \begin{bmatrix} 1 &0  \\ -i r(\lambda) e^{8\lambda t \le( \lambda^2 - \frac{x}{4t} \ri)}  & 1 \end{bmatrix} &\quad \lambda \in \Sigma_{1}\\
\\
\displaystyle \begin{bmatrix} 1 & ir(\lambda) e^{-8\lambda t \le( \lambda^2 -  \frac{x}{4t} \ri)} \\ 0  & 1 \end{bmatrix} & \quad \lambda \in  \Sigma_{2}
 \end{cases}\\
&Y(\lambda) = \begin{bmatrix}1&1 \end{bmatrix} + \mathcal{O}\le(\frac{1}{\lambda}\ri) \qquad \lambda \rightarrow \infty \ .
\end{align}
We are interested in the asymptotic behaviour of $Y(\lambda) $ in the long-time regime ($t\to + \infty$).

The phase appearing in the exponents in the jump matrix shows different sign depending on the value of the quantity 
\begin{gather}
\xi = \frac{x}{4t} \in \R \ .
\end{gather} 
It is clear that in the case $\xi > \eta_2^2$, the phases in the jumps are exponentially decaying in the regime $t \to + \infty$, therefore by a straightforward small norm argument we conclude
\begin{gather}
Y(\lambda) = \begin{bmatrix}1&1 \end{bmatrix} + \mathcal{O}\le( \edit{e^{-8\eta_1(\xi^2-\eta_2^2)t}} \ri) \qquad \text{as } t\to + \infty \edit{\text{ with }  \xi^2>\eta_2^2,}
\end{gather}
 and the potential $u(x,t)$ becomes trivial. 

The more interesting case $\xi \leq \eta_2^2$ will be studied below. It will become clear that we will observe the presence of a critical value $\xi_{\rm crit}$ at which a phase transition occurs when passing from $\xi> \xi_{\rm crit}$ (the ``super-critical" case) to $\xi\leq \xi_{\rm crit}$ (the ``sub-critical" case).  In the first case the asymptotic description gives an asymptotic solution that is a modulated travelling wave (the wave parameters are changing slowly in time),
 while in the  sub-critical case, the asymptotic solution is a travelling wave.

\section{Super-critical case: the $\alpha$-dependency}\label{supercritcase}
\label{sec:5}

We first consider the case 
\begin{gather} \xi  _{\rm crit} < \xi   < \eta_2^2 \end{gather}
where the value of $\xi  _{\rm crit}  \in \R$ will be defined in (\ref{xicrit}).

In order to study the Riemann--Hilbert  problem for $Y$ in this setting we need to split the contours in the following way: let $\alpha \in (\eta_1, \eta_2)$ and define the sub intervals 
\begin{gather} 
\Sigma_{1,\alpha} = (\alpha, \eta_2) \subseteq \Sigma_1 \qquad \text{and} \qquad \Sigma_{2,\alpha}= (-\eta_2, -\alpha) \subseteq \Sigma_2 \ .
\end{gather}
The value of $\alpha$ will be determined in equation (\ref{alpha}) as a function of $\xi  $.

We introduce again   scalar functions  $g(\lambda)$   and $f(\lb)$ (in a slight abuse of notation, we are using the same letter $g$  and  $f$ to denote these functions, 
though properly we should probably use  $g_{\alpha}$ and  $f_{\alpha}$).
   We make the first transformation $Y(\lb)\mapsto T(\lb)$ given by 
\begin{gather}
T(\lambda) =  Y(\lambda) e^{t g(\lambda)\sigma_3}f(\lb)^{\sigma_3} \label{RHPTtnew}
\end{gather}
such that 
\begin{align}
&g_+(\lambda) + g_-(\lambda) +8\lambda^3 - 8\xi  \lambda = 0  & \lambda \in \Sigma_{1,\alpha}\cup\Sigma_{2,\alpha} \label{c1}\\
&g_+(\lambda)-g_-(\lambda) =  \widetilde \Omega & \lambda \in  [-\alpha, \alpha] \label{c2}\\
&g(\lambda) = \mathcal{O}\le(\frac{1}{\lambda}\ri) & \lambda \rightarrow  \infty \ .\label{c3}
\end{align}
 We further require that $g(\lb)-4\lb^3+4\xi  \lb - \wt{\Omega}$ behaves as $(\lb \mp \alpha)^{\frac{3}{2}}$ near $\lb=\pm \alpha $.  In addition, there are two types of inequalities that must be satisfied by this function in order to have a successful Riemann--Hilbert analysis.  First we will need inequalities satisfied on the complement (relative to $\Sigma_1\cup \Sigma_2$) of the sets $\Sigma_{1,\alpha}$ and $\Sigma_{2,\alpha}$:
\begin{eqnarray}
&&\Re\left[ g_+(\lambda) + g_-(\lambda) +8\lambda^3 - 8\xi  \lambda \right] < 0  \ \ \ \  \lambda \in (\eta_{1},\alpha)  \\
&&\Re \left[ g_+(\lambda) + g_-(\lambda) +8\lambda^3 - 8\xi  \lambda  \right] > 0  \ \ \ \  \lambda \in (-\alpha, -\eta_{1}) \ .
\end{eqnarray}
Second, we will require monotonicity properties on $\Sigma_{1}$ and $\Sigma_{2}$:
\begin{eqnarray}
&&-i( g_{+}(\lambda) - g_{-}(\lambda)) \mbox{ is purely real and monotonically decreasing  on }(\alpha, \eta_{2}) \   \\
&&-i( g_{+}(\lambda) - g_{-}(\lambda)) \mbox{ is purely real and monotonically increasing  on }(-\eta_{2}, - \alpha) \ .
\end{eqnarray}

It is well-known that there is a unique function $g$ satisfying all these properties, which we will define explicitly here (we will actually define $g'$, which of course determines $g$).  We define

%
%
\begin{align}
\label{gprime}
g'(\lambda) = - 12\lambda^2 + 4\xi   + 12 \dfrac{Q_2(\lb)}{R_\alpha(\lambda) }-4\xi  \dfrac{Q_1(\lb)}{R_\alpha(\lambda) }\, ,
\end{align}
where 
\begin{gather}
 R_\alpha(\lambda) =\sqrt{(\lambda^2-\alpha^2)(\lambda^2-\eta_2^2)} 
\end{gather}
is taken to be analytic in $\C\backslash \le\{ \Sigma_{1,\alpha}\cup\Sigma_{2,\alpha} \ri\}$  and real and positive on $(\eta_2,+\infty)$; moreover, let
\begin{equation}
\label{P1P2}
Q_1(\lb)=\lb^2+c_1\, ,\quad \mbox{ and } \ Q_2(\lb)=\lb^4-\frac{1}{2}\lb^2(\alpha^2+\eta_2^2)+c_2\, .
\end{equation}
The constants $c_1$ and $c_2$  are chosen so that 
\begin{equation}
\label{eq:5.14}
\int_{-\alpha}^{\alpha}\dfrac{Q_2(\zeta)}{R_{\alpha +}(\zeta)}\d\zeta = 0 \ , \quad \int_{-\alpha}^{\alpha}\dfrac{Q_1(\zeta)}{R_{\alpha +}(\zeta)}\d\zeta = 0 \ .
  \end{equation}
Explicitly, we find
\begin{eqnarray}
\label{c12}
c_1=- \eta_2^2+\eta_2^2\dfrac{E(m_{\alpha})}{K(m_{\alpha})}\, ,\quad 
c_2=\dfrac{1}{3}\alpha^2\eta_2^2+\dfrac{1}{6}(\eta_2^2+\alpha^2)c_1 \ , \quad 
 m_{\alpha}=\dfrac{\alpha}{\eta_2}\, ,
\end{eqnarray}
where $K(m_{\alpha})$ and $E(m_{\alpha})$ are, respectively, the complete elliptic integrals of the first and second kind.

The parameter $\alpha$ is determined by requiring that the  function  $g(\lb)-4\lambda^3 + 4\xi  \lambda - \wt{\Omega}$ has a zero at $\lb=\pm \alpha$, which yields the equation
\begin{equation}
\label{alpha}
\xi  =3\dfrac{Q_2(\pm\alpha)}{Q_1(\pm\alpha)}=\dfrac{1}{2}(\alpha^2+\eta_2^2)+\dfrac{\alpha^2(\alpha^2-\eta_2^2)}{\alpha^2-\eta_2^2+\eta_2^2\frac{E(m_{\alpha})}{K(m_{\alpha})}}\, ,
\end{equation}
and this determines the constant $\alpha$ implicitly as a function of $\xi  $.  
 
 Before continuing our analysis we want to   comment  on equation (\ref{alpha}).
We can rewrite it in the form
\begin{equation}
\label{xim}
\xi  =\dfrac{x}{4t}=\dfrac{\eta_2^2}{2}W(m_\alpha)\, ,\quad W(m_{\alpha})=\left[1+m_{\alpha}^2+2\dfrac{m_{\alpha}^2(1-m_{\alpha}^2)}{1-m_{\alpha}^2-\frac{E(m_{\alpha})}{K(m_{\alpha})}}\right]\, .
\end{equation}
This relation  describes the modulation of the parameter $\alpha$ as a function of $\xi$.   The quantity $\eta_2^2W(m_{\alpha}) $ was  derived by Whitham in his modulation  theory of  the  traveling wave solution of  the  KdV equation  \cite{Whitham}.  In the general theory  there are three parameters involved,  while in our case, two parameters are fixed, one being zero and the other one $\eta_2$.
This specific case gives a self-similar solution to the Whitham equations. This solution   was derived and  used by  Gurevich-Pitaevskii  \cite{GP73}  to describe the  modulation   of the travelling wave  that is formed   in the solution of the KdV equation with step   
initial data $u(x)=-\eta_2^2$ for $x<0$ and $u(x)=0$ for $x>0$ and was called a dispersive shock wave in analogy with the shock wave that is formed in the solution of the Hopf equation $u_t+6uu_x=0$ for step initial data.   

 Using the expansion of the elliptic functions one has
\[
\dfrac{E(m_\alpha)}{K(m_\alpha)}=1-\dfrac{1}{2}m_\alpha^2+\mathcal{O}(m_\alpha^4)\, ,\ \text{as }  m_\alpha\to 0 \quad \text{and}\quad \dfrac{E(m_\alpha)}{K(m_\alpha)}\simeq \dfrac{2}{\log(8/(1-m_{\alpha}))}\, ,\ \text{as } m_\alpha\to 1\, ,
\]
so that 
\[
\lim_{\alpha\to 0 }\dfrac{3Q_2(\alpha)}{Q_1(\alpha)}=-\frac{3\eta_2^2}{2}\, , \quad \text{and} \quad  \lim_{\alpha\to\eta_2}\dfrac{3Q_2(\alpha)}{Q_1(\alpha)}=\eta_2^2 \, .
\]

The Whitham equations are strictly hyperbolic (\cite{Lev88}),  so  that $\dfrac{\partial }{\partial \alpha}W(m_\alpha)>0$ for $0<\alpha<\eta_2$. Hence by  the implicit   function theorem, the  equation (\ref{xim}) 
defines  $\alpha$ as a monotone increasing function of $\xi$ for $\xi \in [\xi_{\rm crit},\eta_2^2]$ where $\xi_{\rm crit}$ is given by
\begin{equation}
\label{xicrit}
\xi_{\rm crit}=\dfrac{3Q_2(\eta_1)}{Q_1(\eta_1)}=\dfrac{1}{2}(\eta_1^2+\eta_2^2)+\dfrac{\eta_1^2(\eta_1^2-\eta_2^2)}{\eta_1^2-\eta_2^2+\eta_2^2\frac{E(m)}{K(m)}}\, , \quad m = \frac{\eta_1}{\eta_2}\, .
\end{equation}
Then, clearly  $\xi_{\rm crit}>-\frac{3\eta_2^2}{2}$.

From $g'(\lb)$,  we also have a representation of $g(\lb)$:
\begin{equation}
\label{g_alpha}
g(\lb)=-4\lb^3+4\xi  \lb+12\int_{\eta_2}^\lb  \dfrac{Q_2(\zeta)}{R_\alpha(\zeta) }\d\zeta-4\xi   \int_{\eta_2}^\lb  \dfrac{Q_1(\zeta)}{R_\alpha(\zeta) }\d\zeta\, .
\end{equation}
 This, together with \eqref{c2}, yields the formula
 \begin{equation}
 \label{tildeOmega}
\wt{\Omega}=24\int_{\eta_2}^{\alpha}\dfrac{Q_2(\zeta)}{R_{\alpha +}(\zeta)}\d\zeta-8\xi\int_{\eta_2}^{\alpha}\dfrac{Q_1(\zeta)}{R_{\alpha +}(\zeta)}\d \zeta \, .
  \end{equation}

For future use we will need the $x$ derivatives of $t g(\lb)$ and $t\wt{\Omega}$.
Before calculating them, let us observe that 
\[
\wt \Omega =24\int_{\eta_2}^{\alpha}\dfrac{Q_2(\zeta)-Q_2(\alpha)}{R_{\alpha +}(\zeta)}\d\zeta-8\xi\int_{\eta_2}^{\alpha}\dfrac{Q_1(\zeta)-Q_1(\alpha)}{R_{\alpha +}(\zeta)}\d\zeta \, ,
\]
  which gives, using the Riemann bilinear relations (see e.g. \cite{Springer}),
\begin{gather}
\wt \Omega =  2\pi i \frac{4\xi  - 2(\alpha^2+\eta_2^2 )}{\int_{-\alpha}^{\alpha} \frac{\d \zeta}{R_\alpha(\zeta) }}=2\pi i \eta_2 \frac{\alpha^2+\eta_2^2 -2\xi  }{K(m_{\alpha})}   \in i\R \ ,\quad m_{\alpha}=\dfrac{\alpha}{\eta_2}\, . \label{wtomega}
\end{gather}

\begin{figure}
\centering
\scalebox{.9}{
\begin{tikzpicture}[>=stealth]
\path (0,0) coordinate (O);

\draw (-4,-1) -- (5,-1);
\draw (-2,1) -- (7,1);
\draw (-4,-1) -- (-2,1); 
\draw (5,-1) -- (7,1);
\node at (5.4,-0.1) {$\times$};
\node[above] at (5.5,-0.1) {$\infty^-$};

\draw (-4, 3) -- (5,3);
\draw (-2,5) -- (7,5);
\draw (-4,3) -- (-2,5);
\draw (5,3) -- (7,5);
\node at (5.4,3.9) {$\times$};
\node[above] at (5.5,3.9) {$\infty^+$};

\draw (-2,0) -- (.5,0);
\draw (1.5,0) -- (4,0);

\draw (-2,4) -- (.5,4);
\draw (1.5,4) -- (4,4);

\draw[dashed, black!30] (-2,0) -- (-2,4);
\draw[dashed, black!30] (.5,0) -- (.5,4);
\draw[dashed, black!30] (1.5,0) -- (1.5,4);
\draw[dashed, black!30] (4,0) -- (4,4);

\draw[fill] (-2,0) circle [radius=0.025];
\node[below ] at (-2,0) {\tiny $-\eta_2$};
\draw[fill] (0.5,0) circle [radius=0.025];
\node[below ] at (0.5,0) {\tiny $-\alpha$};
\draw[fill] (1.5,0) circle [radius=0.025];
\node[below ] at (1.5,0) {\tiny $\alpha$};
\draw[fill] (4,0) circle [radius=0.025];
\node[below ] at (4,0) {\tiny $\eta_2$};

\draw[fill] (-2,4) circle [radius=0.025];
\node[above ] at (-2,4) {\tiny $-\eta_2$};
\draw[fill] (0.5,4) circle [radius=0.025];
\node[above ] at (0.5,4) {\tiny $-\alpha$};
\draw[fill] (1.5,4) circle [radius=0.025];
\node[above ] at (1.5,4) {\tiny $\alpha$};
\draw[fill] (4,4) circle [radius=0.025];
\node[above ] at (4,4) {\tiny $\eta_2$};



\draw[->- = .25, red] (-1,4) .. controls + (70:.5cm) and + (70:.5cm) .. (2.5,4);
\draw[->- = .25, red] (2.5,0) .. controls + (-110:.5cm) and + (-110:.5cm) .. (-1,0);
\draw[red!30, dashed] (-1,0) -- (-1,4);
\draw[red!30, dashed] (2.5,0) -- (2.5,4);
\node[above, red] at (1,4.3) {\small $A$};

\draw[->- = .7, blue] (1,4) .. controls + (70:.5cm) and + (70:.5cm) .. (4.5,4);
\draw[->- = .7, blue] (4.5,4) .. controls + (-110:.5cm) and + (-110:.5cm) .. (1,4);
\node[blue, above] at (3.25,4.3) {\small $B$};

\end{tikzpicture}
}
\caption{Construction of the genus-$1$ Riemann surface $\mathfrak{X}_\alpha$ and its basis of cycles.}
\label{frakXa}
\end{figure}

\begin{lemma}
The following identities are satisfied
\begin{align}
\label{gprime_der}
&\dfrac{\partial}{\partial x}tg'(\lb)=1-\dfrac{Q_1(\lb)}{R_\alpha(\lambda)}\, ,\\
\label{Omega_der}
&\dfrac{\partial}{\partial x}t\wt \Omega =- \frac{ \pi i \eta_2 }{K(m_{\alpha})}\, .
\end{align}
\end{lemma} 
\begin{proof}
We observe that $g'(\lb)\d\lb$  defined in (\ref{gprime})  is a meromorphic one-form on the Riemann surface $\mathfrak{X}_{\alpha}$   defined  as 
\[
\mathfrak{X}_{\alpha}=\le\{(\eta,\lb)\in\C^2\;|\; \eta^2=R_\alpha^2(\lb)=(\lb^2-\alpha^2)(\lb^2-\eta_2^2)\ri\} \, .
\]
   We define a homology basis on $\mathfrak{X}_{\alpha}$ in the following way: the  $B$ cycle   encircles the cut $[\alpha,\eta_2]$ clockwise  and the $A$ cycle starts on the cut $[-\eta_2,-\alpha]$  on the upper semi-plane, 
   goes to the cut $[\alpha,\eta_2]$ and then goes back to $[-\eta_2,-\alpha]$  on the second sheet of $\mathfrak{X}_{\alpha}$. \edit{See Figure~\ref{frakXa}. }
 Then we have 
 \begin{equation}
 \label{integrals}
\oint_{A}g'(\zeta)\, \d\zeta =0\, ,\quad \oint_{B}g'(\zeta)\, \d\zeta=-\wt\Omega\, .
\end{equation}
Regarding the first relation in (\ref{gprime_der}) we have 
\begin{align}
\dfrac{\partial}{\partial x}tg'(\lb)\d\lb&=\dfrac{\partial}{\partial x}\left[  - 12t\lambda^2 \d\lb+ x\d\lb   + 12t \dfrac{Q_2(\lb)}{R_\alpha(\lambda) }\d\lb-x  \dfrac{Q_1(\lb)}{R_\alpha(\lambda) }\d\lb\right]\\
&=\d\lb-\dfrac{Q_1(\lb)}{R_\alpha(\lambda)}\d\lb +\dfrac{\partial}{\partial \alpha}\left[  12t \dfrac{Q_2(\lb)}{R_\alpha(\lambda) }\d\lb-x  \dfrac{Q_1(\lb)}{R_\alpha(\lambda) }\d\lb\right]\dfrac{\partial \alpha}{\partial x}\\
&=\d\lb-\dfrac{Q_1(\lb)}{R_\alpha(\lambda)}\d\lb \, ,
\end{align}
because the term $\dfrac{\partial}{\partial \alpha}\left[  12t \dfrac{Q_2(\lb)}{R_\alpha(\lambda) }\d\lb-x  \dfrac{Q_1(\lb)}{R_\alpha(\lambda) }\d\lb\right]$  vanishes since it is a holomorphic one-form  (no singularity at $\pm\alpha$ or infinity)
which is normalized to zero on the $A$ cycle  because of (\ref{integrals}); therefore it is identically zero \cite{Krichever} (see also \cite{G02},\cite{GravaTian}).  An alternative proof is to calculate the  derivative and  use the explicit formul\ae \ of the constants $c_1$ and $c_2$  in (\ref{c12}).  We conclude that 
\[
\dfrac{\partial}{\partial x}e^{-tg(\lambda) }=- \frac{1}{\lambda} \le[ \frac{\alpha^2+\eta_2^2}{2}  +\eta_2^2\left(\dfrac{E(m_\alpha)}{K(m_\alpha)}-1\right) \ri] + \mathcal{O}\le(\frac{1}{\lambda^2}\ri)\, .
\]
Regarding the  relation  (\ref{Omega_der}), by (\ref{gprime_der}) and (\ref{integrals}) we have 
\[
\dfrac{\partial}{\partial x}(t \wt \Omega)=-\dfrac{\partial}{\partial x} \oint_{B}tg'(\lb)\, \d\lb=- \oint_{B}\dfrac{\partial}{\partial x}(tg'(\lb)\, \d\lb)=- \frac{ \pi i \eta_2 }{K(m_{\alpha})} \, . 
\]
\end{proof}
 As we did in Section \ref{sec:3}, we choose the function $f$ to simplify the jumps on $\Sigma_{1,\alpha}$ and $\Sigma_{2,\alpha}$ via
 \begin{align}
& f_+(\lambda) f_-(\lambda) =\frac{1}{r(\lambda)} & \lambda \in \Sigma_{1,\alpha} \label{Tfcon1} \\
&f_+(\lambda) f_-(\lambda) = r(\lambda) & \lambda \in \Sigma_{2,\alpha} \\
& \frac{f_+(\lambda)}{ f_-(\lambda)} =e^{\wt{\Delta}} & \lambda \in [-\alpha,\alpha] \\
& f(\lambda) = 1+\mathcal{O}\le(\frac{1}{\lambda}\ri) & \lambda \rightarrow  \infty  \ .\label{Tfcon3}
\end{align}
It is easy to check that the function $f(\lb)$ is given by
\begin{equation}
\label{f_alpha}
f(\lb)=\exp\le\{ \frac{R_{\alpha}(\lb)}{2\pi i}\left[\int_{ \Sigma_{1,\alpha}}  \frac{\log  \frac{1}{r(\zeta)} }{R_{\alpha +}(\zeta)(\zeta-\lambda)} \d \zeta+
\int_{\Sigma_{2,\alpha}}  \frac{\log r(\zeta )}{R_{\alpha +}(\zeta)(\zeta-\lambda)} \d \zeta+\int^{\alpha}_{ -\alpha}  \frac{\wt{\Delta}}{R_{\alpha}(\zeta)(\zeta-\lambda)} \d \zeta
\right] \ri\} \, ,
\end{equation}
where the constraint  \eqref{Tfcon3} determines $\wt{\Delta}$ as 
\begin{align}
\nonumber
\wt{\Delta}&=\left[\int_{\Sigma_{1, \alpha}}  \frac{\log r(\zeta )}{R_{\alpha +}(\zeta)} \d \zeta-
\int_{\Sigma_{2, \alpha}}  \frac{\log r(\zeta )}{R_{\alpha +}(\zeta)} \d \zeta\right]\left[\int^{\alpha}_{ -\alpha}  \frac{ \d \zeta}{R_{\alpha}(\zeta)}\right]^{-1}\\
\label{TDelta}
&=2\left[\int_{\Sigma_{1, \alpha}}  \frac{\log r(\zeta)}{R_{\alpha +}(\zeta)} \d \zeta\right]\left[\int^{\alpha}_{ -\alpha}  \frac{ \d \zeta}{R_{\alpha}(\zeta)}\right]^{-1}\, ,
\end{align}
where in the last relation  in (\ref{TDelta}) we have used the fact that $r(-\lb)=r(\lb)$.

As a consequence, $T$ satisfies the following Riemann--Hilbert  problem:
\begin{align}
&T_+(\lambda) = T_-(\lambda)  \begin{cases}
\displaystyle \begin{bmatrix} e^{t\le(g_+(\lb) - g_-(\lb)\ri)} \frac{f_{+}(\lb)}{f_{-}(\lb)}& 0 \\ -i& e^{-t\le(g_+(\lb) - g_-(\lb)\ri)} \frac{f_{-}(\lb)}{f_{+}(\lb)}\end{bmatrix} &\quad \lambda \in  \Sigma_{1,\alpha}\\
\displaystyle \begin{bmatrix} e^{t\le(g_+(\lb) - g_-(\lb)\ri)}  \frac{f_{+}(\lb)}{f_{-}(\lb)}&i \\ 0 & e^{-t\le(g_+(\lb) - g_-(\lb)\ri)} \frac{f_{-}(\lb)}{f_{+}(\lb)} \end{bmatrix} & \quad \lambda \in \Sigma_{2,\alpha} \\
\displaystyle \begin{bmatrix} e^{\wt \Omega t + \wt{ \Delta} } & 0 \\ - i r(\lb) f_{+}(\lb)f_{-}(\lb) e^{t \le(g_+(\lb) + g_-(\lb) +8\lambda^3 - 8\xi  \lambda  \ri)}& e^{-\wt \Omega t - \wt{ \Delta} }  \end{bmatrix} &\quad \lambda \in  [\eta_1,\alpha]\\
\displaystyle \begin{bmatrix} e^{\wt \Omega t +\wt{ \Delta} }  & e^{-t \le(g_+(\lb) + g_-(\lb) +8\lambda^3 - 8\xi  \lambda  \ri)} \frac{ i r(\lb) }{f_{+}(\lb)f_{-}(\lb)} \\ 0 & e^{-\wt \Omega t- \wt{ \Delta} }  \end{bmatrix} & \quad \lambda \in [-\alpha, -\eta_1] \\
\displaystyle\begin{bmatrix} e^{\wt \Omega t+\wt{ \Delta} } &0 \\ 0 & e^{-\wt \Omega t-\wt{ \Delta} }  \end{bmatrix} &\quad \lambda \in   [-\eta_1, \eta_1] 
 \end{cases} \\
 &T(\lambda) = \begin{bmatrix}1&1 \end{bmatrix} + \mathcal{O}\le(\frac{1}{\lambda}\ri) \qquad \lambda \rightarrow \infty\, .
\end{align}

\subsection{Opening lenses}

It is useful to provide representations of the entries appearing in the jump matrix for $T(\lambda)$ in either $\Sigma_{1,\alpha}$ or $\Sigma_{2,\alpha}$, that clearly demonstrate their analytic continuation off  these intervals, as it was done in Section \ref{sec:3}.  The following formul\ae \ are valid on both intervals:
\begin{eqnarray}
&&g_{+}(\lb) - g_{-}(\lb) = 2 g_{+}(\lb) + 8 \lambda^{3} - 8 \xi \lambda \ , \\
&& g_{+}(\lb) - g_{-}(\lb) = -\left(2 g_{-}(\lb) + 8 \lambda^{3} - 8 \xi \lambda \right) \ .
\end{eqnarray}
The following formul\ae \ are valid on $\Sigma_{1,\alpha}$:
\begin{eqnarray}
&& \frac{f_{+}(\lb)}{f_{-}(\lb)} =- \frac{1}{  f_{-}^{2}(\lb) \hat{r}_{-}(\lb) }  \quad \text{and} \quad  \frac{f_{-}(\lb)}{f_{+}(\lb)} = \frac{1}{ f_{+}^{2}(\lb) \hat{r}_{+}(\lb)} \ .
\end{eqnarray}
And the following ones are valid on $\Sigma_{2,\alpha}$:
\begin{eqnarray}
&&\frac{f_{-}(\lb)}{f_{+}(\lb)} = - \frac{f_{-}^{2}(\lb)}{  \hat{r}_{-}(\lb)}  \quad \text{and} \quad \frac{f_{+}(\lb)}{f_{-}(\lb)} = \frac{f_{+}^{2}(\lb)}{ \hat{r}_{+}(\lb)} \ .
\end{eqnarray}

As was done in Section \ref{sec:3}, we can factor the jump matrix on $\Sigma_{1,\alpha}$ as follows
\begin{align*}
&\begin{bmatrix} e^{t\le(g_+(\lb) - g_-(\lb)\ri)}\dfrac{f_+(\lb)}{f_-(\lb)} & 0 \\ -i  & e^{-t\le(g_+(\lb) - g_-(\lb)\ri)}\dfrac{f_-(\lb)}{f_+(\lb)} \end{bmatrix} =\\
&\quad\quad
=\begin{bmatrix}1&-\dfrac{ie^{ -t \left(2 g_{-}(\lb) + 8 \lambda^{3} - 8 \xi \lambda \right))}}{f_-^2 (\lb)\hat{r}_-(\lambda)}\\ 0 &1 \end{bmatrix} 
\begin{bmatrix}0&-i \\ -i &0\end{bmatrix} 
\begin{bmatrix}1&\dfrac{ie^{ -t \left(2 g_{+}(\lb) + 8 \lambda^{3} - 8 \xi \lambda \right)}}{  \hat{r}_+(\lambda)f_+^2(\lb)}&  \\0&1 \end{bmatrix} 
\end{align*}
and on \edit{$\Sigma_{2,\alpha}$} as
\begin{align*}
& \begin{bmatrix} e^{t\le(g_+(\lb) - g_-(\lb)\ri)}\dfrac{f_+(\lb)}{f_-(\lb)}&i \\ 0 & e^{-t\le(g_+(\lb) - g_-(\lb)\ri)}\dfrac{f_-(\lb)}{f_+(\lb)} \end{bmatrix}=\\
&\quad\quad
=\begin{bmatrix}1&0\\i\dfrac{f_-^2 (\lb)}{\hat{r}_-(\lambda)}e^{ t \left(2 g_{-} (\lb)+ 8 \lambda^{3} - 8 \xi \lambda \right))} &1 \end{bmatrix} 
\begin{bmatrix}0&i\\i&0\end{bmatrix} 
\begin{bmatrix}1&0\\ -i\dfrac{f_+^2(\lb)}{ \hat{r}_+(\lambda)}e^{ t\left(2 g_{+}(\lb) + 8 \lambda^{3} - 8 \xi \lambda \right)}&  &1 \end{bmatrix}\, .
\end{align*}

\edit{
These factorizations motivate us to open lenses $\mathcal{C}_j$ around each $\Sigma_{j,\alpha}$, $j=1,2$.  We use these lenses to define the transformation
\begin{equation}
	S(z) = \begin{dcases}
		T(z) \begin{bmatrix} 1 & \frac{-i} {\hat{r}(\lambda) f^2(\lambda)}  e^{-2t( g(\lambda) +\lambda^3 -4\xi \lambda)} \\  0 & 1 \end{bmatrix}
		& \text{inside the lens $\mathcal{C}_1$} \\
		T(z) \begin{bmatrix} 1 & 0 \\ \frac{if^2(\lambda)} {\hat{r}(\lambda) }  e^{2t( g(\lambda) +\lambda^3 -4\xi \lambda)} & 1 \end{bmatrix}
		&  \text{inside the lens $\mathcal{C}_2$}\\ 
		T(z) 
		& \text{elsewhere}
	\end{dcases}
\end{equation}
The lens contours and the resulting jump relations for $S(z)$ are shown in \figurename~\ref{openlensesalpha}.
}

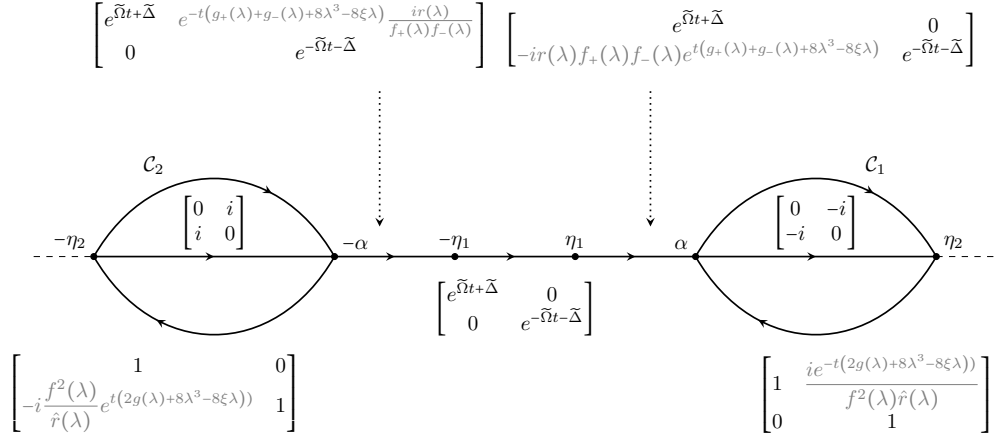
\begin{figure}[h!]
\centering
\scalebox{.8}{
\begin{tikzpicture}[>=stealth]
\path (0,0) coordinate (O);

\draw[dashed] (-8,0) -- (-7,0);
\draw[dashed] (7,0) -- (8,0);

\node [below] at (0,-.25) { $ \begin{bmatrix} e^{\wt \Omega t +\wt{ \Delta} } &0 \\ 0 & e^{-\wt \Omega t - \wt{ \Delta} }\end{bmatrix}$};

\draw[ ->, dotted, thick] (2.25,2.75) -- (2.25,.5);
\draw[ ->, dotted, thick] (-2.25,2.75) -- (-2.25,.5);
\node [above ] at (3.75,3.0) { $\begin{bmatrix} e^{\wt \Omega t + \wt{ \Delta} } & 0 \\ {\color{gray}- i r(\lb) f_{+}(\lb)f_{-}(\lb)  e^{t \le(g_+(\lb) + g_-(\lb) +8\lambda^3 - 8\xi  \lambda  \ri)}}& e^{-\wt \Omega t - \wt{ \Delta} }  \end{bmatrix}$};
\node [above ] at (-3.75,3.0) { $\begin{bmatrix} e^{\wt \Omega t +\wt{ \Delta}  }  & {\color{gray}e^{-t \le(g_+(\lb) + g_- (\lb)+8\lambda^3 - 8\xi  \lambda  \ri)} \frac{ i r(\lb) }{f_{+}(\lb)f_{-}(\lb)} }\\ 0 & e^{-\wt \Omega t- \wt{ \Delta} }  \end{bmatrix} $};

\draw[fill] (3,0) circle [radius=0.05];
\node[above left] at (3,0) {$\alpha$};
\draw[fill] (1,0) circle [radius=0.05];
\node[above ] at (1,0) {$\eta_1$};
\draw[fill] (7,0) circle [radius=0.05];
\node[above right] at (7,0) {$\eta_2$};

\draw[fill] (-3,0) circle [radius=0.05];
\node[above right] at (-3,0) {$- \alpha$};
\draw[fill] (-1,0) circle [radius=0.05];
\node[above ] at (-1,0) {$-\eta_1$};
\draw[fill] (-7,0) circle [radius=0.05];
\node[above left] at (-7,0) {$-\eta_2$};

\draw[ ->-=.5, thick] (-1,0)--(1,0);
\draw[ ->-=.5, thick] (-3,0)--(-1,0);
\draw[ ->-=.5, thick] (1,0)--(3,0);

\draw[->-=.5,thick] (3,0)--(7,0);
\draw[->-=.5,thick] (-7,0)--(-3,0);

\draw[->- = .7,thick] (3,0) .. controls + (60:2cm) and + (120:2cm) .. (7,0);
\draw[->- = .7,thick] (7,0) .. controls + (-120:2cm) and + (-60:2cm)  .. (3,0);

\draw[->- = .7,thick] (-7,0) .. controls + (60:2cm) and + (120:2cm)   .. (-3,0);
\draw[->- = .7,thick] (-3,0) .. controls + (-120:2cm) and + (-60:2cm) .. (-7,0);

\node [above] at (5,0) { $ \begin{bmatrix}0 & -i \\ -i & 0 \end{bmatrix}$};
\node [above] at (6,1.25) {$\mathcal{C}_{1}$};
\node [above] at (-6,1.25) {$\mathcal{C}_{2}$};
\node [below] at (6,-1.5) { $ \begin{bmatrix}1&{\color{gray}\dfrac{ie^{ -t \left(2 g(\lb) + 8 \lambda^{3} - 8 \xi \lambda \right))}}{f^2 (\lb)\hat{r}(\lambda)}}\\ 0 &1 \end{bmatrix} $};


\node [below] at (-6,-1.5) { $ \begin{bmatrix}1&0\\{\color{gray}-i\dfrac{f^2 (\lb)}{\hat{r}(\lambda)}e^{ t \left(2 g (\lb)+ 8 \lambda^{3} - 8 \xi \lambda \right))}} &1 \end{bmatrix} $};
\node [above] at (-5,0) { $ \begin{bmatrix}0 & i \\ i & 0 \end{bmatrix}$};

\end{tikzpicture}
}
\caption{Opening lenses: the gray entries in the jumps represent exponentially small quantities in the limit $t \to + \infty$.  The contours $\mathcal{C}_{1}$ and $\mathcal{C}_{2}$ are the lens boundaries.}
\label{openlensesalpha}
\end{figure}

\begin{lemma}
\label{g_signs}
The following inequalities are satisfied:
\begin{eqnarray} 
&&\Re \le[2 g(\lb) + 8 \lambda^{3} - 8 \xi \lambda \ri]  >0 \ \mbox{ for }\lambda \in \mathcal{C}_{1} \backslash \{ \alpha, \eta_{2} \} \ , 
\\
&&\Re \le[2 g(\lb) + 8 \lambda^{3} - 8 \xi \lambda \ri]  <0 \ \mbox{ for }\lambda \in \mathcal{C}_{2} \backslash \{ - \eta_{2}, - \alpha \} \ , \\
&&\Re \le[ g_+(\lambda) + g_-(\lambda) + 8\lambda^3 - 8\xi  \lambda\ri] < 0 \mbox{ for } \lambda \in [\eta_{1}, \alpha) \ , \\
&& \Re \le[ g_+(\lambda) + g_-(\lambda) + 8\lambda^3 - 8\xi  \lambda\ri] > 0 \mbox{ for } \lambda \in (-\alpha, - \eta_{1}] \ .
\end{eqnarray}
\end{lemma}

\begin{proof}
Using (\ref{alpha}) the function $g'(\lb)$ in (\ref{gprime}) can be written in the form 
\[
g'(\lb)=- 12\lambda^2 + 4\xi   + 12 \dfrac{Q_2(\lb)-Q_2(\alpha)}{R_\alpha(\lambda) }-4\xi  \dfrac{Q_1(\lb)-Q_1(\alpha)}{R_\alpha(\lambda) }\, , 
\]
so that we have
\begin{gather}
g'_+(\lambda) - g'_-(\lambda) = -  i24 \frac{\sqrt{\lambda^2-\alpha^2}}{ \sqrt{\eta_2^2-\lambda^2}} \le[ \lambda^2 - \le(\frac{\eta_2^2 -\alpha^2}{2} + \frac{\xi  }{3} \ri) \ri] 
\end{gather}
and from (\ref{eq:5.14}) we deduce that the quadratic polynomial has one root $\rho_+$ in the interval $[0, \alpha]$ which is positive for $\lambda > \alpha$. Therefore, for $\lambda \in \Sigma_{1,\alpha}$
\begin{gather}
\Im \le[g'_+(\lambda) - g'_-(\lambda)\ri] =-  24 \frac{\sqrt{\lambda^2-\alpha^2}}{ \sqrt{\eta_2^2-\lambda^2}} \le[ \lambda^2 - \le(\frac{\eta_2^2 -\alpha^2}{2} + \frac{\xi  }{3} \ri) \ri]<0 \ .
\end{gather}

From the formula (\ref{g_alpha}) for $g$ we also have that for $\lambda \in [\eta_1,\alpha]$
\begin{gather}
g_+(\lambda) + g_-(\lambda) + 8\lambda^3 - 8\xi  \lambda 
= - 24\int_\lambda^\alpha \frac{\sqrt{\alpha^2-\zeta^2}}{\sqrt{\eta_2^2-\zeta^2}} \le[ \zeta^2-  \le(\frac{\eta_2^2 -\alpha^2}{2} + \frac{\xi  }{3} \ri) \ri]  \d \zeta \ .
\end{gather}

Setting
\begin{gather} h_{\alpha, \xi}(\zeta) = \frac{\sqrt{\alpha^2-\zeta^2}}{\sqrt{\eta_2^2-\zeta^2}} \le[   \zeta^2 - \le(\frac{\eta_2^2 -\alpha^2}{2} + \frac{\xi  }{3} \ri)  \ri] \ ,\end{gather}
we need to show that the function
\begin{gather}
H_{\alpha, \xi}(\lambda) =\int_{\lambda}^\alpha - h_{\alpha,\xi}(\zeta) <0  \qquad  \text{for } \lambda \in [\eta_1,\alpha] \ . \label{Halphaxineg}
\end{gather}

It is easy to check that $H_{\alpha,\xi} (\alpha) =0$ and $  H_{\alpha,\xi} (0) = 0$ (see (\ref{eq:5.14})). Next, $H_{\alpha,\xi}'(\lambda) = h_{\alpha,\xi}(\lambda)$ is negative on $[0,\rho_+]$ and positive on $[\rho_+,\alpha]$. This implies that indeed the inequality (\ref{Halphaxineg}) is satisfied on $[\eta_1,\alpha]$.
\end{proof}

Because of Lemma~\ref{g_signs}, letting $t\to +\infty$, the jump matrices (as depicted in  \figurename \ \ref{openlensesalpha})  will converge to constant jumps exponentially fast outside neighbourhoods of $\pm\alpha$ and $\pm\eta_2$. We then obtain the following model Riemann--Hilbert  problem for $\wt{S}^{\infty}$:
\begin{gather}
\label{Stinfinity1}
\wt{S}^{\infty}_+(\lambda) = \wt{S}^{\infty}_-(\lambda) 
\begin{cases} 
\begin{bmatrix} e^{t\wt{\Omega}+\wt{ \Delta} } &0 \\ 0 & e^{-t\wt{\Omega}-\wt{ \Delta} }\end{bmatrix}  & \lambda \in [-\alpha,\alpha]  \\
 \begin{bmatrix}0 & -i \\ -i & 0 \end{bmatrix} &\lambda \in \Sigma_{1,\alpha}  \\
  \begin{bmatrix}0 & i \\ i & 0 \end{bmatrix} &\lambda \in \Sigma_{2,\alpha}
  \end{cases}\\
\wt{ S}^{\infty}(\lb)=\begin{bmatrix} 1&1\end{bmatrix}+\mathcal{O}\le(\frac{1}{\lb}\ri),\quad \lb\to\infty \, .
 \label{Stinfinity2}
\end{gather}
\subsection{The outer parametrix $\wt {P}^{\infty}$}\label{outerparametrixalpha}

Along the same lines as we did in Section \ref{sec:3}, we construct a (matrix) model problem whose solution will yield a solution of the above (vector) Riemann--Hilbert  problem.  Since the solution of this model problem will be invertible, one is able to arrive at a small-norm Riemann--Hilbert problem for the error in the large-time regime, more directly than if one considers only vector Riemann--Hilbert problems. 

We therefore seek a matrix valued function $\wt{P}^{\infty}$ that is analytic in $\C\backslash (-\eta_2,\eta_2)$ and satisfies the following Riemann--Hilbert problem
\begin{gather}
\label{Pinfty1b}
\wt{P}^{\infty}_+(\lambda) =\wt{ P}^{\infty}_-(\lambda) 
\begin{cases} 
\begin{bmatrix} e^{t\wt{\Omega}+\wt{ \Delta} } &0 \\ 0 & e^{-t\wt{\Omega}-\wt{ \Delta} }\end{bmatrix}  & \lambda \in [-\alpha,\alpha]  \\
 \begin{bmatrix}0 & -i \\ -i & 0 \end{bmatrix} &\lambda \in \Sigma_{1,\alpha}  \\
  \begin{bmatrix}0 & i \\ i & 0 \end{bmatrix} &\lambda \in \Sigma_{2,\alpha}
\end{cases}\\
\label{Pinfty2b}
 \wt{P}^{\infty}(\lb)=\begin{bmatrix}1&0\\0&1\end{bmatrix}+\mathcal{O}\le(\frac{1}{\lb}\ri),\quad \lb\to\infty \, .
\end{gather}
In order to get the solution of the above Riemann-Hilbert problem, let us introduce in analogy to Section~\ref{Sec_outer} the vector
 \begin{equation}
  \label{TildeSinfty_sol}
  \begin{split}
\wt{S}^{\infty}(\lb)=\gamma(\lb)\dfrac{\vartheta_3(0;2\tau)}{\vartheta_3\le( \frac{t \wt{\Omega}+\wt{\Delta}}{2\pi i};2\tau\ri)}
\begin{bmatrix} 
\displaystyle \frac{\vartheta_3 \le(2\wt{w}(\lambda) +\frac{t \wt{\Omega}+\wt{\Delta}}{2\pi i}-\frac{1}{2};2\tau\ri)}{\vartheta_3 \le(2\wt{w}(\lambda)  -\frac{1}{2};2\tau\ri)} & \displaystyle  \frac{\vartheta_3 \le(-2\wt{w}(\lambda) +\frac{t \wt{\Omega}+\wt{\Delta}}{2\pi i}-\frac{1}{2};2\tau\ri)}{\vartheta_3 \le(-2\wt{w}(\lambda)  -\frac{1}{2};2\tau\ri)}
\end{bmatrix}
\end{split} \ ,
\end{equation}
with $\wt{w}(\lambda)$ defined as
\begin{eqnarray}
\wt{w}(\lambda) = \int_{\eta_{2}}^{\lambda} \frac{ \Omega_\alpha}{R_{\alpha}(\lambda)} \frac{d \lambda}{4 \pi i} \,
\end{eqnarray}
where $\Omega_\alpha=- \frac{ \pi i \eta_2 }{K(m_{\alpha})}$.  Further from \eqref{wtomega} we have 
\[
\wt \Omega =2\pi i \eta_2 \frac{\alpha^2+\eta_2^2}{K(m_{\alpha})}+4\xi\Omega_\alpha
\] 
and
\[
 p_{\alpha}(\lambda)=\int_{\eta_2}^\lambda\dfrac{Q_1(\zeta)}{R_{\alpha}(\zeta)}\d\zeta,\quad \Omega_{\alpha}=-2p_{\alpha+}(\alpha)
 \]
where $Q_1$ has been defined in \eqref{P1P2}. We note that for $\lambda \in (-\alpha, \alpha)$, we have 
\begin{eqnarray}
p_{\alpha+}(\lb)  - p_{\alpha_-}(\lb)   = - \Omega_{\alpha} \ .
\end{eqnarray}

Then the  solution $\wt{P}^{\infty}(\lb)$ to the Riemann-Hilbert problem \eqref{Pinfty1b} and \eqref{Pinfty2b} is given explicitly by
\begin{equation}
\label{TildeP_infinity}
\wt{P}^{\infty}(\lb)=\frac{1}{2}\begin{bmatrix}
(1+\frac{p_{\alpha}(\lambda)}{\lambda})\wt{S}^{\infty}_1(\lambda)+\dfrac{1}{\lambda}\nabla_{\Omega_{\alpha}}\wt{S}^{\infty}_1(\lambda)
&(1-\frac{p_{\alpha}(\lambda)}{\lambda})\wt{S}^{\infty}_2(\lambda)+\dfrac{1}{\lambda}\nabla_{\Omega_{\alpha}}\wt{S}^{\infty}_2(\lambda)\\
(1-\frac{p_{\alpha}(\lambda)}{\lambda})\wt{S}^{\infty}_1(\lambda)-\dfrac{1}{\lambda}\nabla_{\Omega_{\alpha}}\wt{S}^{\infty}_1(\lambda)&
(1+\frac{p_{\alpha}(\lambda)}{\lambda})\wt{S}^{\infty}_2(\lambda)-\dfrac{1}{\lambda}\nabla_{\Omega_\alpha}\wt{S}^{\infty}_2(\lambda)
\end{bmatrix}\,,
\end{equation}
where $\wt{S}_{1}^{\infty}$ and $\wt{S}_{2}^{\infty}$ are the entries of the row vector $\wt{S}^\infty$ defined in (\ref{TildeSinfty_sol}),  and

 \[
\nabla_{\Omega_{\alpha}}\wt{S}^{\infty}_1(\lambda):=\gamma(\lb)\dfrac{\vartheta_3(0;2\tau)}{\vartheta_3 \le(2\wt{w}(\lambda)  -\frac{1}{2};2\tau\ri)} 
\dfrac{\Omega_{\alpha}}{2\pi i}\dfrac{\d}{\d z}\left[ \frac{\vartheta_3 \le(z+2\wt{w}(\lambda) +\frac{t \wt{\Omega}+\wt{\Delta}}{2\pi i}-\frac{1}{2};2\tau\ri)}{\vartheta_3\le(z+ \frac{t \wt{\Omega}+\wt{\Delta}}{2\pi i};2\tau\ri)}\right]\bigg|_{z=0} \ , 
\]

\[
\nabla_{\Omega_{\alpha}}\wt{S}^{\infty}_2(\lambda):=\gamma(\lb)\dfrac{\vartheta_3(0;2\tau)}{\vartheta_3 \le(-2\wt{w}(\lambda)  -\frac{1}{2};2\tau\ri)} 
\dfrac{\Omega_{\alpha}}{2\pi i}\dfrac{\d}{\d z}\left[ \frac{\vartheta_3 \le(z-2\wt{w}(\lambda) +\frac{t \wt{\Omega}+\wt{\Delta}}{2\pi i}-\frac{1}{2};2\tau\ri)}{\vartheta_3\le(z+ \frac{t \wt{\Omega}+\wt{\Delta}}{2\pi i};2\tau\ri)}\right]\bigg|_{z=0} \ .
\]
 The above construction has been obtained by modifying the construction of $P^{\infty}$ in (\ref{P_infinity}), in such a way that $\wt{P}^{\infty}(\lambda)$ solves the Riemann-Hilbert problem (\ref{Pinfty1b})-(\ref{Pinfty2b}), with $\mbox{det}\wt{P}^{\infty}(\lambda) = 1$, and $\wt{P}^{\infty}(-\lambda) = \pmtwo{0}{1}{1}{0} \wt{P}^{\infty}(\lambda)  \pmtwo{0}{1}{1}{0}$.

\subsection{The local parametrix $P^{\pm \alpha}$}
We will construct now a (matrix) local parametrix around the points $\lambda = \pm\alpha$. The construction of the local parametrices near $\lambda =\pm \eta_2$ is the same one as in  Section \ref{localparam}.

We focus again on a small but fixed neighbourhood $B^{(-\alpha)}_{\rho} = \le\{ \lambda \in \mathbb{C} \le| \,  \le|\lambda + \alpha\ri|< \rho \ri.  \ri\}$ of the endpoint $\lambda = -\alpha$. 
We define the conformal map
\begin{gather}
\zeta = \le(\frac{3}{4}\ri)^{\frac{2}{3}}  \le[ t\int_{-\alpha}^\lambda g'_+(s) - g'_-(s) \d s  \ri]^{\frac{2}{3}} =   \le[  18 t\int_{-\alpha}^\lambda \le(\frac{\sqrt{\alpha^2-s^2}}{\sqrt{\eta_2^2-s^2}}\ri)_+ \le(s^2 - \frac{\eta_2^2-\alpha^2}{2} - \frac{\xi  }{3}\ri) \d s  \ri]^{\frac{2}{3}} 
\end{gather}
locally in $B^{(-\alpha)}_\rho$. 

To define the local parametrix $P^{-\alpha}$ in $B^{(-\alpha)}_\rho$, we consider 
\begin{eqnarray*}
P(\lambda) = S (\lambda)e^{ \frac{\pi i }{4} \sigma_{3}}\left( \frac{\sqrt{ \pm \hat{r}(\lb)}}{f(\lb)} \right)^{\sigma_{3}} e^{ \mp \frac{1}{2} \left( \wt{\Omega} t + \wt{\Delta} \right)\sigma_{3}}\ , \ \lambda \in B^{(-\alpha)}_{\rho} \cap \mathbb{C}_{\pm},
\end{eqnarray*}
and then, using the inverse of the transformation $\zeta(\lambda)$, we define
\begin{gather*}
P^{(1)}(\zeta) = P(\lambda(\zeta)) e^{-\frac{2}{3} \zeta^{\frac{3}{2}} \sigma_3}
\ , \quad \zeta \in \C 
\end{gather*}
with branch cut $(-\infty,0]$. By construction, $P^{(1)}$ satisfies a Riemann--Hilbert  problem with jumps in a neighbourhood of $\zeta = 0$ as shown in \figurename \ \ref{conformalnew2}. 

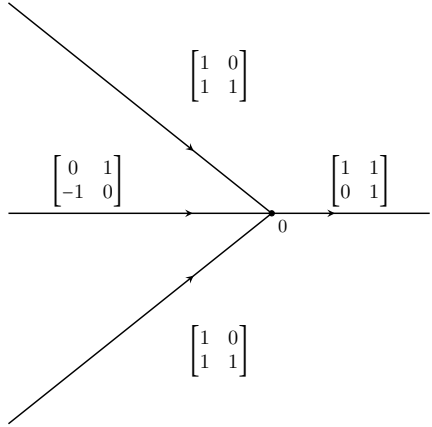
\begin{figure}[h!]
\centering

    \scalebox{.7}{
\begin{tikzpicture}[>=stealth]

\node [above] at (-1,2) {\large $ \begin{bmatrix}1& 0 \\ 1 & 1 \end{bmatrix}$};
\node [below] at (-1,-2) {\large $  \begin{bmatrix} 1 & 0 \\ 1 & 1\end{bmatrix}$};
\node [above] at (-3.5,0) {\large $  \begin{bmatrix}0 & 1 \\ -1 & 0 \end{bmatrix}$};
\node [above right] at (1,0) {\large $ \begin{bmatrix} 1 & 1  \\ 0& 1 \end{bmatrix}$};

\draw[->-=.7,thick] (-5,0)--(0,0);
\draw[->-=.7,thick] (-5,4)--(0,0);
\draw[->-=.7,thick] (-5,-4)--(0,0);
\draw[->-=.4,thick] (0,0)--(3,0);

\draw[fill] (0,0) circle [radius=0.05];
\node[below right] at (0,0) {$0$};

\end{tikzpicture}}
\caption{The contour setting under the conformal map $\zeta$ in a neighbourhood of $0$.}
\label{conformalnew2}
\end{figure}

We introduce the (local) Airy parametrix (see \cite{D99} and \cite{DKMcLVZ99}): let $\Psi_\Ai(\zeta)$ be the solution to the following Riemann--Hilbert  problem
\begin{enumerate}[(a)]
\item $\Psi_\Ai$ is analytic for $\zeta \in \mathbb{C} \backslash \Gamma_{\Psi}$, where the contours $\Gamma_\Psi$ are defined as $\Gamma_\pm = \le\{ \arg \zeta = \pm \frac{2\pi }{3}\ri\}$,  $\Gamma_{0, -} =\le\{ \arg \zeta = \pi \ri\}$ and $\Gamma_{0,+} =\le\{ \arg \zeta = 0 \ri\}$; 
\item $\Psi$ satisfies the following jump relations
\begin{gather}
\Psi_{\Ai \, +} (\zeta) = \Psi_{\Ai \, -}(\zeta) \begin{cases} \begin{bmatrix} 1 & 0 \\ 1 & 1 \end{bmatrix} & \text{on } \Gamma_+  \text{ and } \Gamma_-\\
\begin{bmatrix} 0 & 1 \\ -1\ & 0 \end{bmatrix} & \text{on }\Gamma_{0,-} \\
\begin{bmatrix} 1 & 1 \\ 0\ & 1 \end{bmatrix} & \text{on }\Gamma_{0,+}\, ; 
 \end{cases}\, 
\end{gather}
\item as $\zeta \rightarrow \infty$ 
\begin{gather}
\Psi_\Ai(\zeta) =   \zeta^{-\frac{1}{4}\sigma_3} \frac{1}{\sqrt{2}}\begin{bmatrix} 1& i \\ i&1 \end{bmatrix} \le(I + \mathcal{O}\le(\frac{1}{\zeta^{\frac{3}{2}}}\ri)\ri) e^{-\frac{2}{3}\zeta^{\frac{3}{2}}\sigma_3} \ ,
\end{gather}
\item $\Psi_\Ai$ remains bounded as $\zeta \rightarrow 0$, $\zeta \in \mathbb{C} \backslash \Gamma_{\Psi}$.
\end{enumerate}

The solution to this Riemann--Hilbert  problem is constructed with the help  of Airy functions. Setting $\omega = e^{\frac{2\pi i}{3}}$, we have
\begin{align}
&\Psi_\Ai(\zeta) = \sqrt{2\pi}\begin{bmatrix} \displaystyle \Ai(\zeta)&\displaystyle -\omega^2 \Ai (\omega^2\zeta) \\
\displaystyle -i\Ai'(\zeta) & \displaystyle i\omega \Ai'(\omega^2\zeta)
\end{bmatrix} & \text{for } 0< \arg \zeta < \frac{2\pi }{3} \\
&\Psi_\Ai(\zeta) = \sqrt{2\pi} \begin{bmatrix} \displaystyle -\omega \Ai(\omega\zeta)&\displaystyle -\omega^2 \Ai (\omega^2\zeta) \\
\displaystyle i\omega^2 \Ai'(\zeta) & \displaystyle i\omega \Ai'(\omega^2\zeta)
\end{bmatrix}  & \text{for }   \frac{2\pi }{3} < \arg \zeta < \pi \\
&\Psi_\Ai(\zeta) = \sqrt{2\pi} \begin{bmatrix} \displaystyle -\omega^2 \Ai (\omega^2\zeta) & \displaystyle \omega \Ai(\omega\zeta)\\
 \displaystyle i\omega \Ai'(\omega^2\zeta) & \displaystyle -i\omega^2 \Ai'(\zeta) 
\end{bmatrix}  & \text{for }  -\pi < \arg \zeta< - \frac{2\pi }{3} \\
&\Psi_\Ai(\zeta) = \sqrt{2\pi} \begin{bmatrix} \displaystyle \Ai (\zeta) & \displaystyle \omega \Ai(\omega\zeta)\\
 \displaystyle - i \Ai'(\zeta) & \displaystyle -i\omega^2 \Ai'(\zeta) 
\end{bmatrix}  & \text{for }  - \frac{2\pi }{3}< \arg \zeta< 0 \ ,
\end{align}
where $\Ai(\zeta)$ is the Airy function.

In conclusion, our local parametrix is then defined as
\begin{gather}
 P^{-\alpha}(\zeta(\lambda)) =  A(\lambda) \Psi_{\Ai}(\zeta(\lambda))  e^{\frac{2}{3}\zeta^{\frac{3}{2}} \sigma_3} e^{\pm \frac{1}{2} \left( \wt{\Omega} t + \wt{\Delta} \right)} \left( \frac{f(\lb)}{\sqrt{ \pm \hat{r}(\lb) } } \right)^{\sigma_{3}}e^{ - \frac{\pi i }{4} \sigma_{3}} \ , \ \lambda \in B^{(-\alpha)}_{\rho} \cap \mathbb{C}_{\pm} \, ,
 \end{gather}
where $A$ is an analytic prefactor whose expression is determined by imposing that 
\begin{gather}
P^{-\alpha} (\lambda)  \le(\wt P^{\infty}(\lambda) \ri)^{-1} = I + \mathcal{O}\le(t^{-1}\ri)  \qquad \text{as }  t  \to + \infty\, , \ \text{for } \lambda  \in \partial B^{(-\alpha)}_\rho \backslash \Gamma_\Psi \ .  
\end{gather}
In light of this asymptotic behaviour we set
\begin{gather}
A(\lambda) = \wt P^{\infty}(\lambda) e^{ \mp \frac{1}{2} \left( \wt{\Omega} t + \wt{\Delta} \right) \sigma_{3}} e^{ \frac{\pi i}{4} \sigma_{3}} \left( \frac{ \sqrt{ \pm \hat{r}(\lb) } }{f(\lb)} \right)^{\sigma_{3} } \frac{1}{\sqrt{2}} \begin{bmatrix}1& -i \\ - i  &  1\end{bmatrix}  \zeta(\lambda)^{\frac{1}{4}\sigma_3}  
  \ , \ \mbox{ for } \lambda \in B^{(-\alpha)}_{\rho} \cap \mathbb{C}_{\pm} \ .
\end{gather}
By construction, $A$ is defined and analytic in a neighbourhood of $-\alpha$, minus the cuts $(-\infty, -\alpha] \cup [-\alpha,+ \infty)$; moreover, $A$ is invertible ($\det A(\lambda) \equiv 1$). 

\begin{lemma}
$A(\lambda)$ is analytic in an open neighbourhood of $-\alpha$.
\end{lemma}

\begin{proof}
The proof entails verifying that $A$ has no jumps across the interval $(-\alpha -\rho, -\alpha + \rho)$ and that it has at most a removable singularity at $\lambda = -\alpha$.  We leave the verification that $A_{+}(\lb) = A_{-}(\lb)$ across the interval $(-\alpha -\rho, -\alpha + \rho)$ to the reader, using the jump relations satisfied by $\wt P^{\infty}$ and the above definitions.

The conformal map $\zeta(\lambda)$ has a simple zero at $\lambda = -\alpha$ (by construction), therefore $\zeta(\lambda)^{-\frac{1}{4}\sigma_3}$ has at most a fourth-root singularity at $-\alpha$. Similarly, $\wt P^{\infty}(\lambda) $ has a fourth-root singularity at $- \alpha$, as well; therefore, all the entries of $A(\lambda)$ have at most a square-root singularity at $\lambda = -\alpha$, and  $A(\lambda)$ is analytic in $B^{(-\alpha)}_\rho \backslash \{ -\alpha \}$. The point $\lambda = - \alpha$ is a removable singularity. This implies that $A(\lambda)$ is indeed analytic everywhere in $B^{(-\alpha)}_\rho \ $. 
\end{proof}

The construction of the local parametrix for $\lambda$ near $\alpha$ is obtained from the parametrix near $-\alpha$ as follows.  We fix a disk of {\it the same radius as the radius of the disk used for the parametrix near $-\alpha$}, and within that disk, define
\begin{eqnarray}
\label{eq:AlphaLocalDef2}
&&P^{\alpha} := \pmtwo{0}{1}{1}{0} P^{-\alpha}(-\lambda) \pmtwo{0}{1}{1}{0} \ .
\end{eqnarray}

{\bf Remark}. We note that, as with the analysis presented for $t=0$ and $x \to - \infty$, we may choose the contours so that they are preserved under the transformation $\lambda \mapsto - \lambda$.  Moreover, the construction of $\wt{P}^{\infty}$ continues to enjoy the symmetry relation
\begin{eqnarray}
\wt{P}^{\infty}(-\lambda)  = \pmtwo{0}{1}{1}{0} \wt{P}^{\infty}(\lambda) \pmtwo{0}{1}{1}{0}  \ .
\end{eqnarray}

\subsection{Small norm argument and determination of $u(x,t)$ as $t\to + \infty$}
\label{finalchapteralpha}

As before, we define a global (matrix) parametrix $P$ replacing each model $P^\xi$ in \eqref{global_p} with $\widetilde{P}^\xi$, $\xi = \infty, \pm\alpha,  \pm\eta_2$. Then we define the following ``remainder" Riemann--Hilbert  problem:
 \begin{gather}
 \mathcal E(\lambda) = S(\lambda) P(\lambda)^{-1}\ .
 \end{gather}
\edit{For some $c>0$,} the vector $\mathcal E$ satisfies
 \begin{gather}
\mathcal E_+(\lambda) = \begin{cases}
\mathcal E_-(\lambda) \le( I + \mathcal{O}\le( \displaystyle \edit{ e^{-ct} }\ri) \ri) \qquad \text{on the upper and lower lenses, outside the discs } \\
\mathcal E_-(\lambda) \le( I + \mathcal{O}\le(\displaystyle t^{-1} \ri) \ri) \qquad \text{on the circles around the endpoints} \\
 \end{cases}
 \end{gather}
 and
 \begin{gather} \mathcal E(\lambda) = \begin{bmatrix} 1 & 1 \end{bmatrix} + \mathcal{O}\le(\frac{1}{\lambda}\ri) \qquad \text{as }\lambda \rightarrow \infty \ .
 \end{gather}
Furthermore, as in Section \ref{smaxinfty}, the solution $\mathcal{E}(\lambda)$ is analytic in a neighborhood of $\lambda=0$.  Moreover, the jumps $V_{\mathcal{E}}$ for $\mathcal{E}$ satisfy the symmetry
\begin{eqnarray}
V_{\mathcal{E}}(-\lambda) = \pmtwo{0}{1}{1}{0} V_{\mathcal{E}}(\lambda)  \pmtwo{0}{1}{1}{0}  \ .
\end{eqnarray} 

Therefore, by a small norm argument (see \cite[Section 5.1.3]{smallnormRH}), we learn that there is a unique $\mathcal{E}$ solving the Riemann-Hilbert problem, and (as in Section \ref{smaxinfty}), the solution satisfies the symmetry relation
\begin{eqnarray}
\mathcal{E}(-\lambda) = \mathcal{E}(\lambda)  \pmtwo{0}{1}{1}{0} ,
\end{eqnarray}
and has a complete asymptotic expansion, satisfying
\begin{gather}
\mathcal E(\lambda) = \begin{bmatrix} 1 & 1 \end{bmatrix}  + \frac{\mathcal{E}_{1}( x, t)}{ \lambda t}\ +  \mathcal{O} \left( \frac{1}{\lambda^{2}} \right),
\end{gather}
where $\mathcal{E}_{1}(x,t)$ and its derivatives are bounded.


Unraveling the transformations, we can again get back to the potential. Our original Riemann--Hilbert  problem, for the unknown $Y$, satisfies
\begin{gather*}
 Y(\lambda) = T(\lambda) e^{-tg(\lambda) \sigma_3} f(\lb)^{-\sigma_3} = S(\lambda) e^{-tg(\lambda) \sigma_3} f(\lb)^{-\sigma_3}\\
  =\le( \begin{bmatrix} 1& 1\end{bmatrix} + \frac{\mathcal{E}_{1}( x, t)}{ \lambda t}\ +  \mathcal{O} \left( \frac{1}{\lambda^{2}} \right) \ri)P(\lambda) e^{-tg(\lambda) \sigma_3} f(\lb)^{-\sigma_3}.
\end{gather*}
In particular we are interested in the  vector $Y(\lb)$  for large $\lb$  
\begin{gather*}
 Y(\lambda)  =\le( \begin{bmatrix} 1& 1\end{bmatrix} + \frac{\mathcal{E}_{1}( x, t)}{ \lambda t}\ +  \mathcal{O} \left( \frac{1}{\lambda^{2}} \right) \ri)\widetilde{P}^{\infty}(\lambda) e^{-tg(\lambda) \sigma_3} f(\lb)^{-\sigma_3}\\
  =\le( \widetilde{S}^{\infty}(\lambda) + \frac{\mathcal{E}_{1}( x, t)}{ \lambda t}\ +  \mathcal{O} \left( \frac{1}{\lambda^{2}} \right) \ri)e^{-tg(\lambda) \sigma_3} f(\lb)^{-\sigma_3},
\end{gather*}
so that 
\begin{gather}
Y_1(\lambda) =  \le[ \wt{S}^{\infty}_{1}(\lambda)+ \frac{\mathcal{E}_{1}( x, t)}{ \lambda t}\ +  \mathcal{O} \left( \frac{1}{\lambda^{2}} \right)  \ri]e^{-tg(\lambda) }f(\lb)^{-1} ,  
\label{Y1}
\end{gather}
where $\wt{S}^{\infty}_1$ refers to the the first row vector solution $\wt S^{\infty}$ \eqref{TildeSinfty_sol}. Since 
\begin{gather}
\label{uexp1}
u(x,t) = 2\frac{\d}{\d x} \le[ \lim_{\lambda\rightarrow\infty} \lambda(Y_1(\lambda;x,t) -1) \ri] \ ,
\end{gather}
we have the following theorem.
\begin{theorem}
Given $\xi = \frac{x}{4t}$, in the region $\xi_{\rm crit}<\xi<\eta_2^2$ the solution of the KdV equation in the large time limit is 
\begin{equation}
\label{txt}
u(x,t)= \eta_2^2-\alpha^2   -2\eta_2^2\dfrac{E(m_\alpha)}{K(m_\alpha)} -2\dfrac{\partial^2}{\partial x^2} \log \vartheta_3 \left(\frac{\eta_2}{2K(m_\alpha)}(x-2(\alpha^2+\eta_2^2)t+\wt{\phi});2\tau_\alpha\right)+\mathcal{O}(t^{-1})
\end{equation}
where $E(m_\alpha)$ and $K(m_\alpha)$ are the complete elliptic integrals of first and second kind respectively, with modulus $m_\alpha = \frac{\alpha}{\eta_2}$;  $2\tau_\alpha=i\dfrac{K(m'_\alpha)}{K(m_\alpha)}$, with $m'_\alpha = \sqrt{1-m_\alpha^2}$,
\[
\wt{\phi}=\int_{\alpha}^{\eta_2}\dfrac{\log r(\zeta)}{R_{\alpha +}(\zeta)}\dfrac{\d\zeta}{\pi i}\in\R
\]
and the parameter $\alpha=\alpha(\xi)$ is determined from the equation
\[
\xi=\dfrac{\eta_2^2}{2}\left[1+m_{\alpha}^2+2\dfrac{m_{\alpha}^2(1-m_{\alpha}^2)}{1-m_{\alpha}^2-\frac{E(m_{\alpha})}{K(m_{\alpha})}}\right].
\]
The error term $\mathcal{O}(t^{-1})$ is uniform for $t$ sufficiently large.

Alternatively, 
\begin{equation}
\label{u_dnt}
u(x,t)=\eta_2^2-\alpha^2-2\eta_2^2\Jac^2\le( \eta_2(x-2(\alpha^2+\eta_2^2)t+\wt{\phi}) + K(m_\alpha) \le| \, m_\alpha\ri.  \ri)+ \mathcal{O}\le(t^{-1}\ri)
\end{equation}
where $\Jac\le( z\le| \, m \ri.\ri)$ is the Jacobi elliptic function. 
\end{theorem}
\begin{proof}
Starting from \eqref{Y1} we expand each term of   $Y_1(\lb)$  in a neighbourhood of infinity.
Regarding  $f(\lb)$  defined in (\ref{f_alpha}) we have 
\[
f(\lb)=1+\dfrac{f_1(\alpha,\eta_2)}{\lb}+\mathcal{O}\le(\frac{1}{\lambda^2}\ri)\, ,
\]
where 
\begin{equation*}
f_1(\alpha,\eta_2)=\left[\int_{\alpha}^{\eta_2}\dfrac{\zeta^2 \log r(\zeta)}{R_\alpha(\zeta)}\dfrac{\d\zeta}{\pi i }-\wt{\Delta}\int_{-\alpha}^{\alpha}\dfrac{\zeta^2}{R_\alpha(\zeta)}\dfrac{\d\zeta}{2\pi i }\right]\, .
\end{equation*}
 Regarding $e^{-tg(\lambda) }$ we are \edit{interested} in the $x$ derivative of this expression.  Using (\ref{gprime_der}) we have
\[
\dfrac{\partial}{\partial x}e^{-tg(\lambda) }=- \frac{1}{\lambda} \le[ \frac{\alpha^2+\eta_2^2}{2}  +\eta_2^2\left(\dfrac{E(m_\alpha)}{K(m_\alpha)}-1\right) \ri] + \mathcal{O}\le(\frac{1}{\lambda^2}\ri) \, .
\]
Regarding $\wt{S}^{\infty}_{1}(\lb)$, we have
\[
\wt{S}^{\infty}_{1}(\lb)=1+ \dfrac{1}{\lb}\left[\left( \log \vartheta_3 \le(\frac{t\wt\Omega+\wt\Delta}{2\pi i };2\tau\ri)\right)' - \frac{\vartheta_{3}'(0)}{\vartheta_{3}(0)}\right]\dfrac{\eta_2}{2K(m_\alpha)}+ \mathcal{O}\le(\frac{1}{\lambda^2}\ri)\, ,
\]
where $'$ stands for the derivative with respect to the argument  of the theta-function. By (\ref{Omega_der}) we have
\[
\dfrac{\partial}{\partial x}\wt{S}^{\infty}_{1}(\lb)= \dfrac{1}{\lb}\left[\log \vartheta_3 \le(\frac{t\wt\Omega+\wt\Delta}{2\pi i };2\tau\ri) \right]''\dfrac{\eta_2}{2K(m_\alpha)}\left(-\dfrac{\eta_2}{2K(m_\alpha)}+\dfrac{\partial}{\partial x}\dfrac{\wt\Delta}{2\pi i }\right)+ \mathcal{O}\le(\frac{1}{\lambda^2}\ri)\, ,
\]
where $''$ stands for   second derivative with respect to the argument of the theta-function.  Taking into account that  $\Delta=\Delta( \alpha(\xi),\eta_2)$ and therefore   the quantity $\frac{\partial}{\partial x} \Delta(\alpha(\xi),\eta_2)=\mathcal{O}(t^{-1})$ by (\ref{alpha}), we  
can write the above expression in the form
\[
\dfrac{\partial}{\partial x}\wt{S}^{\infty}_{1}(\lb)= -\dfrac{1}{\lb}\left[\dfrac{\partial}{\partial x^2}\log\vartheta_3 \le(\frac{t\wt\Omega+\wt\Delta}{2\pi i };2\tau\ri)+\mathcal{O}(\frac{1}{t})\right]+ \mathcal{O}\le(\frac{1}{\lambda^2}\ri)\, .
\]

Gathering the above expansions, using the fact that $\mathcal{E}_{1}(x,t)$ and its derivatives are bounded,  and using the explicit  expression of $\wt{\Omega}$ and $\wt{\Delta}$  in (\ref{wtomega}) and (\ref{TDelta})  respectively,  we obtain (\ref{txt}).
Also in this case, using the same calculations as in Theorem~\ref{thm:3.4}  we can reduce the expression of $u(x,t)$ to the form (\ref{u_dnt}).

The equivalence of formulas \eqref{txt} and  \eqref{u_dnt} is  slightly more delicate \edit{than} in the case of Theorem~\ref{thm:3.4}. 
It follows from \eqref{theta-jacobi} and  the relation \eqref{Omega_der},  that is a particular case of the more general relations
obtained in \cite{Krichever} in the context of Whitham modulation theory.
\end{proof}

 The equivalence of the formulas \eqref{txt} and  \eqref{u_dnt} is  a well \edit{known} fact in the theory of dispersive shock waves   for the KdV  equation
where modulated travelling waves are developed  (see e.g. the review \cite{GK}). Such equivalence is valid for any solution of \edit{Whitham's modulation} equations.
For  the  particular  solution of the Whitham modulation equations appearing in the long time asymptotic analysis of KdV this equivalence was obtained in \cite{Teschls}.

\section{Sub-critical case}
\label{sec:6}

As the parameter  $\xi   < \eta_2^2$ decreases, we proved that there is a critical value $\xi  _{\rm crit}$ (see Section \ref{sec:5}, equation (\ref{xicrit})) such that 
\begin{gather} \alpha(\xi  _{\rm crit}) = \eta_1 \ .\end{gather}

For $\xi < \xi_{\rm crit}$, we define 
\begin{align}
\label{gprimea}
g'(\lambda) = - 12\lambda^2 + 4\xi   + 12 \dfrac{Q_2(\lb)}{R(\lambda) }-4\xi  \dfrac{Q_1(\lb)}{R(\lambda) }
\end{align}
where $R$ is defined in (\ref{eq:Rdef}), specifically $R(\lambda) =\sqrt{(\lambda^2 - \eta_{1}^2) (\lambda^2-\eta^2_2)}$, and
\begin{equation}
Q_1(\lb)=\lb^2+c_1\, ,\quad Q_2(\lb)=\lb^4-\frac{1}{2}\lb^2(\eta_{1}^2+\eta_2^2)+c_2\, ,
\end{equation}
with the constants $c_1$ and $c_2$ chosen so that 
\begin{equation}
\int_{0}^{\eta_{1}}\dfrac{Q_2(\zeta)}{R_{+}(\zeta)}\d\zeta = 0 \ , \quad \int_{0}^{\eta_{1}}\dfrac{Q_1(\zeta)}{R_{+}(\zeta)}\d\zeta = 0 \ .
  \end{equation}
Integration yields
\begin{gather}
g(\lambda) = - 4 \lambda^{3} + 4\xi   \lambda + \int_{\eta_1}^{\lambda} \frac{12 Q_{2}(\zeta) - 4 \xi Q_{1}(\zeta)}{R(\zeta)} \d \zeta \ .
\end{gather}

By construction, $g$ satisfies the following constraints:
\begin{align}
&g_+(\lambda) + g_-(\lambda) +8\lambda^3 - 8\xi  \lambda = 0  & \lambda \in \Sigma_{1}\cup\Sigma_{2} \label{c1xneg}\\
&g_+(\lambda)-g_-(\lambda) = \overline{ \Omega} & \lambda \in [-\eta_1,\eta_1] \label{c2xneg}\\
&g(\lambda) =  \mathcal{O}\le(\frac{1}{\lambda}\ri) & \lambda \rightarrow  \infty \ .\label{c3xneg}
\end{align}
with
\begin{gather}
\ov \Omega = 2\pi i \eta_2 \frac{2\xi  -  (\eta_1^2 +\eta_2^2) }{K(m)} \in i\R \ .
\label{Omega_bar}
\end{gather}
\begin{remark}
The reader may verify that for $\xi = \xi_{\rm crit}$ the above function $g(\lb;\eta_1,\eta_2)$ in (\ref{gprimea}) agrees with the function $g(\lb;\alpha=\eta_1,\eta_2)$ in (\ref{g_alpha}).
\end{remark}

In order to show that the usual contour deformations can be carried out, as they were in Sections \ref{sec:3} and \ref{sec:5}, we need to verify that the quantity $\Re \le[2 g(\lb) + 8 \lambda^{3} - 8 \xi^{2} \lambda \ri]$ is positive on the contour $\mathcal{C}_{1}$, and negative on the contour $\mathcal{C}_{2}$, where these contours are as shown in \figurename \ \ref{openinglenses}.  

To accomplish this, we consider the quadratic polynomial
\begin{gather}
q(r; \xi  ) = 12\left( r^2-\frac{1}{2}r(\eta_{1}^2+\eta_2^2)+c_2\right) - 4 \xi (r + c_{1}) \ ,
\end{gather}
with $r\in[0,\eta_1^2]$. A quick inspection shows that $q(\eta_1^2; \xi  _{\rm crit}) = 0$ and $q(0; \xi_{\rm crit}) >0$, and moreover,
for all $\xi   \in \R$
\begin{gather}
\frac{\partial q}{\partial \xi  }(0; \xi  ) > 0 \quad \text{and} \quad  \frac{\partial q}{\partial \xi  }(\eta_1^2; \xi  ) <0\ ;
\end{gather}
therefore, $0= q(\eta_1^2; \xi_{\rm crit}) < q(\eta_1^2; \xi  )$ for all $\xi  <\xi  _{\rm crit}$.
So, for all $\xi < \xi_{\rm crit}$, there are two roots of $q(r; \xi  )$ within $(0, \eta_{1}^{2})$, and the polynomial is strictly positive on $[\eta_{1}^{2}, \eta_{2}^{2}]$.

This in turn implies, using arguments nearly identical to those used to prove Lemma \ref{g_signs}, that
\begin{eqnarray}
& & \Re \le[2 g(\lb) + 8 \lambda^{3} - 8 \xi^{2} \lambda \ri] > 0 \ \mbox{ for } \lambda \in \mathcal{C}_{1}\backslash\{\eta_1,\eta_2\} \, , \\
& & \Re \le[2 g(\lb) + 8 \lambda^{3} - 8 \xi^{2} \lambda \ri] < 0 \ \mbox{ for } \lambda \in \mathcal{C}_{2} \backslash\{-\eta_1,-\eta_2\}\ .
\end{eqnarray}

The use of this function, and the sequence of steps in the Riemann--Hilbert analysis which have been carried out for $t=0$ in Section \ref{sec:3}, may be applied directly to the present situation, and we use the same outer model $P^\infty(\lambda)$ (cf. \eqref{P_infinity}) as was used in Section \ref{sec:3}, with $x\Omega$ replaced by $t \ov{\Omega}$, with $\ov \Omega$ as defined by \eqref{Omega_bar}, along with the same local parametrices near each of the endpoints $\pm \eta_{1}, \pm \eta_{2}$.  Therefore we arrive at the following result.

\begin{thm}\label{thm6.3}
In the regime $t \to + \infty$, $\xi  < \xi  _{\rm crit}$, the potential $u(x,t)$ has the following asymptotic expansion
\begin{eqnarray}
u(x,t)=\eta_2^2-\eta_1^2-2\eta_2^2\Jac^2 \le( \eta_2(x-2(\eta_1^2+\eta_2^2)t+\phi) + K(m)\le| \, m\ri. \ri)+ \mathcal{O}\le(t^{-1}\ri) \ , 
\end{eqnarray}
where  $m = \eta_{1}/\eta_{2}$, and 
\begin{equation}
\phi = \int_{\eta_{1}}^{\eta_{2}} \frac{ \log{r( \zeta)}}{R_{+}(\zeta)} \frac{\d \zeta}{\pi i} \ .
\end{equation}
\end{thm}

\section{Conclusions}
In this paper we have considered the Riemann--Hilbert problem of \cite{ZakZakDya} in the case of one non-trivial reflection coefficient.  We have shown how this Riemann--Hilbert problem describes a soliton gas as the limit of a finite $N$-soliton configuration as $N$ tends to $+\infty$.  Then we established rigorous asymptotics of the KdV potential in several different regimes.  First, for the initial configuration, we studied the challenging behaviour as $x \to - \infty$, and obtained a universal asymptotic description in terms of the periodic travelling wave solution of KdV. Then, we provided a complete analysis of the long-time behavior of the solution of the KdV equation determined by the Riemann--Hilbert problem of \cite{ZakZakDya}.  For large $t$, there are three fundamental spatial domains, in which the solution $u(x,t)$ displays different asymptotic behaviours depending on the value of the parameter $\xi = x/(4t)$: for $\xi > \eta_2^2$ the solution decays exponentially, while for $\xi < \xi_{\text{cirt}}$ the solution is described by the periodic travelling wave solution of KdV with fixed parameters; between, for $\xi \in (\xi_{\text{crit}}, \eta_2^2)$ these two extreme asymptotic states are connected by a periodic travelling wave solution of KdV with slowly varying parameters.

Several challenges remain, like  the asymptotic analysis when there are two nontrivial reflection coefficients  or  the case where the spectral parameters of the soliton gas accumulates in disconnected components of the imaginary axis.
Beyond these, it is enticing to consider the interaction of one large soliton with this gas  like in \cite{CDE16} or the interaction between two such soliton gases.

\appendix
\section{Existence of solution to the soliton gas Riemann-Hilbert problem} \label{appendix}

We will provide a proof of existence and uniqueness for the Riemann-Hilbert problem
\begin{align}
\label{symYApp1}
 & Y(\lb) \text{ is analytic for } \lambda \in  \C \backslash \le\{ \Sigma_1 \cup \Sigma_2 \ri\} \nonumber \\
&Y_+ (\lambda) = Y_-(\lambda) \begin{cases}
\displaystyle \begin{bmatrix} 1 &0  \\ -i r(\lb;x,t)  & 1 \end{bmatrix} &\quad \lambda \in  \Sigma_1\\
\displaystyle \begin{bmatrix} 1 & i r(\lb;x,t)\\ 0  & 1 \end{bmatrix} & \quad \lambda \in  \Sigma_2
 \end{cases}\\
&Y(\lambda) = \begin{bmatrix}a&b \end{bmatrix} + \mathcal{O}\le(\frac{1}{\lambda}\ri) \qquad \lambda \rightarrow \infty\\
\label{symYApp}
&Y(-\lambda) = Y(\lb)\begin{bmatrix}0&\frac{b}{a}\\\frac{a}{b}&0 \end{bmatrix} \,,
\end{align}
with parameters $a>0$ and $b>0$, and with the contours shown in \figurename \ \ref{RHPY}, and $r(\lb;x,t)=r(\lambda) e^{8\lambda t \le( \lambda^2 - \frac{x}{4t} \ri)}$, and where we as usual seek a solution with at worst logarithmic singularities at the endpoints $\pm \eta_{j}$.

To establish that there is a solution to the Riemann-Hilbert problem (\ref{symYApp1})-(\ref{symYApp}), we seek $y_{1}$ with the following representation:
\begin{eqnarray}
\label{eq:Xform}
y_{1} = a + \frac{1}{2 \pi i } \int_{\eta_{1}}^{\eta_{2}} \frac{  \sqrt{r(s;x,t)} f(s)}{s - \lambda} \d s  .
\end{eqnarray}

This is consistent with the jump relations in (\ref{symYApp1}), in which the first entry is analytic across $(-\eta_{2}, -\eta_{1})$, with a jump across $(\eta_{1}, \eta_{2})$. Plugging (\ref{eq:Xform}) into the jump relation across $(\eta_{1}, \eta_{2})$, we find
\begin{eqnarray}
\label{eq:ExistIntOp}
f(\lambda) + \frac{b}{a} \frac{\sqrt{r( \lambda;x,t)}}{2\pi} \int_{\eta_{1}}^{\eta_{2}} \frac{\sqrt{r(s;x,t)}f(s)}{s + \lambda} \d s = -  i b \sqrt{r(\lambda;x,t)} \ .
\end{eqnarray}
The reader may verify that this integral equation appears (after some manipulation) in both entries of the jump relationships.

Now the integral operator appearing on the left hand side of (\ref{eq:ExistIntOp}) is compact (since it can obviously be approximated by a sequence of finite dimensional operators) and hence the index is zero.  It is also positive definite, which shows that this integral equation is uniquely invertible.

To see that the integral operator is positive definite, we follow the classic \cite[Formula 2.9]{KayMoses}, starting with the simple identity
\begin{eqnarray}
\frac{1}{s+\lambda} =\int_{-\infty}^{0} e^{ ( s + \lambda) z } \d z \ .
\end{eqnarray}  
We have
\begin{gather}
\int_{\eta_{1}}^{\eta_{2}}   \sqrt{r(\lambda;x,t)}  \ \overline{f(\lambda)} \int_{\eta_{1}}^{\eta_{2}} \frac{\sqrt{r(s;x,t)}f(s)}{s + \lambda} \d s \d \lambda =  \nonumber \\
\int_{-\infty}^{0}
\int_{\eta_{1}}^{\eta_{2}} \int_{\eta_{1}}^{\eta_{2}}  \sqrt{r(\lambda;x,t)}  \   \overline{f(\lambda)}\sqrt{r(s;x,t)}f(s) e^{(s + \lambda) z} \d s \d \lambda \d z \nonumber \\ 
= \int_{-\infty}^{0} \left|
\int_{\eta_{1}}^{\eta_{2}} \sqrt{r(s;x,t)} f(s) e^{ s z } \d s
\right|^{2} \d z \ > \ 0 \ ,
\end{gather}
provided $f$ is not identically equal to $0$.

Regarding uniqueness, if $\tilde{Y}$ is a solution to the Riemann-Hilbert problem (\ref{symYApp1})-(\ref{symYApp}), then setting $\tilde{f}(s) = - i \sqrt{r(s;x,t)} \tilde{Y}_{2}(s)$ for $s \in (\eta_{1}, \eta_{2})$, one verifies that $\tilde{f}$ must satisfy the integral equation (\ref{eq:ExistIntOp}), which obviously possesses a unique solution.

Returning to the Riemann-Hilbert problem (\ref{symYApp1})-(\ref{symYApp}), if we take $(a,b)=(1,1)$ we have established the existence and uniqueness of the solution to the soliton gas Riemann-Hilbert problem.  But more importantly, if we separately consider $(a,b) = (1,2)$, we find a second independent solution of the Riemann-Hilbert problem, which combined yields a $2 \times 2$ matrix solution to the following Riemann-Hilbert problem:
\begin{align}
 & {\bf Y}(\lb) \text{ is analytic for } \lambda \in  \C \backslash \le\{ \Sigma_1 \cup \Sigma_2 \ri\} \nonumber \\
&{\bf Y}_+ (\lambda) = {\bf Y}_-(\lambda) \begin{cases}
\displaystyle \begin{bmatrix} 1 &0  \\ -i r(\lb;x,t)  & 1 \end{bmatrix} &\quad \lambda \in  \Sigma_1\\
\displaystyle \begin{bmatrix} 1 & i r(\lb;x,t)\\ 0  & 1 \end{bmatrix} & \quad \lambda \in  \Sigma_2
 \end{cases}\\
&{\bf Y}(\lambda) = \begin{bmatrix}1&1 \\
1 & 2 \end{bmatrix} + \mathcal{O}\le(\frac{1}{\lambda}\ri) \qquad \lambda \rightarrow \infty \ .
\end{align}
This solution is invertible for all $\lambda \in \mathbb{C}$ since $\mbox{det} {\bf Y} \equiv 1$.

\vspace{3mm}

\noindent{\bf Acknowledgements.}
T.G. and M.G. acknowledges  the support of the H2020-MSCA-RISE-2017 PROJECT No. 778010 IPADEGAN. K.M. was supported in part by the National Science Foundation under grant DMS-1733967.  Part of the work of M.G. and K.M. was done during their visits at SISSA and while T.G., K.M. and R.J. were visiting CIRM, Luminy, France.  We acknowledge SISSA and CIRM for excellent working conditions and generous support.
We wish to thank  Marco Bertola, and Alexander Minakov  for useful feedback in constructing the matrix Riemann-Hilbert problem  outer parametrix   of KdV. In particular  in  Section~\ref{Outer} we have  implemented  the suggestions by 
Alexander Minakov.

\bibliographystyle{alpha}
\bibliography{DZZSolitonGas_Revision6}

\end{document}